\documentclass[12pt]{article}
\usepackage{geometry}
\usepackage{refcount}
\usepackage{amssymb}
\usepackage{amsmath,bm,bbm,amsthm}
\usepackage{mathrsfs}
\usepackage{amsfonts}
\usepackage{algorithm}
\usepackage{algpseudocode}
\usepackage{graphicx,epsfig,lscape}
\usepackage{color}
\usepackage{setspace}
\usepackage{url}
\usepackage{upgreek}
\usepackage{commath}
\usepackage{enumitem}
\usepackage{indentfirst}
\usepackage{multirow}
\usepackage{xr}
\usepackage{booktabs}
\usepackage{stackengine}
\usepackage[longnamesfirst,round]{natbib}
\usepackage{dsfont}
\usepackage{mathtools}
\usepackage[title]{appendix}
\usepackage{accents}
\usepackage{rotating}
\usepackage{titling}
\usepackage{subfig}
\usepackage{graphicx}
\usepackage{hyperref}
\usepackage{cleveref}
\usepackage{tikz}

\usetikzlibrary{shapes.geometric}
\usetikzlibrary{positioning}
\usetikzlibrary{arrows,arrows.meta}
\allowdisplaybreaks 
\newtheorem{theorem}{Theorem}

\newtheorem{assumption}{Assumption}

\newtheorem{definition}{Definition}

\newtheorem{lemma}{Lemma}

\newtheorem{proposition}{Proposition}
\newtheorem*{proposition*}{Proposition}
\newtheorem*{lemma*}{Lemma}
\newtheorem*{theorem*}{Theorem}
\newtheorem{remark}{Remark}

\newcommand{\N}{\mathbb{N}}

\let\emptyset\varnothing

\newcommand{\tpi}{\widetilde \pi}

\newtheorem{lemmaprime}{Lemma}

\begin{document}

\title{\vspace{-2em} On Rent Dissipation in Dynamic Multi-battle Contests%
\thanks{%
We thank Attila Ambrus, Rahul Deb, Hanming Fang, Navin Kartik, Xianwen Shi,
and seminar participants at PKU HSBC and HUST for helpful discussions,
suggestions, and comments. Deng thanks the National Natural Science
Foundation of China (No.\thinspace 72403270) and the Start-up Research Grant
(SRG2023-00036-FSS) of the University of Macau for financial support. Fu
thanks the Singapore Ministry of Education Tier-1 Academic Research Fund
(R-313-000-139-115) for financial support. Wu thanks the National Natural
Science Foundation of China (Nos.\thinspace 72222002 and 72173002), the Wu
Jiapei Foundation of the China Information Economics Society (No.\thinspace
E21100383), and the Research Seed Fund of the School of Economics, Peking
University, for financial support. Any remaining errors are our own.}}
\author{Shanglyu Deng\thanks{
Department of Economics, University of Macau, E21B Avenida da Universidade,
Taipa, Macau, China, 999078. Email: \href{mailto:sdeng@um.edu.mo}{%
sdeng@um.edu.mo}.} \and Qiang Fu\thanks{
Department of Strategy and Policy, National University of Singapore, 15 Kent
Ridge Drive, Singapore, 119245. Email: \href{mailto:bizfq@nus.edu.sg}{%
bizfq@nus.edu.sg}.} \and Junchi Li\thanks{
Graduate School, Duke University, North Carolina, USA, 27705. } \and Zenan
Wu\thanks{
School of Economics, Sustainability Research Institute, Peking University,
Beijing, China, 100871. Email: \href{mailto:zenan@pku.edu.cn}{%
zenan@pku.edu.cn}.}}
\date{\today}
\maketitle

% \vspace{-1.5em}
\begin{abstract}
\begin{singlespace}
\noindent 
We study dynamic multi-battle contests and examine how the contest structure shapes dynamic incentives and determines the extent of rent dissipation. A discouragement effect often arises---such as in tug-of-war and best-of-$K$ contests---preventing full rent dissipation even when the series can extend infinitely. We identify a structural property, exchangeability, that contributes to the effect. Leveraging this insight, we establish a necessary and sufficient condition for almost-full rent dissipation. As an application, we introduce the iterated incumbency contest, which illustrates how volatility in the surrounding environment sustains dynamic incentives and generates almost-full rent dissipation, and thus offers insights into various competitive phenomena.
\end{singlespace}

\begin{description}
\item[Keywords:] dynamic competition, contests, rent dissipation

\item[JEL Classification Codes: D72, C73, D86.] 
\end{description}
\end{abstract}

% \bigskip 

% % {\large [Preliminary. Comments are welcome.]}

\medskip

% The key to achieving full rent extraction is to avoid a Matthew effect.

\thispagestyle{empty}

\clearpage

\pagenumbering{arabic}

\section{Introduction}

Competitions often unfold over multiple phases rather than being decided by
a single, decisive action. Contenders confront each other repeatedly in a
sequence of battles, sinking costly effort over time. Final victory requires
the accumulation of sufficiently many intermediate successes, rather than a
single stroke of effort. Such dynamics are pervasive in the socioeconomic
landscape.

Military warfare provides a close analogy: Outcomes on individual
battlefields may shift momentum or local control, yet rarely determine
ultimate victory on their own, as illustrated by the ongoing conflict
between Russia and Ukraine. Similar dynamics arise in competition for
technological standards, such as the rivalry between Sony and JVC in home
video. To prevail in such an enduring race, a firm must outperform its
rivals across a sufficient number of component technologies and secure a
critical mass of ecosystem partners, including suppliers and manufacturers
of complementary goods; no single technological win is decisive. U.S.
presidential primaries offer another intuitive example. Campaigns are
conducted sequentially across states, and party nomination is conferred only
after a candidate has accumulated enough delegate victories over time.
High-profile patent litigations exhibit similar multi-battle dynamics. Final
outcomes emerge only after numerous rulings across jurisdictions and
potentially repeated appeals within individual lawsuits, as exemplified by
the canonical litigation between Apple and Samsung from 2011 to 2018, which
spanned multiple countries and legal venues.

A substantial body of research studies the strategic foundations of dynamic
multi-battle contests (see, e.g., \citealp{harris1987racing,%
klumpp2006primaries,konrad2009multi,fuLuPan2015team}), with a sustained
emphasis on contenders' dynamic incentives and on how the structure of the
contest shapes effort provision. This literature has long highlighted the
discouragement effect. Consider a dynamic contest between two ex ante
symmetric contestants. An accidental early win by one contestant grants the
lucky frontrunner an advantage and\ could turn an initially even race into a
lopsided competition. This ex post asymmetry discourages effort provision by
both contestants: The laggard is deterred by bleak winning prospects, while
the reduced resistance affords the frontrunner an easy win (see, e.g.,
\citealp{klumpp2006primaries}). This limits rent dissipation---total effort
relative to the prize awarded to the winner---and generates inefficiency in
terms of effort supply.

In this paper, we examine rent dissipation in a general dynamic contest that
accommodates a wide range of contest architectures or rules. More
specifically, we investigate whether and under what conditions rent may or
may not fully dissipate, particularly when the contest may involve a long
series of encounters. The answer is not obvious a priori. Consider, for
instance, a tug-of-war contest, which is often used to model protracted
litigation and competition for political control. Final victory requires one
contestant to establish a sufficiently large lead, while any existing lead
is immediately offset by an opponent's win, a feature that can in principle
prolong the contest indefinitely. Two opposing forces are at play. On the
one hand, a potentially infinite horizon attenuates incentives to exert
effort for immediate advancement; on the other hand, this allows effort to
accumulate without bound. The latter force is strengthened when the margin
required for final victory under the contest rule increases. We provide a
general analysis of equilibrium rent dissipation in dynamic contests and
identify conditions for full or partial rent dissipation, thereby revealing
the fundamental forces shaping contestants' dynamic incentives.

We consider a generic sequential multi-battle contest model. Two players
compete head-to-head in a sequence of component battles that take place
successively.\footnote{%
We focus on the two-player case solely for ease of exposition. Our main
results extend straightforwardly to the multi-player case, which we discuss
in the Conclusion.} In each battle, the players invest effort, which is
converted into their respective winning probabilities through a \emph{%
success function}. A contest rule specifies the overall winner based on the
history of component battles, and a prize is awarded accordingly. We impose
no restriction on the specific form of the contest rule, allowing the model
to accommodate a broad array of contest mechanisms. Popular examples include
the best-of-$(2K+1)$ contest, with $K\in \mathbb{N}_{++}$ , in which a
player secures overall victory by winning at least $K+1$ battles. The
best-of-$(2K+1)$ format can also be viewed as a first-past-the-post (FPTP)
race, which awards the prize to the contestant who reaches a predetermined
target number of battle wins before the opponent does. Another example is a
tug-of-war contest with margin $K$, in which a player prevails if and only
if they win $K\in \mathbb{N}_{++}$ \emph{more} battles than his rival.

A contest is \emph{nontrivial }if it has a finite \emph{minimum length},
that is, the minimum number of battles required for the contest to decide
its winner is finite. For instance, a tug-of-war contest with margin $K$ has
a minimum length of exactly $K$: Although the contest may last infinitely
long, there exists a finite history that leads to its conclusion. By
contrast, a contest in which no finite history ever determines a winner
(i.e., every history is infinite) is ill-defined and is excluded from our
analysis. We establish that rent cannot fully dissipate in any nontrivial
contest %%, regardless of how long the series can potentially be
%%extended 
(\Cref{thm:L}), even though the contest may admit infinite histories.
Specifically, the equilibrium total effort is bounded from above by an upper
limit strictly below the prize value.

We then proceed to explore the limiting properties of the contest as its
minimum length becomes very large. Our analysis shows that, in the case of a
tug-of-war contest, full rent dissipation does not arise even in the
limit---that is, when the required margin $K$ tends to infinity. As the lead
possessed by the frontrunner continues to grow, the standard discouragement
effect comes into play: Even if a large number of additional battles remain
from the finish line, the laggard becomes fully discouraged, and the
probability of the frontrunner's winning future battles converges to one (%
\Cref{prop:tow-adv}). Owing to this discouragement effect, the equilibrium
rent dissipation ratio---defined by the ratio of equilibrium total effort to
prize value---remains bounded by a constant that is strictly less than one
and uniform to all $K\in \mathbb{N}_{++}$ (\Cref{thm:tow-rent}).

This striking observation compels us to delve deeper into the nature of
contestants' dynamic incentives. We show that a structural feature common in
many forms of dynamic contests---\emph{exchangeability}---plays a critical
role in generating the discouragement effect. An exchangeable contest does
not distinguish between two histories of battle outcomes if they only differ
in the order of past wins and losses. For example, in a best-of-three
contest (such as a tennis match), the same deciding set arises regardless of
whether a player wins the first game and loses the second or vice versa.
Similarly, in a tug-of-war contest, only the net number of battles won by
each player matters. We demonstrate that this feature induces a form of
``long memory'' in a dynamic contest, whereby early battle outcome exerts a
persistent influence on future play. Consequently, a frontrunner can
leverage early successes to accumulate an advantage without bound, leading
to excessive discouragement over time. We establish that an exchangeable
contest can never fully dissipate its rent regardless of its minimum length (%
\Cref{prop:homogeneous-ex}). For instance, in a best-of-$(2K+1)$ contests,
total equilibrium effort is bounded away from its prize value even as $K$
approaches infinity.

To illustrate the notion and role of exchangeability, we construct an
intuitive counterexample: the $K$-consecutive-win contest, which awards the
prize to the first contestant who achieves $K$ consecutive battle wins. This
contest resembles deuce in tennis, under which a player must win two points
in a row to close the game. This setting nullifies exchangeability because
the order of battle outcomes matters. A player gains an advantage by winning
a point immediately after having lost one, whereas winning a point but
losing the subsequent one leaves the player on the brink of losing the game.
We show that, in contrast to exchangeable contests, the $K$-consecutive-win
contest can almost fully dissipate rent in the limit, such that the
equilibrium total effort can be arbitrarily close to prize value when $K$
becomes large (\Cref{thm:cw-rent}).

These findings pave the way for a more general characterization of
almost-full rent dissipation in dynamic contests. We propose the \emph{%
transient dominance property} (\Cref{def:si}) and show that it is necessary
and sufficient for almost-full rent dissipation (\Cref{thm:sufficient}).
More specifically, the property imposes two requirements. First, the
frontrunner remains sufficiently motivated to exert effort toward final
victory. Second, the competition remains persistently \textquotedblleft
fluid,\textquotedblright\ in the sense that one's leadership can always be
reversed with sufficiently high likelihood, so that no player's advantage
can accumulate without a bound. Together, these conditions induce a form of
\textquotedblleft short memory\textquotedblright\ and keep the contest
competitive as long as final victory is not yet formally awarded.

This transient dominance property can be readily illustrated by a tug-of-war
contest with a random reset, in which the contest may revert to its initial
state with a nonzero probability after each battle. More importantly, we
propose an \textit{iterated incumbency contest}, which provides an intuitive
framework and novel insights into a wide range of dynamic competitions in
uncertain environments, such as competitions over emerging technologies\ and
evolutionary competition between species. Our analysis verifies that the
transient dominance property arises in such contests, leading to almost-full
rent dissipation in the limit (\Cref{prop:full-suff}). The result provides a
natural rationale for the pervasively observed Red Queen effect in
technological or biological evolutionary processes. More details are provided
in \Cref{sec:full-rent-condition}.

\paragraph{Related literature}

\citet{harris1987racing} lay the groundwork for
research on dynamic contests with successive, disjoint component battles.
They propose two modelling approaches that differ in their prevailing
winning rules. The first is a race model, in which a player prevails if he
is the first to win a prespecified number of battles. The best-of-$(2K+1)$
contest is its most intuitive variant, whereby a player must win a majority of
component battles for the final prize. \citet{klumpp2006primaries},
\citet{konrad2009multi,konrad2010stochastic}, \citet{gelder2014custer}, and
\citet{klumppKonradSolomon2019blotto} extend this research stream. The
second is a tug-of-war model, in which the ultimate winner must accumulate a
sufficiently large lead over his opponent along the path.
\citet{mcafee2000continuing}, \citet{konradKovenock2005tug}, and
\citet{agastyaMcAfee2006continuing} further develop contests with this
feature.

This literature has long recognized the discouragement effect that arises in
dynamic multi-battle contests, whereby early outcomes stifle future
competition, disincentivize effort, and limit rent dissipation.
\citet{klumpp2006primaries} saliently interpret the New Hampshire effect in
U.S. presidential primaries as a manifestation of this discouragement effect.
A number of studies propose environmental or structural features that may
mitigate the discouragement effect and sustain effort supply over time, such
as intermediate prizes in \citet{konrad2009multi}\footnote{%
\citet*{fuKeTan2015success} allow a player to derive a utility from winning each
component battles. The \textquotedblleft utility of
winning\textquotedblright\ plays a similar role to the intermediate prizes
in
\citet{konrad2009multi}.} and uncertainty in effort cost functions in
\citet{konrad2010stochastic}. \citet{gelder2014custer} shows that a player
may drastically outperform the frontrunner when falling behind if the loser
faces a severe penalty.

However, these studies are all situated within race models with a finite
number of component battles. In contrast, our analysis accommodates a
general setup without imposing specific winning rules and sheds light on the
asymptotic properties of equilibrium rent dissipation when length of a
contest approaches infinity. In this general setting, our analysis reveals
the fundamental role of exchangeability and the long-memory it enables in
driving the discouragement effect and underdissipation of rents; this, in
turn, motivates conditions of transient dominance that yield almost-full
rent dissipation.

The notion of exchangeability was first proposed by \citet{ewerhart2019multi}.
Allowing for infinite horizon and without imposing specific winning rules, %
\citet{ewerhart2019multi} study dynamic multi-battle contests that satisfy
certain properties, i.e., exchangeability, monotonicity, and centeredness.
They show that such contests can be described by state machines (or
automata), and that these contests eventually enter a tie-breaking phase
isomorphic to a tug-of-war game. Our paper differs in that we focus on
dynamic rent dissipation, whereas their main contribution lies in
establishing the connection between contests with infinite horizons and
tug-of-war games.\footnote{%
They also provide a comprehensive analysis of the tug-of-war game, extending
the Tullock success function case considered in \citet*{%
karagozouglu2021perseverance} to more general success functions.}

All of these studies examine contentions between individual players. In
contrast, \citet*{fuLuPan2015team} study multi-battle races between teams,
with members from rival groups matched in head-to-head competitions.
\citet{hafnerKonrad2016team} and \citet{hafner2017tug} study tug-of-war
contests between teams with a structure similar to that in
\citet{fuLuPan2015team}.

The economics literature has considered a wide array of dynamic contests in
alternative forms. All aforementioned studies assume that players
participate in disjoint battles, each of which yields a winner. In contrast,
another stream of the literature allows each player's effort to accumulates
over time, with cumulative output determining the final winner at the end of
the contest. Notable examples include
\citet{meyer1992biased,yildirim2005rounds,gershkovPerry2009midterm,%
aoyagi2010feedback,ederer2010feedback,gurtlerHarbring2010feedback,%
goltsmanMukherjee2011interim}.

\section{Model}

Two risk-neutral players $A$ and $B$ compete head-to-head for a prize of
common value $v>0$.\footnote{%
We do not normalize the value to 1 because our setting does not require the
winning probability of each battle (battle success function) to be
homogeneous. Prize value does affect equilibrium outcomes without
homogeneity.} The contest consists of a sequence of successive \emph{%
component battles}, and the final winner is determined by the history of
battle outcomes.

\subsection{Component Battles}
\label{sec:component-battles}

In each component battle, player $\ell \in \{A,B\}$ simultaneously exerts an
effort $x_{\ell }\in \mathbb{R}_{+}:=[0,+\infty )$. Effort is measured
directly in units of disutility, so exerting effort $x_{\ell }$ entails a
cost of $x_{\ell }$. For a given effort profile $(x_{A},x_{B})$, player $A$
wins the current battle with a probability $p(x_{A},x_{B}):\mathbb{R}%
_{+}\times \mathbb{R}_{+}\rightarrow \lbrack 0,1]$, and player $B$ wins with
complementary probability $p(x_{B},x_{A})=1-p(x_{A},x_{B})$. We call $%
p(x,x^{\prime })$ the $\emph{success}$ $\emph{function}$ (SF) for component
battles. With an effort profile $(x_{A},x_{B})$, a player $\ell \in \{A,B\}$
receives an expected payoff $p(x_{\ell },x_{-\ell })\Delta _{\ell }-x_{\ell
} $ from the battle, where $\Delta _{\ell }$ denotes the value he can
generate from winning this battle, i.e., the differential in terms of his
payoff in the dynamic contest between winning and losing this battle.

For expositional efficiency, the baseline analysis focuses on a homogeneous
(of degree zero) success function. However, our analysis applies in a broad
context and the results remain largely intact when allowing for alternative
contest technologies, which we will elaborate on in
\Cref{sec:key-component-battle-properties}. More
specifically, we assume the following.

\begin{assumption}[Homogeneous Success Function]
\label{cond:bsf} For a given effort profile $(x,x^{\prime })$, a player, by
exerting an effort $x$, wins the battle with a probability%
\begin{equation*}
p(x^{\prime },x)=\gamma \!\left( \frac{x^{\prime }}{x}\right) ,
\end{equation*}
where $\gamma :[0,+\infty ]\rightarrow \lbrack 0,1]$ is a continuous and
twice-differentiable function with $\gamma (0)=0$, $\gamma (+\infty )=1$, $%
\gamma ^{\prime }>0$, $\gamma ^{\prime \prime }\leq 0$, and $\gamma
(x)+\gamma (1/x)=1$.
\end{assumption}

The popularly adopted Tullock success function provides a classic example
within this family of models, whereby%
\begin{equation*}
p(x,x^{\prime })=%
\begin{cases}
\frac{x^{r}}{x^{r}+{(x^{\prime })}^r}, & \text{ if }(x,x^{\prime })\neq
(0,0), \\ 
\frac{1}{2}, & \text{ if }(x,x^{\prime })=(0,0),%
\end{cases}%
\end{equation*}%
with $r\in (0,1]$. A serial SF \citep{alcalde2007tullock} also satisfies the
homogeneous-of-degree-zero requirement: For $\alpha \in (0,1)$, a player
wins with probability $p(x^{\prime },x)=1-\left[ (x/x^{\prime })^{\alpha }%
\right] /2$ if $x^{\prime }\geq x$ and $p(x^{\prime },x)=\left[ (x^{\prime
}/x)^{\alpha }\right] /2$ otherwise.

We define 
\begin{equation*}
\phi (\theta ):=\gamma (\theta )-\theta \gamma ^{\prime }(\theta )\text{ for 
}\theta \geq 0.
\end{equation*}%
Let $\Delta _{\ell }$, $\ell \in \{A,B\}$, denote a player $\ell $'s
valuation of winning a battle. The equilibrium in a single component battle
is characterized as follows.

\begin{lemma}[\citealp{malueg2005equilibria}]
\label{lemma:single} Under \Cref{cond:bsf}, in a single battle with winning
values $\Delta _{A},\Delta _{B}>0$, there exists a unique pure-strategy Nash
equilibrium. Equilibrium effort levels are 
\begin{equation*}
x_{A}^{\ast }=\Delta _{B}\gamma ^{\prime }\!\left( \frac{\Delta _{B}}{\Delta
_{A}}\right) ,\qquad x_{B}^{\ast }=\Delta _{A}\gamma ^{\prime }\!\left( 
\frac{\Delta _{A}}{\Delta _{B}}\right) ,
\end{equation*}%
and players' expected equilibrium utilities are 
\begin{equation*}
\Pi _{A}^{\ast }=\Delta _{A}\phi \!\left( \frac{\Delta _{A}}{\Delta _{B}}%
\right) ,\qquad \Pi _{B}^{\ast }=\Delta _{B}\phi \!\left( \frac{\Delta _{B}}{%
\Delta _{A}}\right) .
\end{equation*}
\end{lemma}

We define $\pi _{\ell }:=\Pi _{\ell }^{\ast }/\Delta _{\ell }$ to be the 
\emph{gain function} of player $\ell $ in a battle, which is given by the
ratio of his equilibrium expected utility in the battle to his winning
value. By \Cref{lemma:single}, the gain function boils down to $\phi \!\left( \Delta
_{\ell }/\Delta _{-\ell }\right) $.

\subsection{Contest Architecture}

We assume that, at the beginning of each battle, the outcomes of all
previous battles are commonly known. An \emph{outcome path }is denoted by $%
\ell ^{t}:=(\ell _{s})_{s=1}^{t}$ for $t>1$, where each element $\ell
_{s}\in \{A,B\}$ indicates the winner of battle $s$. Let $\widetilde{H}%
:=\{\ell ^{t}:t\in \mathbb{N}_{++}\bigcup \{+\infty \}\}$ be the set of all
possible outcome paths. Further, let $H^{\dagger }\in \widetilde{H}$ denote
the set of \emph{terminal histories}: Each terminal history determines the
ultimate winner of the contest according to the contest rule.

For a finite terminal history of length $t$, the winner's net payoff is $%
v-\sum_{s=1}^{t}x_{\text{winner},s}$, and the loser's payoff is $%
-\sum_{s=1}^{t}x_{\text{loser},s}$. For an infinite terminal history, we
assign a payoff $\frac{v}{2}-\sum_{s=1}^{+\infty }x_{\ell ,s}$ to each
player $\ell \in \{A,B\}$. Given a path $\ell ^{t}=(\ell _{1},\ldots ,\ell
_{t})$, let $\ell ^{t^{\prime }|t}:=(\ell _{1},\ldots ,\ell _{t^{\prime }})$
denote its prefix of length $t^{\prime }\leq t$. We impose a no-redundancy
condition: No terminal history in $H^{\dagger }$ is a proper prefix of
another. We can then define the \emph{history set} of the contest as $%
H:=\{\ell ^{t^{\prime }|t}:\ell ^{t}\in H^{\dagger }\}$. Each element of $H$
is called a \emph{history} of the contest: That is, a history is an outcome
path that does not extend beyond any terminal history.

We construct \Cref{fig:dyn-ex} to illustrate the terminologies:
\Cref{fig:dyn-ex}(a) depicts a best-of-three contest, while
\Cref{fig:dyn-ex}(b) presents a
tug-of-war with margin two, i.e., a tug-of-war in which a player must lead
by two wins to secure final victory. Each branch represents a possible
battle outcome, and an upward (resp., downward) branch corresponds to a
battle won by player $A$ (resp., player $B$). The set of terminal histories
for the best-of-three contest is $\{AA,BAA,ABB,BB\}$, and terminal histories
of the tug-of-war with margin 2 include $AA$, $BB$, $ABBB$, $BAAA$, and so
on.

\begin{figure}[ht!]
\centering
\subfloat[Best-of-3 Contest.]{
        \includegraphics[width = 0.38\textwidth]{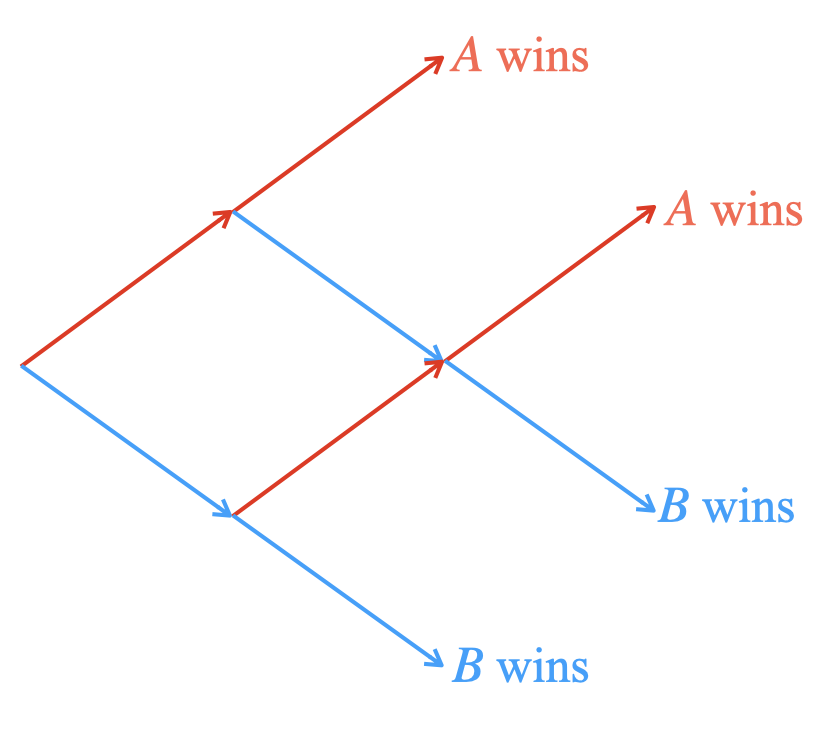}
        \label{fig:dyn-bo3}
    } 
\subfloat[Tug-of-war with Margin 2.]{
        \includegraphics[width = 0.5\textwidth]{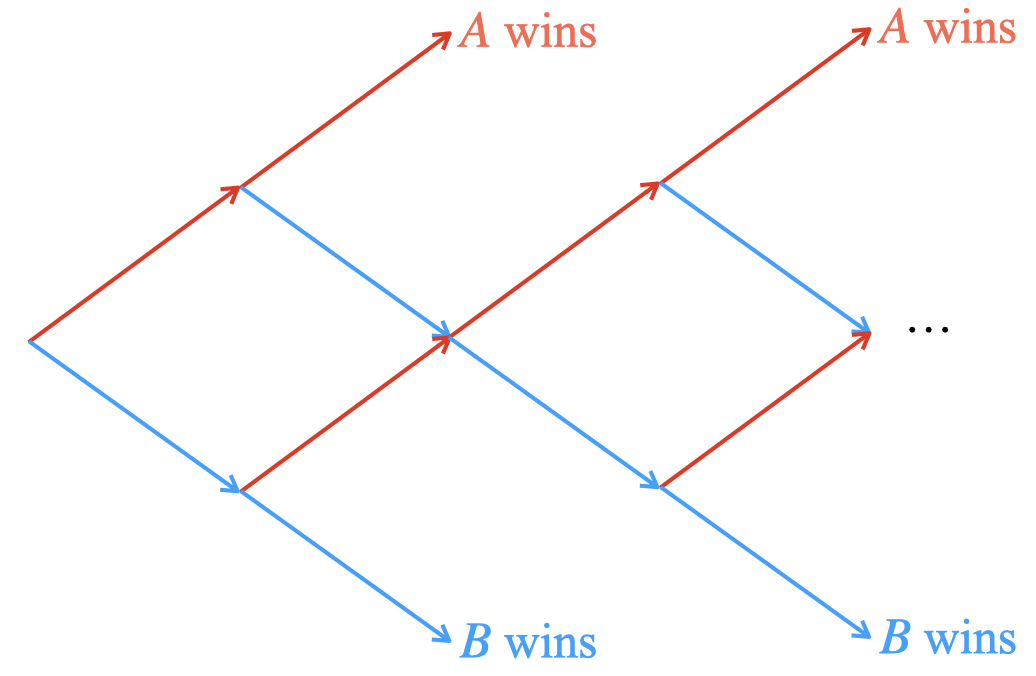}
        \label{fig:dyn-tow2}
    } % \vspace{5pt}
% \begin{minipage}{1\linewidth}
%     \footnotesize
% \end{minipage}
\caption{Example Contest Architectures.}
\label{fig:dyn-ex}
\end{figure}

Denote by $V_{\ell}(h)$ a player $\ell $'s continuation value in the contest
given a history $h \in H$---i.e., the equilibrium payoff the player expects
from the subsequent competition. The history ends up as $(h,\ell )$ if a
player $\ell $ wins the current battle and would otherwise evolve into $%
(h,\ell ^{\prime })$. Hence, the value generated by winning the battle is
given by%
\begin{equation*}
\Delta _{\ell }(h)=V_{\ell }(h,\ell )-V_{\ell }(h,\ell ^{\prime })\text{.}
\end{equation*}
Further, let $V_{\ell }^{0}$ be a player $\ell $'s continuation value at the
beginning of the contest, i.e., his equilibrium expected payoff in the
contest. The expected total effort the contest can elicit in equilibrium is
thus given by $v-\sum_{\ell \in \{A,B\}}V_{\ell }^{0}$.

Let $L(\mathcal{M}):=\min \{t:\ell ^{t}\in H^{\dagger }\}$ denote the
minimum length of the contest $\mathcal{M}$, i.e., the length of the
shortest terminal history. If $L(\mathcal{M})=+\infty $, so that all
terminal histories are infinite, then in equilibrium neither player would
exert any positive effort. Our analysis therefore focuses on contests that
can conclude after finitely many battles, which we call \emph{nontrivial
contests} and define this notion formally as follows.

\begin{definition}
\normalfont A contest is \emph{nontrivial} if $L(\mathcal{M}):=\min \{t:\ell
^{t}\in H^{\dagger }\}<+\infty $.
\end{definition}

It is not meaningful to study contests in which all terminal histories are
infinite. However, our analysis can be applied to examines sequences of
contests $\{\mathcal{M}_{k}\}_{k\in \mathbb{N}_{++}}$, in which each $%
\mathcal{M}_{k}$ is nontrivial but $L(\mathcal{M}_{k})\rightarrow +\infty $
as $k\rightarrow +\infty $. This formulation allows us to explore the
limiting properties of dynamic contest as they become arbitrarily long.

\section{Analysis: Bounded Rent Dissipation and Discouragement Effect}

In this section, we examine equilibrium rent dissipation rate---i.e., the
ratio of the expected equilibrium total effort to the prize value---in
dynamic contests. We establish upper bounds for rent dissipation and delve
into the forces that prevent full rent dissipation in equilibrium.

\subsection{An Upper Bound for Rent Dissipation}

We first establish an upper bound for rent dissipation in any nontrivial
contest. Recall that $L(\mathcal{M}):=\min \{t:\ell ^{t}\in H^{\dagger }\}$
denotes the length of the shortest terminal history of a contest $\mathcal{M}
$. The following ensues.

\begin{theorem}
\label{thm:L} Under \Cref{cond:bsf}, the expected total effort in any
equilibrium (if one exists) of contest $\mathcal{M}$ is less than $\left\{
1-[\phi(1)]^{L(\mathcal{M})}\right\} v$. Therefore, rent does not fully
dissipate in nontrivial contests.
\end{theorem}

By \Cref{thm:L}, the equilibrium rent dissipation rate in a contest bounded
from above, and the upper bound, $\left\{ 1-[\phi (1)]^{L(\mathcal{M}%
)}\right\} $, is determined by the shortest terminal history. With finite $L(%
\mathcal{M})$ in nontrivial contests, rent is never fully dissipated.%
\footnote{%
Of course, rent does not fully dissipate in a trivial contest with $L(%
\mathcal{M})=+\infty $ either, because equilibrium effort is 0.} For
instance, a standard best-of-$(2K+1)$ contest can end after $K+1$ battles, so
the upper bound is $\left\{ 1-[\phi (1)]^{K+1}\right\} $. Consider a
tug-of-war with a margin $N<+\infty $---in which a player must win $N\geq 1$ 
\emph{more} battles than his opponent to secure final victory. The contest
could last infinitely long if neither party manages to establishs a
sufficient margin. However, the length of its shortest terminal history is
exactly $N$, so its expected total effort cannot exceed $\left\{ 1-[\phi
(1)]^{N}\right\} v$ by the theorem.

\Cref{thm:L} inspires a natural question. Consider a sequence of contests $\{%
\mathcal{M}_{k}\}_{k\in \mathbb{N}_{++}}$, in which each $\mathcal{M}_{k}$
is nontrivial with at least one finite terminal history. Let $L(\mathcal{M}%
_{k})\rightarrow +\infty $ as $k\rightarrow +\infty $, and we examine the
limiting property of this sequence. Does the sequence of expected
equilibrium total efforts converge to $v$? The upper bound $1-[\phi (1)]^{L(%
\mathcal{M}_{k})}$ for rent dissipation rate provided by \Cref{thm:L}
approaches $1$ as $L(\mathcal{M}_{k})\rightarrow +\infty $. However, we next
examine the case of tug-of-war; the results demonstrate that this conjecture
does not hold and the upper bound $1-[\phi (1)]^{L(\mathcal{M}_{k})}$ does
not predict asymptotic properties for all contests.

\subsection{Case of Tug-of-War Contests}

\label{sec:nofull}

In this part, we focus on the case of tug-of-war contests with margin $N\geq
1$ and let $N$ grows to infinity, whereby the upper bound established in %
\Cref{thm:L} loses its bite. We establish a uniform upper bound that remains
valid asymptotically and show that rent does not fully dissipate in the limit.

\subsubsection{Equilibrium Analysis of Tug-of-War}

We first analyze the equilibrium of this contest game and adopt symmetric
Markov Perfect Equilibrium (MPE) as the solution concept. An MPE is also a
subgame-perfect equilibrium, but requires that players' strategies depend on
history only through the most recent state.

The state of a tug-of-war can be summarized by $i\in \mathbb{N}$ with $%
-N\leq i\leq N$, where $N$ is required margin for final victory and $i$
stands for a player's lead over this opponent, i.e., the number of extra
wins he currently possesses. The value of winning a battle (or the incentive
to win), $\Delta (\left. i\right\vert N,v)$, is given by the difference
between the continuation value after winning the battle and that after
losing it, i.e.,%
\begin{equation*}
\Delta (\left. i\right\vert N,v):=V(\left. i+1\right\vert N,v)-V(\left.
i-1\right\vert N,v).
\end{equation*}%
The opponent's winning value is thus $\Delta (\left. -i\right\vert
N,v):=V(\left. -i+1\right\vert N,v)-V(\left. -i-1\right\vert N,v)$.

A symmetric MPE can be described by the value function $V(\left.
i\right\vert N,v)$ for $i$ over $\left\{ \left. i\in \mathbb{N}\right\vert
-N\leq i\leq N\right\} $, which represents the expected value of a player
who leads by $i$ battles in a tug-of-war with margin $N$ and final prize $v$%
. For brevity, we omit $N,v$ in the expressions when there is no confusion.
Recall that a player $\ell $'s gain function in a battle---$\pi _{\ell }=\Pi
_{\ell }^{\ast }/\Delta _{\ell }=\phi \!\left( \Delta _{\ell }/\Delta
_{-\ell }\right) $---defined in \Cref{sec:component-battles}. In this
context, a player's
gain function in a battle with state $i$ is simply $\pi (i):=\phi \!\left(
\Delta (i)/\Delta (-i)\right) $. The equilibrium conditions are 
\begin{equation}
V(N)=v,~V(-N)=0,V(i)=V(i-1)+\pi (i)[V(i+1)-V(i-1)]  \label{eqn:mpe}
\end{equation}%
for $-(N-1)\leq i\leq N-1$. The following ensues.

\begin{proposition}[Existence of Symmetric MPE in Tug-of-War]
\label{prop:tow-mpe} Under \Cref{cond:bsf}, a unique symmetric MPE exists in
a tug-of-war with margin $N\geq 1$. In the equilibrium, $0=V(-N)<V(-N+1)<%
\ldots <V(N-1)<V(N)=v$.\footnote{\citet{ewerhart2019multi} establish
existence and uniqueness of a symmetric MPE for tug-of-war contests under
the additional assumption that $\frac{\phi (1/\theta )}{\theta \lbrack
1-\phi (\theta )]}$ is strictly declining in $\theta \in (0,1)$. We relax
this assumption by establishing the required monotonicity directly in our
proof.}
\end{proposition}

The equilibrium result paves the way for analysis of equilibrium rent
dissipation in the game.

\subsubsection{Discouragement and Uniform Upper Bound on Rent-dissipation
for Tug-of-War Contests}

Winning is self-reinforcing in a tug-of-war, since winning the current
battle increases the winner's relative incentive to win in the next battle,
i.e., its winning value $\Delta _{\ell }$ increasing relative to the
other's. The following result formalizes this logic.

\begin{lemma}[Self-reinforcement of Wins in Tug-of-War]
\label{lemma:tow-reinf}In the symmetric MPE of the tug-of-war with margin $N$%
, the following holds for $-(N-2)\leq i\leq N-2$: 
\begin{equation}
\frac{\Delta (i+1)}{\Delta (-i-1)}=\underbrace{\frac{1-\pi (-(i+1))}{\pi
(i+1)}}_{>1}\underbrace{\frac{1-\pi (i)}{\pi (-i)}}_{>1}\frac{\Delta (i)}{%
\Delta (-i)}>\frac{\Delta (i)}{\Delta (-i)}.  \label{eqn:tow-win-reinf}
\end{equation}
\end{lemma}

\Cref{lemma:tow-reinf} shows that, by winning a battle and moving the state
from $i$ to $i+1$, the winner's relative strength strictly increases for the
next battle: Both $\left[ 1-\pi (-(i+1))\right] /\pi (i+1)$ and $\left[
1-\pi (i)\right] /\pi (-i)$ are strictly greater than $1$, which enlarges
the ratio of winning values. Further, let $Q(i;N,v)$ denote the winning
probability of a player with $i$ extra wins for the contest overall. The
following result verifies that as the lead enlarges---i.e., the number of
extra wins increases---the probability of winning the overall contest
approaches 1.

\begin{proposition}[Unbounded Advantage Accumulation]
\label{prop:tow-adv} In the symmetric MPE of the tug-of-war with margin $N$,
it holds that $\lim_{i\rightarrow +\infty }\inf_{N>i}Q(i;N,v)=1$; that is,
the frontrunner secures an almost-sure win once its lead is sufficiently
large.
\end{proposition}

\Cref{prop:tow-adv} implies a discouragement effect that often arises in
dynamic contest. The laggard will be completely discouraged when the gap is
sufficiently large; he would give up the competition---which leads to $%
\lim_{i\rightarrow +\infty }\inf_{N>i}Q(i;N,v)=1$---even if the contest has
yet to end, i.e., $N>i$. This discouragement effect results in
underdissipation of rent in the contest: The concession of the laggard
affords the frontrunner easy wins. This is formalized by the following.

\begin{theorem}[Bounded Rent Dissipation in Tug-of-War]
\label{thm:tow-rent} Under \Cref{cond:bsf}, there exists a constant $\alpha
>0$ such that $V(0)\geq \alpha v$ for all $N\geq 1$ and $v>0$. Consequently,
the expected total effort in any symmetric MPE of the tug-of-war with any
margin $N\geq 1$ is bounded from above by $(1-2\alpha )v$.
\end{theorem}

\Cref{thm:tow-rent} establishes a uniform upper bound on rent dissipation in
tug-of-war contests that applies for all $N$ and remains valid
asymptotically. This formally verifies that rent cannot fully dissipate in
any tug-of-war contest, even when the marginal requirement for victory
becomes arbitrarily large and the lengths of all terminal histories grow
without bound.

The frontrunner's advantage accumulates unboundedly.

\subsection{Exchangeable Contests}

\Cref{thm:tow-rent} establishes in the case of tug-of-war that rent does not
necessarily dissipate in a contest even if its minimum length can be
infinitely long and demonstrates a discouragement effect that limits effort
provision. We now develop a general condition---exchangeability---for
rent underdissipation that underpins discouragement effect in dynamic
contests.

\subsubsection{Exchangeability and Discouragement Effect}

\label{sec:exchangeable}
We first provide a formal definition of
exchangeability and then elaborate on its role in shaping players' dynamic
incentives.

\begin{definition}
\label{exchangeable} \normalfont A contest is \emph{exchangeable} if any two
histories that differ only in the order of the battle outcomes lead to the
same subgame or final outcome.
\end{definition}

Exchangeability is satisfied by many dynamic contests, such as tug-of-war
and best-of-$(2K+1)$ with $K\in \mathbb{N}_{++}$. Consider, for instance, a
best-of-$(2K+1)$ contest with $K>1$ or a tug-of-war with $N>2$. Two
histories, $(A,B,A)$ and $(A,A,B)$, lead to the same subsequent play. By
definition, all histories of an exchangeable contest can be summarized by a
pair $(i,j)$, where $i$ and $j$ are, respectively, the number of wins player 
$A$ and player $B$ have respectively secured. As a result, the equilibrium
of any exchangeable contest can be described by the value function $V_{\ell
}^{i,j}$ that represents player $\ell $'s continuation value at state $(i,j)$%
.

Exchangeability allows a frontrunner to accumulate advantage, thereby
catalyzing the discouragement effect. For player $A$, given past outcomes $%
(\ell _{1},\ldots ,\ell _{t})$, his value of winning the next battle is
given by $V(\ell _{1},\ldots ,\ell _{t},A)-V(\ell _{1},\ldots ,\ell _{t},B)$;
this difference coincides with the gap between the continuation values
associated with the alternative histories, $(A,\ell _{1},\ldots ,\ell _{t})$
and $(B,\ell _{1},\ldots ,\ell _{t})$. Consequently, the impact of an early
outcome does not vanish as the contest unfolds, even after a large number of
subsequent battles. This long-memory property of exchangeable contests
amplifies early success into a persistent advantage, while an early
disadvantage---possibly arising from an accidental loss---can be perpetuated
indefinitely. As a result, the laggard is discouraged from catching up,
which dampens incentives for both players and leads to rent underdissipation.

We now develop this intuition into a formal analysis of rent dissipation in
the limit in exchangeable contests. For this purpose, we introduce the
notion of\emph{\ (exchangeability-preserving) extension}. For expositional
ease, we focus on contests that are symmetric, meaning the contest rules are
independent of the player's identity.

\begin{definition}
\normalfont 
For a given symmetric exchangeable contest $\mathcal{M}$, a symmetric
contest $\mathcal{M}^{\prime}$ is said to be the \emph{$N$-extension} of $%
\mathcal{M}$, with $N\geq2$, if both the following conditions are met:

\begin{enumerate}
\item[(i)] $\mathcal{M}^{\prime }$ is exchangeable, and its subgame
following histories $(A,B)$ or $(B,A)$ is $\mathcal{M}$;

\item[(ii)] from the start of $\mathcal{M}^{\prime }$, a player wins the
contest after winning exactly $N$ battles in a row.
\end{enumerate}
\end{definition}

For instance, a best-of-$(2K+1)$ contest is the $(K+1)$-extension of the
best-of-$(2K-1)$ contest. As another example, a tug-of-war with margin-$N$
is the $N$-extension of itself. With this definition, we explore a sequence
of exchangeable contests $\{\mathcal{M}_{k}\}_{k=1}^{K}$ such that, for each 
$k$, the contest $\mathcal{M}_{k+1}$ is an $N_{k+1}$-extension of $\mathcal{M%
}_{k}$. By construction, $L(\mathcal{M}_{k})\leq N_{k}$ for every $k$, and %
\Cref{thm:L} therefore implies that rent cannot fully dissipate in any $%
\mathcal{M}_{k}$. Accordingly, the substantive question is whether total
effort can converge to full dissipation along an increasingly long sequence,
so we focus on the asymptotic case $K\rightarrow +\infty $. Further, since
we focus on the limit behavior, it is without loss to assume that the
initial contest $\mathcal{M}_{1}$ is a single-battle contest. Let $V_{0;k}$
denote the equilibrium payoff of a player in $\mathcal{M}_{k}$. The
following result ensues.

\begin{proposition}[Limit of Exchangeable Contests]
\label{prop:homogeneous-ex} Impose \Cref{cond:bsf}. For any sequence of
exchangeable extensions $\{\mathcal{M}_{k}\}_{k=1}^{+\infty }$ such that $%
\lim_{k\rightarrow +\infty }V_{0;k}$ exists, it holds that $%
\lim_{k\rightarrow +\infty }V_{0;k}\geq \widetilde{\alpha }v$, where $%
\widetilde{\alpha }>0$ is constant that depends only on the success
function. That is, rent does not fully dissipate in the limit of the
sequence of exchangeable contests, and the expected total effort is less
than $(1-2\widetilde{\alpha })v$.
\end{proposition}

\Cref{sec:exchangeable} thus establishes a uniform upper bound of
equilibrium rent dissipation in symmetric exchangeable contest that remains
valid asymptotically. Rent does not fully dissipate in an exchangeable
contest even if it continues infinitely. The key takeaway here is the role
of exchangeability in shaping players' dynamic incentives. Exchangeability,
with its long-memory property, allows early success to translate into a
lasting advantage. This persistence discourages the laggard and attenuates
competition. The well-known discouragement effect in the dynamic contest
literature thus arises.

To further illustrate the nuances of exchangeability, we now examine a
simple dynamic contest with a contest rule that violates exchangeability. We
show that, in this setting, full rent dissipation can be achieved in the
limit.

\subsection{Consecutive-win Contests: When Exchangeability is Absent}

Consider a $K$-consecutive-win contest, with $K\in \mathbb{N}_{++}$: Namely,
the player who first wins $K$ battles in a row is awarded the final prize.
Notably, this contest is not exchangeable: For example, say $K=3$, $(A,A,B)$
and $(A,B,A)$ do not lead to the same subsequent subgame: Player $B$ holds
an advantage in the former case, whereas player $A$ is leading in the
latter. More importantly, it is worth noting that memory is short in this
contest: The impact of one's wins is entirely wiped out once he loses a
battle before achieving $K$ consecutive wins.

We first establish equilibrium existence in this. As before, we focus on
symmetric MPE with the state space $\{i\in \mathbb{N}:-K\leq i\leq K\}$,
where $i\geq 0$ denotes the number of consecutive wins accumulated by the
leader, and $-i$, with $i<0$, represents the number of losses the laggard
has recorded in the current losing streak. Obviously, if one player is in
state $i$, his opponent must be in state $-i$. The symmetric MPE is
characterized by a value function $\widehat{V}(\left. i\right\vert K,v)$
that represents a player's continuation value in state $i$ of the $K$%
-consecutive-win contest with a prize $v$.\ For convenience, when no
confusion arises, we suppress the arguments $K$ and $v$ for $\widehat{V}%
(\cdot )$ and other equilibrium objects defined below.

To specify equilibrium conditions, for a player in state $i$, we let $%
\widehat{\Delta }(i)$ denote the difference in his continuation values after
winning and losing the upcoming battle:%
\begin{equation*}
\widehat{\Delta }(i):=%
\begin{cases}
\widehat{V}(i+1)-\widehat{V}(-1), & \text{ if }K-1\geq i\geq 0, \\ 
\widehat{V}(1)-\widehat{V}(i-1), & \text{ if }-(K-1)\leq i<0.%
\end{cases}%
\end{equation*}%
Clearly, $\widehat{\Delta }(i)$ measures the player's incentive to win and
the value of this battle. The equilibrium conditions are: $\widehat{V}(K)=v$%
, $\widehat{V}(-K)=0$, and 
\begin{equation}
\widehat{V}(i)=%
\begin{cases}
\widehat{V}(-1)+\widehat{\pi }(i)\widehat{\Delta }(i), & \text{ if }K-1\geq
i\geq 0, \\ 
\widehat{V}(i-1)+\widehat{\pi }(i)\widehat{\Delta }(i), & \text{ if }%
-(K-1)\leq i<0.%
\end{cases}
\label{consective_win_valuation}
\end{equation}%
where $\widehat{\pi }(i):=\phi \left( \widehat{\Delta }(i)/\widehat{\Delta }%
(-i)\right) $ represents the player's equilibrium gain ratio in state $i$
for the upcoming battle. We obtain the following.

\begin{proposition}[Existence and Uniqueness of Symmetric MPE in
Consecutive-win Contests]
\label{prop:cw-mpe} Under \Cref{cond:bsf}, a unique symmetric MPE exists for
the $K$-consecutive-win contest with $K\in \mathbb{N}_{++}$. In the
symmetric MPE, $0=\widehat{V}(-K)<\widehat{V}(-K+1)<\ldots <\widehat{V}(K-1)<%
\widehat{V}(K)=v $.
\end{proposition}

As in tug-of-war contests, winning in consecutive-win contests is also
self-reinforcing, which is demonstrated by the following.

\begin{lemma}[Self-reinforcement of Wins in Consecutive-win Contests]
\label{lemma:cw-reinf} In any symmetric MPE of a $K$-consecutive-win
contest, the following holds for $0\leq i\leq K-2$: 
\begin{equation}
\frac{\widehat{\Delta }(i+1)}{\widehat{\Delta }(-i-1)}=\underbrace{\frac{1-%
\widehat{\pi }(-i-1)}{\widehat{\pi }(i+1)}}_{>1}\frac{\widehat{\Delta }(i)}{%
\widehat{\Delta }(-i)}.
\end{equation}
\end{lemma}

\Cref{lemma:cw-reinf} shows that as the leader accumulates more consecutive
wins, he becomes stronger in the sense that his incentive to win the next
battle increases relative to that of his opponent. However, the laggard is
never completely discouraged, regardless of how many consecutive wins the
leader has accumulated, as long as final victory still requires a large
number of additional consecutive wins. Formally, let $\widehat{Q}(i;K,v)$
denote a player $i$'s equilibrium probability of ultimately winning the
prize when the state is $i$ of the $K$-consecutive-win contest. We have the
following result.

\begin{proposition}[Bounded Accumulated Advantage in Consecutive-win Contests%
]
\label{prop:cw-adv} Fix any symmetric MPE of the consecutive-win contest.
Under \Cref{cond:bsf}, it holds that $\lim_{K>|i|,K\rightarrow +\infty }%
\widehat{Q}(i;K,v)=\frac{1}{2}$ for all $i\in \mathbb{N}$.
\end{proposition}

\Cref{prop:cw-adv} implies that a lead in the consecutive-win contest is 
\emph{insecure}, in sharp contrast to \Cref{prop:tow-adv}. Intuitively, when
the laggard wins a battle in a tug-of-war contest, the frontrunner's lead is
merely reduced, whereas in a $K$-consecutive-win contest, a single loss
entirely erases the lead and more than fully rebalances the playing field.
This feature largely mutes the discouragement effect and leads to the
following result.

\begin{theorem}[Full Rent Dissipation in the Limit of Consecutive-win Contests%
]
\label{thm:cw-rent} Under \Cref{cond:bsf}, for all $\epsilon >0$, there
exists $K^{\dagger }$ such that $\widehat{V}(0;K,v)<\epsilon v$ in the
equilibrium of the $K$-consecutive-win contest for $K>K^{\dagger }$, which
implies that the expected total effort is greater than $(1-2\epsilon )v$.
\end{theorem}

The theorem formally establishes that, in the absence of exchangeability,
the $K$-consecutive-win contest fully dissipates the rent in the limit.

\section{A Necessary-and-Sufficient Condition for Almost-Full Rent
Dissipation}
\label{sec:full-rent-condition}

In this section, we first provide a necessary-and-sufficient condition for
almost-full rent dissipation in dynamic multi-battle contests. We then
introduce an \emph{iterated incumbency contest} and demonstrate that it
satisfies this condition and fully dissipates rent in the limit. This
contest provides an intuitive framework for studying a wide range of
real-world competitive phenomena, spanning technological competition and
biological evolutionary processes.

\subsection{Transient Dominance Property and Almost-Full Rent Dissipation}

\Cref{sec:exchangeable} unveils the critical role played by exchangeability in
catalyzing the discouragement effect. The long memory it induces in the
contest allows a player to leverage his early success to accumulate a
persistent advantage, which undermines the fluidity and contestability of
the competition and thereby discourages effort provision. In contrast, the
consecutive-win contest achieves almost-full rent dissipation. It rules out
exchangeability and renders the impacts of winning a battle vis--\`{a}-vis
losing one asymmetric: advantage must accumulate gradually, but leadership
can be lost abruptly. This feature keeps the laggard hopeful and
incentivized, thereby nullifying the discouragement effect.

We now formalize such insights and first propose the following.

\begin{definition}
\label{def:si} \normalfont Consider a contest $\mathcal{M}$ and its
equilibrium. Fix any small $\epsilon>0$. We say that the equilibrium has the 
\emph{transient dominance property} if the following conditions are
satisfied:

\begin{enumerate}
\item[(i)] There exist two subsets of the nonterminal histories, $H_{A}^{-}$
and $H_{B}^{-}$, such that at any history $h\in H_{\ell }^{-}$, the
continuation value $V_{\ell }(h)\leq \epsilon v$.

\item[(ii)] In equilibrium, the probability that the realized outcome
history reaches both $H_{A}^-$ and $H_{B}^-$ is at least $1-\epsilon$.
\end{enumerate}
\end{definition}

Condition (i) identifies, for each player $\ell $, a set $H_{\ell }^{-}$ of
nonterminal histories at which the player $\ell $ is in a \textquotedblleft
weak\textquotedblright\ position: The small continuation value---bounded by $%
\epsilon v$---limits the player's incentive to exert costly effort, implying
that his opponent is likely to be in a dominant position in the contest.%
\footnote{%
It is possible that the opponent also has a low continuation value. However,
this implies that the contest is already close to full rent dissipation. So
we primarily interpret such a scenario as an outcome of players' ex post
asymmetry due to a temporary performance gap.} Condition (ii) nevertheless
requires that, along the equilibrium path, the contest can reach \emph{both} 
$H_{A}^{-}$ and $H_{B}^{-}$ with sufficiently large probabilities. The
outcome path must visit at least one of $H_{A}^{-}$ or $H_{B}^{-}$ before
the contest ends. Condition (ii) implies that, conditional on the outcome
path arriving at one of $H_{A}^{-}$ or $H_{B}^{-}$ for the first time, there
is a high chance of a shift in dominance down the road, i.e., a transition
from $H_{\ell }^{-}$ to $H_{-\ell }^{-}$ along the subsequent outcome path.
Taken together, the definition implies that a player's lead is temporary and
will almost certainly be reversed. Our analysis yields the following.

\begin{theorem}[Transient Dominance Property as a Necessary and Sufficient
Condition for Almost-Full Rent Dissipation]
\label{thm:sufficient} If an equilibrium of a contest $\mathcal{M}$ with
prize $v$ exhibits the transient dominance property, the expected
equilibrium total effort is greater than $(1-4\epsilon )v$. Conversely, if
the expected total effort in an equilibrium of contest $\mathcal{M}$ with
prize $v$ exceeds $(1-\epsilon )v$, then the equilibrium has the transient
dominance property.
\end{theorem}

\Cref{thm:sufficient} shows that the transient dominance property both
implies and is necessary for almost-full rent dissipation. A contest almost
fully dissipates its rent---that is, its expected equilibrium total effort
approaches $v$---if and only if its equilibrium satisfies the transient
dominance property.

Contests with the transient dominance property are not rare. An immediate
example arises from a simple modification of a standard tug-of-war contest.
Suppose that, after each battle, with probability $p\in (0,1)$, the contest
resets to its initial state, provided that no final winner has yet been
determined. That is, following a reset, each player must again achieve a
lead of $N$ wins over the opponent to secure final victory before another
reset occurs. In this environment, exchangeability fails, and any lead is
inherently fragile. We have the following result.

\begin{theorem}[Tug-of-war with Random Resets]
\label{thm:towp} There exists a unique symmetric MPE in a tug-of-war contest
with margin $N\geq 2$ and a probability of reset $p\in \lbrack 0,1)$. For
any small $\epsilon >0$, there exists $N^{\dagger }$ such that the
tug-of-war with margin $N\geq N^{\dagger }$ and reset probability $p\in
(0,1) $ satisfies the transient dominance property.
\end{theorem}

This tug-of-war with random resets exhibits the transient dominance property
when the required margin $N$ is sufficiently large. As predicted by %
\Cref{thm:sufficient}, rent is then almost fully dissipated in equilibrium.
In what follows, we present an alternative contest game that also exhibits
this property and discuss its implications.

\subsection{Application: Dynamic Competition in an Uncertain Environment}

We now examine an intuitive dynamic contest that illustrates how %
\Cref{def:si} applies and leads to full rent dissipation. Economic agents
often engage in protracted competition in uncertain environments. Consider,
for instance, two firms striving for dominance in an evolving market driven
by an emerging technology; technological progress and shifts in consumer
tastes continuously reshape the competitive landscape. As a result, a firm's
temporary market dominance may not suffice to eliminate its opponent or to
establish permanent market leadership. When the underlying fundamentals
change, the laggard may regain the opportunity to overtake the frontrunner.
We incorporate these features and construct a model of dynamic competition
with uncertainty, which we call an \emph{iterated incumbency contest} and
describe in detail below.

The dynamic competition unfolds over multiple \emph{rounds}, indexed by $%
n=0,1,\ldots ,N$. In round 0, a fair coin toss determines a temporary leader
(i.e., the \emph{incumbent}) and a laggard (i.e., the challenger). In each
round $n\geq 1$, an exogenous shock determines whether the fundamentals of
the contest have shifted and renewed competition emerges: With probability $%
q\in (0,1]$, the fundamentals change, and the firms must compete for the
incumbency; with the complementary probability $1-q$, the status quo is
preserved, and the incumbent retains his status without a battle. The firm
which exits round $N$ as the incumbent secures the final victory and
receives the prize $v$.

The competition for incumbency, once triggered, proceeds as a subcontest
with multiple successive battles, which we denote by $\mathcal{M}^{\text{sub}%
}$. To reflect the fact that an incumbent typically holds an advantageous
position in competition, the subcontest is assumed to be biased in favor of
the incumbent. This bias can take various forms. One example is that the
incumbent can win the subcontest by winning a single battle; while the
challenger must win $K>1$ battles before the incumbent wins one. We refer to
such a subcontest as a $\mathcal{M}(K,1)$ contest. Alternatively, $\mathcal{M%
}^{\text{sub}}$ can take the form of a tug-of-war contest with margin $K+1$
and initial state $K\geq 1$, whereby the incumbent enjoys a head start of $K$%
.

We focus on subcontests that are \emph{sufficiently biased} in favor of the
incumbent. Consider a subcontest $\mathcal{M}^{\text{sub}}$, and imagine an
auxiliary \emph{standalone} contest with the same structure and winning rule
as $\mathcal{M}^{\text{sub}}$, but with a unit prize. Let $V_{+}^{\text{sub}%
} $ and $V_{-}^{\text{sub}}$ denote the respective equilibrium payoffs of
the incumbent and the challenger in the auxiliary contest. For any fixed
small $\epsilon >0$, a (sub)contest $\mathcal{M}$ is said to be \emph{%
sufficiently biased} if%
\begin{equation*}
\frac{1-V_{+}^{\text{sub}}}{V_{-}^{\text{sub}}}>\frac{1}{\epsilon }.
\end{equation*}%
The condition requires infinitesimal $V_{-}^{\text{sub}}$, which implies his
unfavorable position in the competition.\footnote{%
This condition can also be satisfied if the subcontest or the auxiliary
contest is unbiased, but both $V_{+}^{\text{sub}}$ and $V_{-}^{\text{sub}}$
are infinitely small. In this case, the analysis is trivial because this
auxiliary contest itself would fully dissipate rent in equilibrium.} The two
intuitive contest formats mentioned above, $\mathcal{M}(K,1)$ contest and
tug-of-war contest with margin $K+1$ and initial state $K$, both satisfy the
requirement in the limit, which is formally verified below.

\begin{lemma}
\label{lemma:ratio-instances} Let the subcontest take the form of either an $%
\mathcal{M}(K,1)$ contest or a tug-of-war contest with margin $K+1$ and
initial state $K$. Then $\frac{1-V_{+}^{\text{sub}}}{V_{-}^{\text{sub}}}%
\rightarrow +\infty $ as $K\rightarrow +\infty $.
\end{lemma}

The structure of an iterated incumbency contest is illustrated by %
\Cref{fig:iter} below. If no competition occurs in a given round---which
occurs with probability $1-q$---the incumbent is deemed to win that round
automatically.

\begin{figure}[th]
\begin{center}
\resizebox{1\textwidth}{!}{
\tikzset{
		arrow1/.style = {
			draw =black, semithick, -{Stealth[length = 3mm, width = 2mm]},
		}	
	} 
\tikzset{
	arrow2/.style = {
		draw =black, semithick,densely dashed,-{Stealth[length = 3mm, width = 2mm]},
		}	
	}
\par
\begin{tikzpicture}[]
		\node[]at(-1.5,3.2){\footnotesize Round 1 };
	\node[]at(1.5,3.2){\footnotesize Round 2};
    \node[]at(4.2,3.2){\footnotesize Round 3};
	\node[]at(8,3.2){\footnotesize Round $N$};

	\node [draw, rectangle,inner sep=0.22cm,label distance = 0.05cm]at(-3,2) (a1){\footnotesize Incumbent};
		\node[right of=a1,draw,rectangle,inner sep=0.22cm,node distance=3cm,on grid](a2) {\footnotesize Incumbent};
	\node[right of=a2,draw,rectangle,inner sep=0.22cm,node distance=3cm,on grid](a3) {\footnotesize Incumbent};
    \node[right of=a3,rectangle,inner sep=0.22cm,node distance=2cm,on grid](a4) {};
    \node[right of=a4,rectangle,inner sep=0.22cm,node distance=2cm,on grid](a5) {}; 
    \node[right of=a5,rectangle,inner sep=0.22cm,node distance=2cm,on grid](a6) {}; 
    
    \node [draw, rectangle,inner sep=0.22cm,label distance = 0.05cm]at(-5,0.5) (a0){\footnotesize Initial};
    
    \node [draw, rectangle,inner sep=0.22cm,label distance = 0.05cm]at(10,2) (e1){\footnotesize Final Winner};

    \node [draw, rectangle,inner sep=0.22cm,label distance = 0.05cm]at(10,-1) (e3){\footnotesize Final Loser};

    \node [rectangle,inner sep=0.22cm,label distance = 0.05cm]at(6,2) (c1){\Large $\cdots$};
    \node [rectangle,inner sep=0.22cm,label distance = 0.05cm]at(6,0.5) (c2){\Large $\cdots$};
    \node [rectangle,inner sep=0.22cm,label distance = 0.05cm]at(6,-1) (c3){\Large $\cdots$};

		\node[below of=a1,draw,rectangle,inner sep=0.22cm,node distance=3cm,on grid](b1) {\footnotesize Laggard};
	\node[right of=b1,draw,rectangle,inner sep=0.22cm,node distance=3cm,on grid](b2) {\footnotesize Laggard};
	\node[right of=b2,draw,rectangle,inner sep=0.22cm,node distance=3cm,on grid](b3) {\footnotesize Laggard};
    \node[right of=b3,rectangle,inner sep=0.22cm,node distance=2cm,on grid](b4) {};
    \node[right of=b4,rectangle,inner sep=0.22cm,node distance=2cm,on grid](b5) {};
	\node[right of=b5,rectangle,inner sep=0.22cm,node distance=2cm,on grid](b6) {};

	\draw[arrow1](a1)node[above,xshift = 40pt]{\footnotesize win} -- (a2);	
	\draw[arrow1](a2)node[above,xshift = 40pt]{\footnotesize win} -- (a3);
	\draw[arrow1](b1)node[above,xshift = 33pt,yshift=15pt]{\footnotesize win} -- (a2);
	\draw[arrow1](b2)node[above,xshift = 33pt,yshift=15pt]{\footnotesize win} -- (a3);
    \draw[arrow1](a0)node[above,xshift = 20pt,yshift=18pt]{\footnotesize win} -- (a1);
    \draw[arrow1](b3)node[above,xshift = 25pt,yshift=15pt]{\footnotesize win} -- (a4);
    \draw[arrow1](a3)node[above,xshift = 35pt]{\footnotesize win} -- (a4);
    \draw[arrow1](b5)node[above,xshift = 25pt,yshift=15pt]{\footnotesize win} -- (a6);
    \draw[arrow1](a5)node[above,xshift = 25pt]{\footnotesize win} -- (a6);
	
		\draw[arrow2](a1)node[below,xshift = 37pt,yshift=-15pt]{\footnotesize lose} -- (b2);
	\draw[arrow2](b1)node[below,xshift = 40pt]{\footnotesize lose} -- (b2);
	\draw[arrow2](a2)node[below,xshift = 37pt,yshift=-15pt]{\footnotesize lose} -- (b3);
	\draw[arrow2](b2)node[below,xshift = 40pt]{\footnotesize lose} -- (b3);
    \draw[arrow2](b3)node[below,xshift = 35pt]{\footnotesize lose} -- (b4);
    \draw[arrow2](a3)node[below,xshift = 27pt,yshift=-15pt]{\footnotesize lose} -- (b4);
    \draw[arrow2](b5)node[below,xshift = 25pt]{\footnotesize lose} -- (b6);
    \draw[arrow2](a5)node[below,xshift = 27pt,yshift=-15pt]{\footnotesize lose} -- (b6);
	\draw[arrow2](a0)node[below,xshift = 22pt,yshift=-18pt]{\footnotesize lose} -- (b1);
	\end{tikzpicture}
}
\end{center}
\caption{$N$-round Iterated Incumbency Contest.}
\label{fig:iter}
\end{figure}
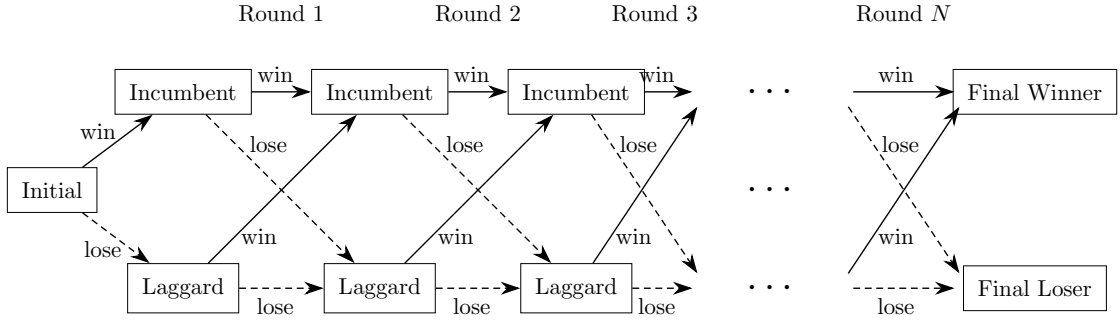

\subsubsection{Rent Dissipation in Iterated Incumbency Contests}

We are now ready to apply the sufficiency part of \Cref{thm:sufficient} to
show that the iterated incumbency contest can approximate full rent
dissipation. The following result is obtained.

\begin{theorem}[Full Dissipation in Iterated Incumbency Contests]
\label{prop:full-suff} Consider a $N$-round iterated incumbency contest with
subcontest $\mathcal{M}^{\text{sub}}$ and impose \Cref{cond:bsf}. If the
subcontest is sufficiently biased, the iterated incumbency contest satisfies
the transient dominance property as $N\rightarrow +\infty $, which leads to
asymptotic full rent dissipation. 
%%Suppose that its subcontest $\mathcal{M}^{\text{sub}}$ satisfies the
%%transient dominance property with parameters $M=\frac{1-V_{+}^{\text{sub}}}{%
%%V_{-}^{\text{sub}}}$ and $P=1-(qQ_{+}^{\text{sub}}+1-q)^{N}$. 
%%%%Therefore, if the ratio $%
%%%%\frac{1-V_{+}^{\text{sub}}}{V_{-}^{\text{sub}}}$ is sufficiently large, the
%%%%iterated incumbency contest fully dissipates rent with $N\rightarrow +\infty$.
\end{theorem}

The intuition is as follows. First, a large $\left( 1-V_{+}^{\text{sub}%
}\right) $ relative to $V_{-}^{\text{sub}}$ ensures that the incumbent
remains strongly motivated to exert effort toward final victory and avoids
being relegated to the status of challenger, which is associated with a unit
prize in the auxiliary contest. Second, given homogeneity of each battle's
success function, the probability that the incumbent wins a subcontest is
independent of the total number of rounds $N$. As $N$ becomes large,
incumbency can almost certainly be overturned in some round. Consequently,
the challenger remains hopeful of overtaking the incumbent before the
contest ends, rendering leadership fragile and preserving contestability
throughout the game. Taken together, these features generate transient
dominance, which ensures full rent dissipation in the limit.

We now provide a graphical illustration to clarify the logic of the iterated
incumbency contest, i.e., how equilibrium total effort can rapidly converge
to the prize value even when $N$ is small. In \Cref{fig:iter-vals}, the blue
(diamond) dot near the upper left corner represents their continuation
values at the beginning of round $N$ when player $B$ is the incumbent. Since
this is the final round and its outcome decides the overall winner of the
contest, the location of this point is $(qV_{-}^{\text{sub}}v,(1-q+qV_{+}^{%
\text{sub}})v)$, as implied by the rule of the subcontest $\mathcal{M}^{%
\text{sub}}$ and the terminal payoff profiles $(0,v)$ or $(v,0)$. The red
(solid) arrow and the black (dotted) arrow, respectively, point to the
terminal payoffs when the incumbent wins and when the challenger wins.

Notably, the absolute value of the slope of the red arrow is exactly $\left(
1-V_{+}^{\text{sub}}\right) /V_{-}^{\text{sub}}$. Symmetrically, there is a
blue dot near the lower right corner corresponding to the case in which $A$
is the incumbent at the beginning of round $N$.

Then, moving one round backward, the black (round) dot near the upper left
corner represents the players' continuation values at the beginning of round 
$N-1$ when $B$ is the incumbent. Again, the red and black arrows illustrate
how the continuation values evolve once that round concludes. As we move
further backward, the continuation values of the incumbent and the
challenger become progressively closer. Eventually, as shown in %
\Cref{fig:iter-vals}, they converge to the same point, represented by the
circle in the figure.

\Cref{prop:full-suff} implies that when the subcontest is sufficiently
biased, this limiting point (i.e., the circle) can be made arbitrarily close
to the origin. In other words, their ex ante equilibrium expected payoff in
the contest, $V_{A}^{0}$ and $V_{B}^{0}$, converge to zero, implying full
rent dissipation.

\begin{figure}[th]
\centering
\begin{tikzpicture}[scale=7]
    \node[font=\footnotesize] at (-0.08,0.5) {$V_B$};
    \node[font=\footnotesize] at (0.5,-0.08) {$V_A$};
    \node[font=\footnotesize, align=left] at (0.3,0.92){Incumbent (i.e. player $B$) \\ wins round $N$};
    \node[font=\footnotesize, align=left] at (0.78,0.42){Laggard (i.e. player $A$) \\ wins round $N$};

%%  % 
%%\draw[-] (0,0) -- (1,0) node[pos=1, below, font = \footnotesize] {$1$};
%%\draw[-] (0,0) -- (0,1) node[pos=1, left, font = \footnotesize] {$1$};
%%\draw[-] (0,0) -- (0,0) node[pos=1, below left, font = \footnotesize] {$0$};
\draw[-] (0,0) -- (1,0) node[pos=1, below, font=\footnotesize] {$v$};
\draw[-] (0,0) -- (0,1) node[pos=1, left,  font=\footnotesize] {$v$};
\draw      (0,0) node[below left, font=\footnotesize] {$0$};

  % network
%%%  \draw[step=0.1,gray!20,thin] (0,0) grid (1,1);
  % node
  \coordinate (vup) at (0,1.0);
  \coordinate (vdown) at (1.0,0);
%%    \coordinate (ori) at (0,0);
  \coordinate (I1) at (0.03,0.831);
  \coordinate (L1) at (0.831,0.03);
  \coordinate (I2) at (0.054,0.695);
  \coordinate (L2) at (0.695,0.054);
  \coordinate (-1) at (0,0);
  \coordinate (0) at (0.03,0.03);
  \coordinate (1) at (0.054,0.054);
  \coordinate (3) at (0.1507,0.1507);

  \node[
        diamond,
        draw=blue,
        fill=blue,
        minimum size=4pt,
        yscale=1.5,
        inner sep=0pt,
        label=left:{}
      ] at (I1) {};
    \node[
        diamond,
        draw=blue,
        fill=blue,
        minimum size=4pt,
        yscale=1.5,
        inner sep=0pt,
        label=left:{}
      ] at (L1) {};
%%  \fill[black] (vup) circle (0.01) node[above left] {};
%%  \fill[black] (vdown) circle (0.01) node[below right] {};
  \fill[black] (I2) circle (0.01) node[above left] {};
  \fill[black] (L2) circle (0.01) node[below right] {};
  \node[draw, circle, minimum size=5pt, inner sep=0pt] at (3) {};

   \draw[red, thick, ->,
          >={Straight Barb[length=3pt,width=3pt]},
          shorten >=0.5mm, shorten <=0.7mm]
      (I1) -- (vup);

    \draw[dashed, thick, ->,
          >={Straight Barb[length=3pt,width=3pt]},
          shorten >=0.5mm, shorten <=0.7mm]
      (I1) -- (vdown);

    \draw[dashed, thick, ->,
          >={Straight Barb[length=3pt,width=3pt]},
          shorten >=0.5mm, shorten <=0.5mm]
      (I2) -- (L1);
      
    \draw[red, thick, ->,
          >={Straight Barb[length=3pt,width=3pt]},
          shorten >=0.5mm, shorten <=0.4mm]
      (I2) -- (I1);

  \draw[thin]
  (vdown) -- (L1) -- (L2) -- (3) -- (I2) ;
%  \draw[thin, dash pattern=on 1pt off 1pt](vdown) -- (vup);
  \draw[thin, dash pattern=on 1pt off 1pt]
  (I1) -- (L1) -- (0.03,0.03) -- (I1);

  \draw[thin, dash pattern=on 1pt off 1pt]
  (I2) -- (L2) -- (0.054,0.054) -- (I2);

 \draw[thin, dash pattern=on 1pt off 1pt]
  (0.3,0.12) -- (0.12,0.3) -- (0.12,0.12) -- (0.3,0.12);
\end{tikzpicture}
\caption{Transition of Continuation Values in an Iterated Incumbency
Contest. }
\label{fig:iter-vals}
\end{figure}
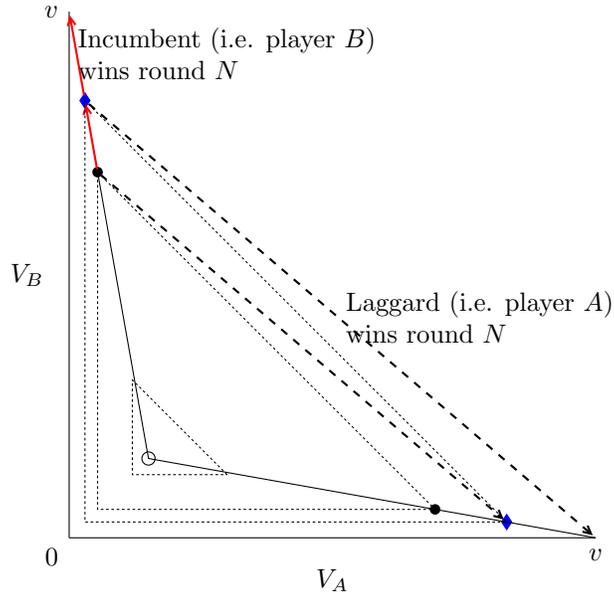

Next, we discuss how the model and analysis shed light on a range of
real-world competitive phenomena.

\subsubsection{Iterated Incumbency Contests: Applications and Relevance}

A diverse array of socioeconomic and biological phenomena exhibit a
structural pattern resembling an iterated incumbency contest. The model
captures protracted rivalries in which exogenous shocks periodically
rejuvenate competition, preventing temporary advantages from compounding
into permanent supremacy. In such environments, leadership is inherently
contingent: While incumbency confers a short-term edge, that advantage
operates only within a given regime and does not accumulate indefinitely
across regimes. Dominance is therefore transient, not because of explicit
institutional design, but because the surrounding landscape shifts and
repeatedly resets the contest.

Schumpeterian competition provides a natural economic illustration of this
mechanism. Radical innovations and technological paradigm shifts often
function as de facto reset mechanisms, narrowing the effective gap between
established firms and emerging challengers. Nokia's dominance in feature
phones, Kodak's leadership in film photography, and Blockbuster's command of
physical video rental all rested on capabilities finely tuned to a
particular technological environment. When the smartphone era, digital
imaging, and high-speed internet transformed those environments, competition
was redefined along new dimensions. Although incumbents retained residual
strengths---brand recognition, distribution networks, organizational
capital---the rules of performance were sufficiently altered that
challengers could re-enter the race without having to overturn decades of
accumulated superiority. In this sense, incumbency provided a bounded
advantage within each regime, but regime shifts prevented that advantage
from accumulating without limit. The insecurity of leadership sustains
incentives both for leapfrogging entry by challengers and for continual
reinvestment by incumbents.

A similar logic arises in biological evolution. In host-parasite arms races,
an immune system may successfully neutralize a pathogen strain, thereby
establishing a temporary lead. Yet the pathogen's rapid mutation generates
new variants, effectively resetting the contest and restoring competitive
pressure. Previous immunological victories do not accumulate indefinitely;
their effectiveness is contingent on the prevailing biological environment.
Predator-prey dynamics display the same structure: a predator may evolve
speed to gain an edge, but selective pressure on prey narrows that advantage
over time.

The finches of Daphne Major provide a particularly vivid example %
\citep{grant1980annual}. During prolonged periods of soft-seed abundance,
small-beaked finches prospered. A severe drought altered seed availability,
enabling large-beaked finches to dominate. Environmental volatility thus
redefines the basis of success, ensuring that advantages remain conditional
rather than permanently compounding.

More broadly, this structure aligns with the Red Queen hypothesis %
\citep{vanvalen1973new}: Species must continually adapt merely to preserve
relative position in an ever-changing ecosystem. Persistent environmental
fluctuation sustains contestability by preventing any participant from
converting temporary superiority into irreversible dominance. In such
systems, incentives can be sustained precisely because the lead is
ultimately ephemeral. The iterated incumbency framework provides a
parsimonious way to formalize this logic: Incumbency confers a bounded
advantage within each phase of competition, but exogenous shifts repeatedly
reopen the contest, maintaining long-run rivalry and sustained investment.
The volatility in the coevolutionary process leads to sustained effort,
forcing every species to \textquotedblleft run on a
treadmill\textquotedblright\ of adaptation simply to avoid extinction. A
natural incentive structure endogenously emerges without exogenous design.

\section{Extension: General Battle Success Functions}

\label{sec:general-bsf}

Our baseline analysis assumes a homogeneous success function (\Cref{cond:bsf}%
) for expositional efficiency. We now demonstrate that our results hold for
more general success functions that admit common properties for component
battles. Our discussion proceeds in three steps. We first extract the key
properties for component battles that underpin our main results. Second, we
demonstrate that these properties also emerge for other success functions.
Third, we explain how our results extend to the broader setting with more
general success functions.

\subsection{Key Properties for Component Battles}
\label{sec:key-component-battle-properties}

Consider a component battle in the contest with a winning value $\Delta
_{\ell }>0$ for player $\ell \in \{A,B\}$. Let $x_{\ell }^{\ast }=x^{\ast
}(\Delta _{\ell },\Delta _{-\ell })$ denote the player's equilibrium effort. 
%%%Define 
%%%\begin{equation*}
%%%\pi _{\ell }^{\ast }=\pi ^{\ast }(\Delta _{\ell },\Delta _{-\ell }):=\frac{%
%%%p(x_{\ell }^{\ast },x_{-\ell }^{\ast })\Delta _{\ell }-x_{\ell }^{\ast }}{%
%%%\Delta _{\ell }}\equiv p(x_{\ell }^{\ast },x_{-\ell }^{\ast })-\frac{x_{\ell
%%%}^{\ast }}{\Delta _{\ell }}
%%%\end{equation*}%
%%%to be the ratio of player $\ell $'s equilibrium payoff in the battle to his
%%%value of winning, which we call the gain ratio function of the battle. With
%%%homogeneous success function, $\pi _{\ell }^{\ast }=\phi (\Delta _{\ell
%%%}/\Delta _{-\ell })$. 
We have the following.

\begin{lemma}[Key Properties for Component Battles]
\label{prop:gensf} Consider a single battle between two players with
respective winning values $\Delta ,\Delta ^{\prime }>0$. With homogeneous
success functions as described in \Cref{cond:bsf}, the following hold.

\begin{enumerate}[label=(\roman*)]

\item \label{cond0} The battle yields a unique pure-strategy Nash
equilibrium; the equilibrium strategy $x^{\ast }(\Delta ,\Delta ^{\prime })$
is continuous in the winning values, and $x^{\ast }(\Delta ,\Delta ^{\prime
})>0$.

\item \label{cond1} There exists a constant $C>0$ such that $\frac{x^{\ast
}(\Delta ^{\prime },\Delta )}{x^{\ast }(\Delta ,\Delta ^{\prime })}<C\frac{%
\Delta ^{\prime }}{\Delta }$ for all $\Delta ^{\prime }\geq \Delta $.
\end{enumerate}
\end{lemma}

The property laid out in \Cref{prop:gensf}\ref{cond1} requires that the
ratio of players' equilibrium efforts be bounded by a constant times the
ratio of winning values. Further, define 
\begin{equation*}
\pi _{\ell }^{\ast }=\pi ^{\ast }(\Delta _{\ell },\Delta _{-\ell }):=\frac{%
p(x_{\ell }^{\ast },x_{-\ell }^{\ast })\Delta _{\ell }-x_{\ell }^{\ast }}{%
\Delta _{\ell }}\equiv p(x_{\ell }^{\ast },x_{-\ell }^{\ast })-\frac{x_{\ell
}^{\ast }}{\Delta _{\ell }},
\end{equation*}%
which is the ratio of player $\ell $'s equilibrium payoff in the battle to
his value of winning, and we call it the battle's \emph{gain ratio function}%
. With homogeneous success function, $\pi _{\ell }^{\ast }=\phi (\Delta
_{\ell }/\Delta _{-\ell })$.

\begin{lemma}
\label{lemma:gainratio} Consider a single battle between two players with
respective winning values $\Delta ,\Delta ^{\prime }>0$. With homogeneous
success functions as described in \Cref{cond:bsf}, the gain ratio function
of the battle has the following properties.

\begin{enumerate}[label=(\roman*)]

\item \label{pi0} For all $\epsilon >0$, $\pi ^{\ast }(\Delta ^{\prime
},\Delta )$ uniformly converges to a continuous function $\pi _{0}(\Delta
^{\prime })>0$ as $\Delta \rightarrow 0^{+}$ across all $\Delta ^{\prime
}\in \lbrack \epsilon ,+\infty )$.

\item \label{pi-ratio} For all $R>0$, there exists ${\pi _{R}^{\ast }}>0$
such that if $\Delta ^{\prime }\geq R\Delta $, then $\pi ^{\ast }(\Delta
^{\prime },\Delta )\geq {\pi _{R}^{\ast }}$. Moreover, there exists $%
R^{\prime }>0$ and $\widetilde{\pi }>\frac{1}{2}$ such that if $\Delta
^{\prime }\geq R^{\prime }\Delta $, then $\pi ^{\ast }(\Delta ^{\prime
},\Delta )\geq \widetilde{\pi }$.

\item \label{pi-high} For all $R^{\prime \prime }>1$, there exists $%
d_{R^{\prime \prime }}>0$ such that if $1\leq \frac{\Delta ^{\prime }}{%
\Delta }\leq R^{\prime \prime }$, $\pi ^{\ast }(\Delta ^{\prime },\Delta
)+\pi ^{\ast }(\Delta ,\Delta ^{\prime })\leq 1-d_{R^{\prime \prime }}.$
\end{enumerate}
\end{lemma}

\Cref{lemma:gainratio} concerns the fraction of winning value that a player
can extract as his equilibrium payoff in the battle. The properties can be
intuitively interpreted. \Cref{lemma:gainratio}\ref{pi0} characterizes a
standard continuity property. \Cref{lemma:gainratio}\ref{pi-ratio}, in
words, states that for any $R>0$, if a player's winning value $\Delta
^{\prime }$ does not fall below a fraction $R$ of his opponent's winning
value $\Delta $, then his equilibrium payoff must be at least a fraction ${%
\pi _{R}^{\ast }}$ of $\Delta ^{\prime }$. Further, it states that a
player's equilibrium payoff is more than half of his winning value $%
\Delta^{\prime }$, as long as he is sufficiently strong compared to his
opponent, i.e., $\Delta ^{\prime }\geq R^{\prime }\Delta $ for a constant $%
R^{\prime }$, which can be chosen to be arbitrarily large. It is noteworthy
that $\pi ^{\ast }(\Delta ^{\prime },\Delta )+\pi ^{\ast }(\Delta ,\Delta
^{\prime })$ in \Cref{lemma:gainratio}\ref{pi-high}, as the sum of the two
players' gain ratios, is indicative of rent dissipation in the battle. The
property says that when the battle is not too unbalanced, i.e., $\Delta
^{\prime }/\Delta $ falls within $[1,R^{\prime \prime }]$, the fractions of
winning value they can extract must also be bounded, i.e., $\pi ^{\ast
}(\Delta ^{\prime },\Delta )+\pi ^{\ast }(\Delta ,\Delta ^{\prime })\leq
1-d_{R^{\prime \prime }}$.

The equilibrium properties laid out in \Cref{prop:gensf} and %
\Cref{lemma:gainratio} play critical roles in shaping our main results.

\subsection{Alternative Success Functions}

We derive the properties in \Cref{prop:gensf} and \Cref{lemma:gainratio}
from our baseline setting of homogeneous success function. However, these
properties can be satisfied in alternative contexts. For example, consider
the popularly adopted ratio-form success functions, which are specified as
follows. Given an effort profile $(x,x^{\prime })$, a player exerting an
effort $x$ wins the battle with a probability 
\begin{equation}  \label{ratio-bsf}
p(x,x^{\prime })=%
\begin{cases}
\frac{\gamma (x)}{\gamma (x)+\gamma (x^{\prime })}, & \text{ if }%
(x,x^{\prime })\neq (0,0), \\ 
\frac{1}{2}, & \text{ if }(x,x^{\prime })=(0,0),%
\end{cases}%
\end{equation}%
where $\gamma :\mathbb{R}_{+}\rightarrow \mathbb{R}_{+}$ is continuous and
twice differentiable with $\gamma (0)=0$, $\gamma ^{\prime }>0$, $\gamma
^{\prime \prime }<0$, and $\inf_{x>0}\frac{x\gamma ^{\prime }(x)}{\gamma (x)}%
>0$. We obtain the following.

\begin{proposition}
\label{lemma:bsf-ex} Ratio-form success functions have the properties in~%
\Cref{prop:gensf} and \Cref{lemma:gainratio}.
\end{proposition}

\subsection{How the results extend}

Our results largely hold as long as the success functions for component
battles allow for the properties laid out in \Cref{prop:gensf} and %
\Cref{lemma:gainratio}. For example, \Cref{thm:L} remains intact with the
upper bound $\left[ 1-{(\pi _{1}^{\ast })}^{L(\mathcal{M})}\right] v$ for
rent dissipation. Similarly, \Cref{lemma:tow-reinf}, \Cref{prop:tow-adv}, %
\Cref{thm:tow-rent}, \Cref{lemma:cw-reinf}, \Cref{prop:cw-adv}, %
\Cref{thm:cw-rent}, \Cref{thm:sufficient}, \Cref{lemma:ratio-instances}, and %
\Cref{prop:full-suff} also remain intact.

\Cref{prop:homogeneous-ex} requires an additional property: There exist $%
\widehat{\theta }>1$ and $M>0$ such that, for any single battle in which the
winning values satisfy $\frac{\Delta _{\ell }}{\Delta _{-\ell }}>\widehat{%
\theta }$, it holds that $\frac{p(x_{-\ell }^{\ast },x_{\ell }^{\ast
})\Delta _{\ell }}{p(x_{\ell }^{\ast },x_{-\ell }^{\ast })\Delta _{-\ell }}%
>M $. This additional property, together with \Cref{prop:gensf} and %
\Cref{lemma:gainratio}, reinstates \Cref{prop:homogeneous-ex}. It can be
verified that this property (and thus \Cref{prop:homogeneous-ex}) holds for
both the homogeneous and ratio-form success functions. The details are
provided in \Cref{app:assumption_3_ex}.

%%%\Cref{prop:homogeneous-ex} can be reinstated once this condition is
%%%satisfied, together with the properties in \Cref{prop:gensf} and \Cref{lemma:gainratio}. 
%%%
%%%However, we can
%%%verify that both homogeneous and ratio-form success functions accommodate
%%%this requirement. 
%%%

A caveat is worth noting for \Cref{thm:towp}. 
% Suppose that the properties in
% \Cref{prop:gensf} and \Cref{lemma:gainratio} hold. 
To extend \Cref{thm:towp} beyond homogeneous success functions, we need to
verify the existence of an MPE. This is technically challenging: We
conjecture that this equilibrium exists but cannot prove it in general.
However, assuming equilibrium existence, the transient dominance property
can still be established for tug-of-war with resets. 
% We conjecture that this equilibrium exists but
% cannot prove it in general, which should be attempted in future research.

Finally, we provide the following remark to close this section.

\begin{remark}
\label{lemma:noisier} Suppose that $p(x^{\prime },x)$ admits the properties
in \Cref{prop:gensf} and \Cref{lemma:gainratio}. Then the success function $%
q\times p(x^{\prime },x)+(1-q)\times \frac{1}{2}$ for any $q\in (0,1]$ also
has those properties.
\end{remark}

\Cref{lemma:noisier} states that the properties in \Cref{prop:gensf} and %
\Cref{lemma:gainratio} are robust to introducing additional noise into the
success function by mixing it with a fair coin toss. Consequently, our main
results also continue to hold for such variations in the battle success
function.

\section{Conclusion}

This paper develops a general analysis of dynamic multi-battle contests to
explore how contest architecture governs rent dissipation. We first show
that a contest cannot fully dissipate its rent as long as its minimum length
remains finite. Further, we establish in the case of tug-of-war contests
that full rent dissipation may not arise asymptotically even if the minimum
length of the contest approaches infinity.

We identify a structural property common in many dynamic contests,
exchangeability, which plays a critical role in limiting rent dissipation in
equilibrium. Exchangeability enables each battle outcome to exert a
long-lasting impact along the dynamics. The long memory of each battle
cements one's lucky early wins into persistent advantage, which, in turn,
discourages the opponent and disincentivizes both players.

Further, we propose a transient dominance property that defies the long
memory caused by exchangeability and establish that dynamic contests
exhibiting such a property approach full rent dissipation in the limit. An
iterated incumbency contest model is then proposed to demonstrate the
principle of transient dominance, which provides a natural and intuitive
account of dynamic competitions with random exogenous shocks, such as
Schumpeterian competition with technological disruptions and coevolutionary
processes in nature.

Our main analysis is set up in two-player settings. We conjecture that the
qualitative insights---e.g., the roles played by exchangeability and
transient dominance---may extend to multiplayer contests, but a thorough
investigation is left for future research.

\clearpage

\bibliographystyle{apalike-fullnames}
\bibliography{multi-battle}

\clearpage
\begin{appendices}
\appendix

\label{app:proofs}

\renewcommand{\thetable}{A\arabic{table}}
\setcounter{table}{0}
\renewcommand{\thefigure}{A\arabic{figure}}
\setcounter{figure}{0}
\setcounter{equation}{0}
\renewcommand{\theequation}{A\arabic{equation}}

\section{Single-battle Properties}\label{app:single}
In this appendix, we establish single-battle properties (i.e., \Cref{prop:gensf}, \Cref{lemma:gainratio}, and \Cref{assumption_3_ex}) under the homogeneous (and also the ratio-form) success function. Proofs of \Cref{lemma:bsf-ex} and \Cref{lemma:noisier} are also included here. Then in \Cref{sec:proof}, we prove the results for a general success function based on these single-battle properties wherever possible.

\subsection{Proof of \Cref{prop:gensf} and \Cref{lemma:gainratio}}
\begin{proof}
Suppose that the battle success function is homogeneous as given by \Cref{cond:bsf}.
By the FOC, we have $x_A = \Delta_B\gamma'(\frac{\Delta_B}{\Delta_A})$ and $x_B = \Delta_A\gamma'(\frac{\Delta_A}{\Delta_B})$ in equilibrium, so \Cref{prop:gensf}\ref{cond0} is satisfied. 
    Taking derivative of $\gamma(x)+\gamma(1/x) = 1$ yields that $\gamma'(x) x^2 = \gamma'(1/x)$. As a result,
    \begin{align*}
        \frac{x_A}{x_B} = \frac{\Delta_A}{\Delta_B}\frac{\frac{\Delta_B}{\Delta_A}\gamma'\left(\frac{\Delta_B}{\Delta_A}\right)}{\frac{\Delta_A}{\Delta_B}\gamma'\left(\frac{\Delta_A} {\Delta_B}\right)} = \frac{\Delta_A}{\Delta_B},
    \end{align*} and \Cref{prop:gensf}\ref{cond1} is naturally satisfied. 
    
    For \Cref{lemma:gainratio}, we have
    \begin{align*}
        \pi^*(\Delta',\Delta) = \phi\left(\frac{\Delta'}{\Delta}\right) :=\gamma\left(\frac{\Delta'}{\Delta}\right) -  \frac{\Delta'}{\Delta}\gamma'\left(\frac{\Delta'}{\Delta}\right),
    \end{align*}
    where $\phi$ is continuous and strictly increasing, with $\phi(0) = 0$, $\phi(1)\in (0,\frac{1}{2})$, and $\lim_{x\to +\infty}\phi(x) = 1$. Parts~\ref{pi0} and~\ref{pi-ratio} follow immediately. For part~\ref{pi-high}, since $1-\phi(x) - \phi(1/x)$ is a continuous and positive function on the closed interval $[1,R'']$, letting $d_{R''} := \min_{x\in [1,R'']}\{1-\phi(x) - \phi(1/x)\}$ finishes the proof.\end{proof}

\subsection{Proof of \Cref{lemma:bsf-ex}}
\begin{proof}
Suppose that the battle success function is given by \eqref{ratio-bsf}.
  \Cref{prop:gensf}\ref{cond0} is satisfied because the equilibrium is determined by the FOC:
    \begin{align}
        \frac{\gamma'(x_A)\gamma (x_B)}{(\gamma(x_A)+\gamma(x_B))^2}\Delta_A = \frac{\gamma'(x_B)\gamma (x_A)}{(\gamma(x_A)+\gamma(x_B))^2}\Delta_B = 1. \label{FOC_ratioform}
    \end{align}
    
    For \Cref{prop:gensf}\ref{cond1},  since $\gamma(\cdot)$ is concave and $\inf_{x>0}\frac{x\gamma'(x)}{\gamma(x)}>0$, there exists constant $C>1$ such that $\frac{x\gamma'(x)}{\gamma(x)}\in[1/C,1]$. Therefore,
    \begin{align}
        \frac{x_A}{x_B} = \underbrace{\frac{x_A\gamma'(x_A)}{\gamma(x_A)}}_{\text{bounded in }[1/C,1]}
        \underbrace{
         \frac{\gamma(x_A)\gamma'(x_B)}{\gamma(x_B)\gamma'(x_A)}}_{=\frac{\Delta_A}{\Delta_B}\text{ by FOC}}
         \underbrace{\frac{\gamma(x_B)}{x_B\gamma'(x_B)}}_{\text{bounded in }[1,C]} \in  \left[\frac{1}{C}\frac{\Delta_A}{\Delta_B},C\frac{\Delta_A}{\Delta_B}\right]. \label{lemma1(i)}
    \end{align}
    
    For \Cref{lemma:gainratio}, we first show that 
    $\pi^*(\Delta',\Delta)\in \left[p(x',x)^2,p(x',x)\right]$.
    Let $x' = x^*(\Delta',\Delta), x= x^*(\Delta,\Delta')$, we have that
    \begin{align*}
        \pi^*(\Delta',\Delta)\quad & = p(x',x)- \frac{x'}{\Delta'} = \frac{\gamma(x')}{\gamma(x')+\gamma(x)} - \frac{x'}{\Delta'}\\
        \quad & \underbrace{=}_{\text{FOC}} \frac{\gamma(x')}{\gamma(x')+\gamma(x)}  - \frac{x'\gamma'(x')\gamma(x)}{(\gamma(x)+\gamma(x'))^2} \\
        \quad & = \frac{\gamma(x')^2 + \gamma(x)\gamma(x') -x'\gamma'(x')\gamma(x)}{(\gamma(x)+\gamma(x'))^2} \underbrace{\geq}_{\gamma(x')\geq x'\gamma'(x') } p(x',x)^2.
    \end{align*}

    When $\Delta'\geq\Delta$, by \eqref{FOC_ratioform}, $x'\geq x$, then it follows that
    \begin{equation}
    \begin{split}
        \log\frac{\gamma(x')}{\gamma(x)} = \int_{x}^{x'} \frac{\gamma'(t)}{\gamma(t)}dt \in \left[\int_{x}^{x'} \frac{1}{Ct}dt,\int_{x}^{x'} \frac{1}{t}dt\right] \underbrace{\subseteq}_{\text{by \eqref{lemma1(i)}}} \left[\frac{1}{C}\log\left(\frac{1}{C}\frac{\Delta'}{\Delta}\right), \log\left(C\frac{\Delta'}{\Delta}\right)\right].\label{control_R}
        \end{split}
    \end{equation}
Let $\pi^*_R:=\left[G\left(\frac{1}{C}\log\left(\frac{1}{C}R\right)\right)\right]^2$, where $G(z) = \frac{e^z}{1+e^z}$.
Then we have $\pi^*(\Delta',\Delta)\geq [p(x',x)]^2\geq\pi^*_R$. Moreover, since $\pi^*_R\to1$ as $R\to+\infty$, parts~\ref{pi0} and~\ref{pi-ratio} are proved. For part~\ref{pi-high}, notice that
    \begin{align*}
        1-\pi^*(\Delta',\Delta)-\pi^*(\Delta,\Delta') = \frac{x'\gamma'(x')\gamma(x)+ x\gamma'(x)\gamma(x')}{(\gamma(x)+\gamma(x'))^2} \geq \frac{1}{C}p(x,x')p(x',x).
    \end{align*}
Since $p(x,x')p(x',x)\geq (1-G(\log (CR)))G(\log (CR))$ by \eqref{control_R}, the proof is completed.\end{proof}

%%%\begin{lemma}
%%%\label{assumption_3_ex} Both the homogeneous and ratio-form success
%%%functions admit the following single-battle property: There exist $%
%%%\widehat{\theta }>1$ and $M>0$ such that $\frac{p(x_{-\ell }^{\ast
%%%},x_{\ell }^{\ast })\Delta _{\ell }}{p(x_{\ell }^{\ast },x_{-\ell }^{\ast
%%%})\Delta _{-\ell }}>M$ for all $\frac{\Delta _{\ell }}{\Delta _{-\ell }}>%
%%%\widehat{\theta }$.
%%%\end{lemma}
\subsection{Proof of \Cref{assumption_3_ex}}\label{app:assumption_3_ex}
%%%\begin{lemma}
%%%Both Ratio-form and Homogeneous success functions satisfy Assumption~\ref{assumption_3_ex}. 
%%%\end{lemma}
\begin{lemma}
\label{assumption_3_ex} Both the homogeneous and ratio-form success
functions admit the following single-battle property: There exist $%
\widehat{\theta }>1$ and $M>0$ such that $\frac{p(x_{-\ell }^{\ast
},x_{\ell }^{\ast })\Delta _{\ell }}{p(x_{\ell }^{\ast },x_{-\ell }^{\ast
})\Delta _{-\ell }}>M$ for all $\frac{\Delta _{\ell }}{\Delta _{-\ell }}>%
\widehat{\theta }$.
\end{lemma}
\begin{proof}
\textbf{Ratio-form SF.} For any $M\in (0,1)$, suppose that $\Delta_{\ell}>\Delta_{-\ell}$. If $x^*_{\ell}>x^*_{-\ell}$, we have that $\gamma'(x^*_{\ell})<\gamma'(x^*_{-\ell})$ and
\begin{align*}
\frac{p(x^*_{-\ell},x^*_{\ell})\Delta_{\ell}}{p(x^*_{\ell}, x^*_{-\ell})\Delta_{-\ell}}
&= \frac{\gamma(x^*_{-\ell})\Delta_{\ell}}{\gamma(x^*_{\ell})\Delta_{-\ell}}
\underbrace{=}_{\text{FOC}}
\frac{\gamma(x^*_{-\ell})\gamma(x^*_{\ell})\gamma'(x^*_{-\ell})}
{\gamma(x^*_{\ell})\gamma'(x^*_{\ell})\gamma(x^*_{-\ell})} \\
&= \frac{\gamma'(x^*_{-\ell})}{\gamma'(x^*_{\ell})}
> 1 > M.
\end{align*}
Otherwise if $x_\ell^*\leq x^*_{-\ell}$, it is clear that 
\[\frac{p(x^*_{-\ell},x^*_{\ell})\Delta_{\ell}}{p(x^*_{\ell}, x^*_{-\ell})\Delta_{-\ell}}\geq
\frac{\Delta_{\ell}}{\Delta_{-\ell}}\geq1>M.
\]

\textbf{Homogeneous SF.} Let $\theta=\frac{\Delta_\ell}{\Delta_{-\ell}}$, then
in this case $\frac{p(x^*_{-\ell},x^*_{\ell})\Delta_{\ell}}{p(x^*_{\ell}, x^*_{-\ell})\Delta_{-\ell}}$
becomes $\frac{\theta (1-\phi(\theta))}{1-\phi(1/\theta)}$.
We have that
\begin{align*}
\frac{\theta (1-\phi(\theta))}{1-\phi(1/\theta)}
&= \frac{\theta\bigl(1-\gamma(\theta)+\theta \gamma'(\theta)\bigr)}
{1-\gamma\!\left(\frac{1}{\theta}\right)+\frac{1}{\theta}\gamma'\!\left(\frac{1}{\theta}\right)}
\underbrace{=}_{w=1/\theta}
\frac{\gamma(w)+w\gamma'(w)}{w\bigl(1-\gamma(w)+w\gamma'(w)\bigr)}.
\end{align*}
Carrying out algebra, $\frac{\theta (1-\phi(\theta))}{1-\phi(1/\theta)}>M$ is equivalent to
\begin{align}\label{ex-eq0}
\frac{\gamma(w)}{w}(1+Mw) + \gamma'(w)(1-Mw) > M.
\end{align}
Let $M=\gamma'(1)$. For any $w\in\left(0,\min\{1,\frac{1}{2M}\}\right)$, concavity implies $\frac{\gamma(w)}{w}\ge \gamma'(w)$ and monotonicity of $\gamma'$ implies $\gamma'(w)>\gamma'(1)=M$. Hence
\[
\frac{\gamma(w)}{w}(1+Mw)+\gamma'(w)(1-Mw)
\;\ge\;
\frac{\gamma(w)}{w}+\frac12\gamma'(w)
\;\ge\;
\frac32\gamma'(w)
\;>\;
\frac32 M,
\]
so \eqref{ex-eq0} holds for all sufficiently small $w$, completing the proof.\end{proof}

\subsection{Proof of \Cref{lemma:noisier}}
\begin{proof}
    Let $p'(x',x) = q\times p(x',x)+ (1-q)\times \frac{1}{2}$ for $q\in (0,1]$.
    Suppose $p$ satisfies \Cref{prop:gensf} and \Cref{lemma:gainratio}.
    We proceed to prove that $p'$ also satisfies \Cref{prop:gensf} and \Cref{lemma:gainratio}.
     For the winning value profile $(\Delta_A,\Delta_B)$, notice that
    \begin{align}
     \arg\max_{x_\ell}\{p'(x_\ell,x_{-\ell})\Delta_\ell - x_\ell\}
     \;\Leftrightarrow\;
     \arg\max_{x_\ell}\{p(x_\ell,x_{-\ell})q\Delta_\ell - x_\ell\}
\label{eq:argmax}
   \end{align}
    Therefore, the equilibrium strategy under $p'$ is $x_{p'}^*(\Delta',\Delta)=x_{p}^*(q\Delta',q\Delta)$, and \Cref{prop:gensf}\ref{cond0} naturally holds. 
    Moreover, we have
    \begin{align*}
        \frac{x_{p'}^*(\Delta',\Delta)}{x_{p'}^*(\Delta,\Delta')} = \frac{x_{p}^*(q\Delta',q\Delta)}{x_{p}^*(q\Delta,q\Delta')} < C \frac{q\Delta'}{q\Delta} = C \frac{\Delta'}{\Delta},
    \end{align*}
    which proves \Cref{prop:gensf}\ref{cond1}. \Cref{lemma:gainratio} follows immediately from the fact that
    \begin{align*}
        \pi_{p'}^*(\Delta',\Delta) = \frac{1-q}{2}+  q\pi_{p}^*(\Delta',\Delta).
    \end{align*}
%%%    For the details, let $\pi_0(\Delta')_{p'} = \frac{1-q}{2}+q\pi_0(\Delta')_{p}$, $\pi^*_{R,p'} = \frac{1-q}{2}+q\pi^*_{R,p}$, $\tpi_{p'} = \frac{1-q}{2}+q\tpi_{p}$, $d_{R'',p'}= qd_{R'',p}$.
\end{proof}

\section{Proofs for Main Results}
\label{sec:proof}
In this appendix, we prove the results for a general success function that admits single-battle properties given in \Cref{prop:gensf} and \Cref{lemma:gainratio}, except for the uniqueness part of \Cref{prop:tow-mpe}, \Cref{prop:homogeneous-ex}, and \Cref{thm:towp}: 
We prove \Cref{prop:homogeneous-ex} with \Cref{assumption_3_ex} in addition to \Cref{prop:gensf,lemma:gainratio}; and
we prove the uniqueness part of \Cref{prop:tow-mpe} and \Cref{thm:towp} for the homogeneous success function (i.e., under \Cref{cond:bsf}). 

\subsection{Proof of \Cref{thm:L} Based on \Cref{prop:gensf,lemma:gainratio}}
\begin{proof}
%%%In a single battle with continuation values (V_A^A>V_A^B)
For a single battle in a dynamic contest, we denote by $V_\ell^{\ell'}$ the continuation value of player $\ell\in{A,B}$ if player $\ell'$ wins this battle. It is useful to prove the following intermediate result. 
    \begin{lemma}\label{lemma:battleshare}
    Consider a single-battle contest with $V_A^A>V_A^B$ and $V_B^B>V_B^A$. The equilibrium payoffs in the single-battle contest $(U_A^*,U_B^*)$ satisfy 
    \begin{equation}\label{eqn:1}
        U_A^*+U_B^*>{\pi^*_1}\max\{V_A^A+V_B^A,V_A^B+V_B^B\}.
    \end{equation}
\end{lemma}
\begin{proof}
Without loss of generality, we assume that $\frac{V_A^A-V_A^B}{V_B^B-V_B^A}\geq1$. Note that the value of winning is $\Delta_A=V_A^A-V_A^B$ for player $A$ and $\Delta_B=V_B^B-V_B^A$ for player $B$. 
We have that 
\[
\begin{split}
    U_A^*+U_B^*&=V_A^B+(V_A^A-V_A^B)\pi_A^*+V_B^A+(V_B^B-V_B^A)\pi_B^*\\
    &=(V_A^A+V_B^A)\pi_A^*+(V_A^B+V_B^B)\pi_B^*+\left(1-\pi_A^*-\pi_B^*\right)(V_A^B+V_B^A).
\end{split}
\]
Note that $\pi_A^*+\pi_B^*=p_A^*+p_B^*-\frac{x_A^*}{\Delta_A}-\frac{x_B^*}{\Delta_B}\leq1$. 
Further, since $\frac{V_A^A-V_A^B}{V_B^B-V_B^A}\geq1$, $V_A^A+V_B^A\geq V_A^B+V_B^B$. 
As a result, 
$ U_A^*+U_B^*> (V_A^A+V_B^A)\pi_A^*=\max\{V_A^A+V_B^A,V_A^B+V_B^B\}\pi_A^*$.
Since $\Delta_A\geq\Delta_B$, by \Cref{lemma:gainratio}\ref{pi-ratio}, $\pi_A^*\geq{\pi^*_1}$. This completes the proof. 
\end{proof}
\Cref{lemma:battleshare} means that the expected aggregate payoff in a battle is higher than a factor ${\pi^*_1}$ times the maximum of aggregate continuation payoff across both possible outcomes. Applying \Cref{lemma:battleshare} repeatedly along the shortest path to a terminal node (which has aggregate payoff $v$) implies that the initial aggregate payoff is no less than ${(\pi^*_1)}^{L(\mathcal M)}v$. Therefore, the expected total effort is less than $\left[1-{(\pi^*_1)}^{L(\mathcal M)}\right]v$. 
\end{proof}

\subsection{Proof for the Existence Part of \Cref{prop:tow-mpe} Based on \Cref{prop:gensf,lemma:gainratio}}
\begin{proof}
Our goal is to construct a mapping from a set of values $(v_i)_{-N\leq i\leq N}$ to a new set of values, such that the fixed point of this mapping satisfies \eqref{eqn:mpe}, and then prove that a fixed point exists for this mapping. 

First, we define the following function $\Pi^*(\Delta',\Delta):[0,v]^2\to[0,v]$, which augments the equilibrium gain function for a single battle in a natural way: 
\begin{equation}\label{def:payoff}
\Pi^*(\Delta',\Delta):=
\begin{cases}
p(x^*(\Delta',\Delta),x^*(\Delta,\Delta'))\Delta'-x^*(\Delta',\Delta),&\text{ if }\Delta',\Delta>0,\\
\pi_0(\Delta')\Delta',&\text{ if }\Delta=0.
\end{cases}
\end{equation}
By \Cref{lemma:gainratio}\ref{pi0}, $\Pi^*(\Delta',\Delta)$ is continuous on $[0,v]^2$. Moreover, it can be verified that the function $\bm\Pi^*(\Delta',\Delta):=(\Pi^*(\Delta',\Delta),\Pi^*(\Delta,\Delta'))$ is continuous on $[0,v]^2$. 

Next, we proceed to define a mapping from $\mathcal V:=\{\bm v:=(v_i)_{i=-(N-1)}^{N-1}:0\leq v_{-N+1}\leq v_{-N+2}\leq\ldots\leq v_{N-1}\leq v\text{ and }v_i+v_{-i}\leq v,\text{ for all }-(N-1)\leq i\leq N-1\}$ to itself. Let $v_{-N}=0$, $v_N=v$, and $\bm\Phi:\mathcal V\to\mathcal V$. $\bm v'=\bm\Phi(\bm v)$ is recursively given by the following, with initial values $v'_N=v$ and $v'_{-N}=0$:
\begin{align}
\begin{pmatrix}
v_i'\\v_{-i}'
\end{pmatrix}
&=
\begin{pmatrix}
v_{i-1}\\ 
v'_{-i-1}
\end{pmatrix}
+\bm\Pi^*(v'_{i+1}-v_{i-1},v_{-i+1}-v'_{-i-1}),\text{ for }i=N-1,\ldots,1,
\label{eqn:phi1}
\\
v_0'&=v_{-1}'+\Pi^*(v'_1-v'_{-1},v'_1-v'_{-1}).
\label{eqn:phi2}
\end{align}
It is easy to see that the any fixed point of $\bm\Phi$ corresponds to the value function in a symmetric MPE of the tug-of-war, and vice versa. 

We now verify that $\bm \Phi$ is well-defined. 
First, for $i=N-1$, it is clear that $v'_{i+1}-v_{i-1}\in[0,v]$ and $v_{-i+1}-v'_{-i-1}\in[0,v]$. 
Suppose that $v'_{i+1}-v_{i-1}\in[0,v]$ and $v_{-i+1}-v'_{-i-1}\in[0,v]$ hold for $i=i'\in\{1,\ldots,N-1\}$,
we show below that they also hold for $i=i'-1$. 
Since $0\leq\Pi^*(\Delta',\Delta)\leq \Delta'$ for all $(\Delta',\Delta)\in[0,v]^2$, 
it follows from \eqref{eqn:phi1} that $v'_{i'+1}\geq v_{i'}'\geq v_{i'-1}$.
This implies that $v \geq v_{i'}'\geq v_{i'-2}$, and $v'_{i+1}-v_{i-1}\in[0,v]$ holds for $i=i'-1$.
Again, since $\Pi^*(\Delta',\Delta)\leq \Delta'$ for all $(\Delta',\Delta)\in[0,v]^2$, it follows from \eqref{eqn:phi1} that $v_{-i'}'\leq v_{-i'+1}\leq v_{-i'+2}$, so $v_{-i+1}-v'_{-i-1}\in[0,v]$ holds for $i=i'-1$. Finally, when $i=1$, \eqref{eqn:phi1} implies that $v_2'\geq v_1'\geq v_0\geq v_{-1}'$, so $v'_1-v'_{-1}\in[0,v]$. 

Then we verify that $\Phi(\bm v)\in\mathcal V$. From \eqref{eqn:phi1}, it is easy to see that $v_i'\leq v_{i+1}'$ and $v_{-i}'\geq v_{-i-1}'$ for $i=N-1,\ldots,1$. We have shown above that $v_1'\geq v_0\geq v_{-1}'$. Therefore, $v\geq v_{N-1}'\geq\ldots\geq v_1'\geq v_0\geq v_{-1}'\geq\ldots\geq v_{-N+1}'\geq0$. By \eqref{eqn:phi2}, $v_1'\geq v_0'\geq v_{-1}'$, and thus $1\geq v_{N-1}'\geq\ldots\geq v_1'\geq v_0'\geq v_{-1}'\geq\ldots\geq v_{-N+1}'\geq0$. 
Note that $\Pi^*(\Delta',\Delta)+\Pi^*(\Delta,\Delta')\leq\max\{\Delta',\Delta\}$. By \eqref{eqn:phi2}, for $1\leq i\leq (N-1)$,
\[
\begin{split}
v_i'+v_{-i}'&\leq v_{i-1}+v'_{-i-1}+\max\{v'_{i+1}-v_{i-1},v_{-i+1}-v'_{-i-1}\}\\
&=\max\{v'_{i+1}+v'_{-i-1},v_{-i+1}+v_{i-1}\}\leq\max\{v'_{i+1}+v'_{-i-1},v\}.
\end{split}
\]
Applying the inequality iteratively yields that $v_i'+v_{-i}'\leq \max\{v'_{N}+v'_{-N},v\}=v$. Finally, by \eqref{eqn:phi2}, 
$2v_0'\leq 2v_{-1}'+\max\{v'_1-v'_{-1},v'_1-v'_{-1}\}=v_{-1}'+v_1'\leq v$. 

Since $\mathcal V$ is non-empty, convex, and compact, and $\bm\Phi$ is continuous due to the continuity of $\bm\Pi^*$, by Brouwer's fixed-point theorem, there exists $\bm v^*\in\mathcal V$ such that $\bm v^*=\bm \Phi(\bm v^*)$.
%%%It is straightforward to verify that 

It only remains to show that $0<v_{-N+1}^*<\ldots<v_{N-1}^*<v$. 
Suppose that there exists $-N+1\leq i'\leq N-1$ such that $v^*_{i'-1}=v_{i'}^*<v_{i'+1}^*$. Then $v_{i'}^*=v^*_{i'-1}+\Pi^*(v_{i'+1}^*-v_{i'-1}^*,v_{-i'+1}^*-v_{-i'-1}^*)>v_{i'-1}^*$, a contradiction. Therefore, if $v_i^*=v_{i+1}^*$ for some $i$, then $v_i^*=v_j^*$ for all $j> i$. Suppose that $i'$ is the smallest integer satisfying $v_{i'}^*=v_{i'+1}^*$, then 
\begin{equation}
\label{eqn:exist-contradiction}
0=v_{-N}^*<v_{-N+1}^*<\ldots<v_{i'}^*=\ldots=v_{N}^*=v.
\end{equation}
 If $i'>0$, from $v_{i'}^*=v_{i-1}^*+\Pi^*(v_{i'+1}^*-v_{i'-1}^*,v_{-i'+1}^*-v_{-i'-1}^*)=v_{i'+1}^*$, it follows that $v_{-i'+1}^*=v_{-i'-1}^*$, contradicting \eqref{eqn:exist-contradiction}.
 If $i'\leq0$, $v^*_{i'}+v_{-i'}^*=2v$, contradicting $v_{-i}^*+v_i^*\leq v$.\end{proof}

\subsection{Proof for the Uniqueness Part of \Cref{prop:tow-mpe} under \Cref{cond:bsf}}
\begin{proof}
See the proof of \Cref{thm:towp} in \Cref{sec:towp-proof}.
\end{proof}

\subsection{Proof of \Cref{lemma:tow-reinf} Based on \Cref{prop:gensf,lemma:gainratio}}
\begin{proof}
For $-(N-1)\leq i\leq N-1$, the equilibrium condition \eqref{eqn:mpe} can be rearranged in the following two ways: 
\begin{align}
V(i) - V(i-1) &= \pi(i)\,\bigl[V(i+1) - V(i-1)\bigr], \label{eqn:tow-eqm-1}\\
V(i+1) - V(i) &= \bigl[1-\pi(i)\bigr]\,\bigl[V(i+1) - V(i-1)\bigr]. \label{eqn:tow-eqm-2}
\end{align}
Replacing $i$ with $-i$ in \eqref{eqn:tow-eqm-1} yields
\begin{equation}
V(-i) - V(-i-1) = \pi(-i)\,\bigl[V(-i+1) - V(-i-1)\bigr], \label{eqn:tow-eqm-3}
\end{equation}
and dividing \eqref{eqn:tow-eqm-2} by \eqref{eqn:tow-eqm-3} yields
\begin{equation}
\frac{V(i+1) - V(i) }{V(-i) - V(-i-1)} = \frac{1-\pi(i)}{\pi(-i)}\frac{V(i+1) - V(i-1)}{V(-i+1) - V(-i-1)}. \label{eqn:tow-eqm-4}
\end{equation}
Replacing $i$ with $-(i+1)$ in \eqref{eqn:tow-eqm-4}, we have that 
\begin{equation}
\frac{V(-i) - V(-i-1) }{V(i+1) - V(i)} = \frac{1-\pi(-i-1)}{\pi(i+1)}\frac{V(-i) - V(-i-2)}{V(i+2) - V(i)}\text{ for }-(N-2)\leq i\leq N-2. \label{eqn:tow-eqm-5}
\end{equation}
Finally, multiplying \eqref{eqn:tow-eqm-4} and \eqref{eqn:tow-eqm-5} and simple rearrangement give \eqref{eqn:tow-win-reinf}.
\end{proof}

\subsection{Proof of \Cref{prop:tow-adv} Based on \Cref{prop:gensf,lemma:gainratio}}
\begin{proof}
It is useful to establish the following intermediate result first. 
\begin{lemma}\label{lemma:int-N}
There exists a positive integer $N'$ such that for all $v>0$, $i\geq N'$ and $N>i$, $\pi(i;N,v)>\tpi>\frac12$. 
\end{lemma}
\begin{proof}
By \Cref{lemma:tow-reinf}, $\frac{\Delta(i)}{\Delta(-i)}$ strictly increases in $i$ with $\frac{\Delta(0)}{\Delta(-0)}=1$. 
For all $i$ such that $\frac{\Delta(i)}{\Delta(-i)}\leq R'$, by \Cref{lemma:gainratio}\ref{pi-high}, $\frac{1-\pi(i)}{\pi(-i)}\geq 1+d_{R'}$. 
This and \eqref{eqn:tow-win-reinf} imply that $\frac{\Delta(i)}{\Delta(-i)}\geq(1+d_{R'})^i$ for all $i\geq0$ satisfying $\frac{\Delta(i)}{\Delta(-i)}\leq R'$. Therefore, there exists a positive integer $N'$ (independent of $N$ and $v$) such that $\frac{\Delta(i)}{\Delta(-i)}\geq R'$ for all $i\geq N'$. 
%%%The selection of $N'$ is unrelated to $v$, $N$.
We know from \Cref{lemma:gainratio}\ref{pi-ratio} that $\pi(i)\geq\tpi$ for $i\geq N'$. 
\end{proof}
Dividing \eqref{eqn:tow-eqm-1} by \eqref{eqn:tow-eqm-2}, we have $\frac{V(i) - V(i-1)}{V(i+1) - V(i)} = \frac{\pi(i)}{1-\pi(i)}$, which can be rewritten as $V(i+1) - V(i)=\frac{1-\pi(i)}{\pi(i)}[V(i) - V(i-1)]$. Therefore, 
\[
v-V(i)=\sum_{j=i}^{N-1} [V(j+1)-V(j)]=[V(-N+1)-V(-N)]\sum_{j=i}^{N-1} \left[\prod_{k=-N+1}^{j} \frac{1-\pi(k)}{\pi(k)}\right],
\]
which yields that
%%\begin{align}
%%    \frac{V(i)}{v}\quad & = \frac{\sum_{j=-N}^{i-1}(V(j+1)-V(j))}{\sum_{j=-N}^{N-1}(V(j+1)-V(j))} \notag\\
%%    \quad & = \frac{\sum_{j=-N}^{i-1}\left[(V(1-N)-V(N))\Pi_{k=1-N}^j \frac{1-\pi(k)}{\pi(k)}\right]}{\sum_{j=-N}^{N-1}\left[(V(1-N)-V(N))\Pi_{k=1-N}^j \frac{1-\pi(k)}{\pi(k)}\right]} \notag\\
%%    \quad & = \frac{\sum_{j=-N}^{i-1}\left[\Pi_{k=1-N}^j \frac{1-\pi(k)}{\pi(k)}\right]}{\sum_{j=-N}^{N-1}\left[\Pi_{k=1-N}^j \frac{1-\pi(k)}{\pi(k)}\right]}\label{expression_valuation}
%%\end{align}
%%It follows from \eqref{expression_valuation} that
\begin{align}
    \frac{v-V(i+1)}{v-V(i)} =
     \frac{\sum_{j=i+1}^{N-1} \left[\prod_{k=-N+1}^{j} \frac{1-\pi(k)}{\pi(k)}\right]}{\sum_{j=i}^{N-1} \left[\prod_{k=-N+1}^{j} \frac{1-\pi(k)}{\pi(k)}\right]} 
     = \frac{\sum_{j=i+1}^{N-1}\left[\Pi_{k=i+1}^j \frac{1-\pi(k)}{\pi(k)}\right]}{1+\sum_{j=i+1}^{N-1}\left[\Pi_{k=i+1}^j \frac{1-\pi(k)}{\pi(k)}\right]}.\label{expression_valuation2}
\end{align}
By \Cref{lemma:int-N}, for $i>N'$, we have that
\begin{align}\label{expression-ratio1}
    \frac{\sum_{j=i+1}^{N-1}\left[\prod_{k=i+1}^j \frac{1-\pi(k)}{\pi(k)}\right]}{1+\sum_{j=i+1}^{N-1}\left[\prod_{k=i+1}^j \frac{1-\pi(k)}{\pi(k)}\right]} \leq \frac{\sum_{j=1}^{+\infty}(\frac{1-\tpi}{\tpi})^j}{1+\sum_{j=1}^{+\infty}(\frac{1-\tpi}{\tpi})^j}=\frac{1-\tilde{\pi}}{\tilde{\pi}}.
\end{align}
Consequently,
\begin{align}\label{expression-ratio2}
    \frac{v-V(N'+i)}{v}=  \frac{v-V(N')}{v} \prod_{j = N'}^{N'+i-1} \frac{v-V(j+1)}{v-V(j)} \leq  \frac{v-V(N')}{v} \left(\frac{1-\tilde{\pi}}{\tilde{\pi}}\right)^i,
\end{align}
which implies that $ \frac{v-V(N'+i)}{v}\to0$ as $i\to+\infty$
holds uniformly for $N,v$ and over all symmetric MPE.
Notice that $V(N'+i;N,v) \leq  vQ(N'+i;N,v)$. It then follows that
\begin{align*}
    Q(N'+i;N,v) \geq \frac{V(N'+i)}{v} \to 1
\end{align*}
uniformly for any $v, N$ and over all symmetric MPE.
\end{proof}

\subsection{Proof of \Cref{thm:tow-rent} Based on \Cref{prop:gensf,lemma:gainratio}}

\begin{proof}
Using \eqref{expression_valuation2} with $i = -1$, we have that
\begin{align}
    \frac{v-V(0)}{v}<\frac{v - V(0)}{v-V(-1)} = \frac{\sum_{j=0}^{N-1}\left[\prod_{k=0}^j \frac{1-\pi(k)}{\pi(k)}\right]}{1+\sum_{j=0}^{N-1}\left[\prod_{k=0}^j \frac{1-\pi(k)}{\pi(k)}\right]} \label{eq:V(0)}.
\end{align} 
By \Cref{lemma:int-N} and \Cref{lemma:gainratio}\ref{pi-ratio}, the following estimation holds.
\begin{align}
    \sum_{j=0}^{N-1}\left[\prod_{k=0}^j \frac{1-\pi(k)}{\pi(k)}\right]\quad & 
     \leq \sum_{j=0}^{N'-1}\left(\prod_{k=0}^j \frac{1-{\pi^*_1}}{{\pi^*_1}}\right) + \prod_{k=0}^{N'-1} \left(\frac{1-{\pi^*_1}}{{\pi^*_1}} \right)\sum_{j=N'}^{+\infty} \left( \prod_{k=N'}^{j}\frac{1-\tpi}{\tpi}\right) \notag\\
    \quad & \leq \sum_{j=0}^{N'-1}\left(\prod_{k=0}^j \frac{1-{\pi^*_1}}{{\pi^*_1}}\right) + \left(\frac{\frac{1-\tpi}{\tpi}}{1-\frac{1-\tpi}{\tpi}}\right)\prod_{k=0}^{N'-1} \left(\frac{1-{\pi^*_1}}{{\pi^*_1}} \right).
    \label{eq:V(0)_est}
\end{align}
Plug \eqref{eq:V(0)_est} into \eqref{eq:V(0)}, we have that 
\begin{align*}
    \frac{V(0)}{v}\geq \left[1+\sum_{j=0}^{N'-1}\left(\prod_{k=0}^j \frac{1-{\pi^*_1}}{{\pi^*_1}}\right) + \left(\frac{\frac{1-\tpi}{\tpi}}{1-\frac{1-\tpi}{\tpi}}\right)\prod_{k=0}^{N'-1} \left(\frac{1-{\pi^*_1}}{{\pi^*_1}} \right)\right]^{-1} =: \alpha,
\end{align*}
which finishes the proof.
\end{proof}

\subsection{Proof of \Cref{prop:homogeneous-ex}  Based on Lemmas \ref{prop:gensf}, \ref{lemma:gainratio}, and \ref{assumption_3_ex}}
\begin{proof}
To prove this result, we show that there exists $\widetilde\alpha>0$ depending only on the success function such that for every sequence $\{M_k\}_{k=1}^K$ of exchangeable contests with $M_{k+1}$ being the $N_{k+1}$-extension of $M_k$,
\[
\frac1K\sum_{k=1}^K V_{0;k}\ge \widetilde\alpha v.
\]
Then, when $\lim_{k\to+\infty} V_{0;k}$ exists, it coincides with $\lim_{K\to+\infty}\frac1K\sum_{k=1}^K V_{0;k}$, and thus is also bounded from below.

Consider a sequence of exchangeable contests $\{M_k\}_{k=1}^K$ where $M_{k+1}$ is the $N_{k+1}$-extension of $M_k$.
In contest $M_k$, the game ends once some player achieves $N_k$ consecutive wins from the start.
Let $V_{i;k}$ be the equilibrium continuation value along the \emph{consecutive path} in contest $M_k$:
\begin{itemize}
\item $V_{i;k}$ ($i\ge 0$): value of the player who currently has a streak of $i$ consecutive wins;
\item $V_{-i;k}$ ($i\ge 0$): value of the player who currently has a streak of $i$ consecutive losses.
\end{itemize}
Note that $V_{N_k;k}=v$ and $V_{-N_k;k}=0$.

The extension structure implies the following along the consecutive path: at state $i\in\{1,\dots,N_k-1\}$ in $M_k$,
if the current streak leader loses the next battle, the continuation becomes contest $M_{k-1}$ at state $i-1$.
Hence, for $-N_k+1\leq i\leq N_k-1$, the stake and stake ratio are
\[
\Delta_{i;k}:=V_{i+1;k}-V_{i-1;k-1},\quad
\theta_{i,k}:=\frac{\Delta_{i;k}}{\Delta_{-i,k}}.
\]
We proceed to prove the desired result in the following steps. 

\textbf{Step 1 (useful bounds in a single-battle).}
We set $\theta^*>\max\{\widehat\theta,R'\}$.
Let 
\[
 \eta(\Delta_\ell,\Delta_{-\ell}):=\frac{1-\pi^*(\Delta_\ell,\Delta_{-\ell})}{1-\pi^*(\Delta_{-\ell},\Delta_{\ell})}\frac{\Delta_\ell}{\Delta_{-\ell}}.
\]
We show next that
\begin{equation}\label{eq:forward-vs-lossdiag}
\eta(\Delta_\ell,\Delta_{-\ell})\leq C:=\max\left\{\frac{\theta^*}{d''_{\theta^*}},\frac{1}{\min\{M,1\}}\right\}\text{ if }\Delta_\ell/\Delta_{-\ell}\leq \theta^*,
\end{equation}
where $\widehat\theta$ and $M$ are given in \Cref{assumption_3_ex}. 
On one hand, if $\Delta_\ell/\Delta_{-\ell}\in [\frac{1}{\theta^*},\theta^*]$, by \Cref{lemma:gainratio}\ref{pi-high}, $\eta(\Delta_\ell,\Delta_{-\ell})\in [\frac{d''_{\theta^*}}{\theta^*},\frac{\theta^*}{d''_{\theta^*}}]$. On the other hand, if $\Delta_\ell/\Delta_{-\ell}<1/\theta^*$, $\Delta_{-\ell}/\Delta_{\ell}> {\theta^*}>\widehat \theta$, thus
    \begin{align*}
        \eta(\Delta_{-\ell},\Delta_{\ell}) = \frac{p(x^*_{-\ell},x^*_l)\Delta_{\ell}+x^*_l}{p(x^*_l, x^*_{-\ell})\Delta_{-\ell}+x^*_{-\ell}}\geq \min\left\{\frac{p(x^*_{-\ell},x^*_l)\Delta_{\ell}}{p(x^*_l, x^*_{-\ell})\Delta_{-\ell}}, \frac{x^*_l}{x^*_{-\ell}}\right\}\geq \min\{M,1\}.
    \end{align*}
    which indicates that $\eta(\Delta_{\ell},\Delta_{-\ell})\leq \frac{1}{\min\{M,1\}}$.

It is also worth pointing out that, by \Cref{lemma:gainratio}\ref{pi-ratio}, 
 \begin{equation}\label{eq:theta-star}
\frac{1-\pi^*(\Delta_\ell,\Delta_{-\ell})}{\pi^*(\Delta_\ell,\Delta_{-\ell})}\leq\frac{1-\widetilde \pi}{\widetilde\pi}\text{ if }\Delta_\ell/\Delta_{-\ell}>\theta^*.
\end{equation}

\textbf{Step 2 (within-contest telescoping).}
Since $V_{N_k;k}=v$,
\[
v-V_{1;k}=\sum_{i=1}^{N_k-1}(V_{i+1;k}-V_{i;k}).
\]
Summing over $k=2,\dots,K$ and using the base case $v-V_{1;1}=0$ yields
\begin{equation}\label{eq:big-telescope}
Kv-\sum_{k=1}^K V_{1;k}=\sum_{(k,i)\in\mathcal R}(V_{i+1;k}-V_{i;k}),
\end{equation}
where
\[
\mathcal R := \{(k,i): 2\le k\le K,\; 1\le i\le N_k-1\}.
\]

%%%By Lemma~\ref{lem:phi-limits}, pick $\theta^\ast>1$ large enough that
%%%\begin{equation}\label{eq:theta-star}
%%%\frac{\phi(\theta^\ast)}{1-\phi(\theta^\ast)} > 3
%%%\quad\Longleftrightarrow\quad
%%%\frac{1-\phi(\theta^\ast)}{\phi(\theta^\ast)}<\frac13.
%%%\end{equation}
%%%Also choose $\theta^\ast>\widehat\theta$, where $\widehat\theta$ and $M$ come from Lemma~\ref{lem:eta-large}.
%%%Define
%%%\[
%%%\beta:=\sup_{\theta\in[1/\theta^\ast,\theta^\ast]} \eta(\theta) <+\infty
%%%\qquad\text{and}\qquad
%%%C:=\max\{\beta,1/M\}.
%%%\]
%%%Then $\eta(\theta)\le C$ for all $\theta\in(0,\theta^\ast]$:
%%%if $\theta\in[1/\theta^\ast,\theta^\ast]$, this is by $\beta$;
%%%if $\theta<1/\theta^\ast$, then $1/\theta>\theta^\ast>\widehat\theta$, so Lemma~\ref{lem:eta-large} gives $\eta(1/\theta)>M$ and by \eqref{eq:eta-involution} we have $\eta(\theta)=1/\eta(1/\theta)<1/M\le C$.

\textbf{Step 3 (split the RHS of \eqref{eq:big-telescope}).}
Split $\mathcal R$ into
\[
\mathcal R_{\text{large}}:=\{(k,i)\in\mathcal R:\theta_{i,k}>\theta^\ast\},
\qquad
\mathcal R_{\text{small}}:=\{(k,i)\in\mathcal R:\theta_{i,k}\le \theta^\ast\}.
\]
Then
\[
\sum_{(k,i)\in\mathcal R}(V_{i+1;k}-V_{i;k})
=
\sum_{(k,i)\in\mathcal R_{\text{large}}}(V_{i+1;k}-V_{i;k})
+
\sum_{(k,i)\in\mathcal R_{\text{small}}}(V_{i+1;k}-V_{i;k}).
\]

\textbf{Step 4 (bound the large-$\theta$ part).}
If $(k,i)\in\mathcal R_{\text{large}}$, then by \eqref{eq:theta-star},
\[
V_{i+1;k}-V_{i;k}
=
\frac{1-\pi^*(\Delta_{i;k},\Delta_{-i;k})}{\pi^*(\Delta_{i;k},\Delta_{-i;k})}(V_{i;k}-V_{i-1;k-1})
\leq 
\frac{1-\widetilde \pi}{\widetilde\pi} (V_{i;k}-V_{i-1;k-1}).
\]
Hence
\[
\sum_{(k,i)\in\mathcal R_{\text{large}}}(V_{i+1;k}-V_{i;k})
\leq
\frac{1-\widetilde \pi}{\widetilde\pi}\sum_{(k,i)\in\mathcal R}(V_{i;k}-V_{i-1;k-1})
\le
\frac{1-\widetilde \pi}{\widetilde\pi}\Bigl(Kv-\sum_{k=1}^K V_{0;k}\Bigr).
\]
To see the reason for the last inequality, group the terms of $\sum_{(k,i)\in\mathcal R}(V_{i;k}-V_{i-1;k-1})$ by diagonals $d:=k-i$.
Along a fixed diagonal, the sum telescopes:
\[
\sum_{\substack{(k,i)\in\mathcal R\\k-i=d}}
\bigl(V_{i;k}-V_{i-1;k-1}\bigr)
=
V_{i_{\max};\,d+i_{\max}}-V_{0;d},
\]
where $i_{\max}$ is the largest $i$ on that diagonal in $\mathcal R$.
Because values are bounded above by the prize, $V_{i_{\max};\,d+i_{\max}}\le 1$, so the diagonal sum is at most $v-V_{0;d}$.
Summing over $d=1,\dots,K$ yields the inequality.

\textbf{Step 5 (bound the small-$\theta$ part).}
If $(k,i)\in\mathcal R_{\text{small}}$, then $\theta_{i,k}\le\theta^\ast$ and by \eqref{eq:forward-vs-lossdiag},
\[
V_{i+1;k}-V_{i;k}=\eta(\theta_{i,k})(V_{1-i;k-1}-V_{-i;k})\le C (V_{1-i;k-1}-V_{-i;k}).
\]
Therefore
\[
\sum_{(k,i)\in\mathcal R_{\text{small}}}(V_{i+1;k}-V_{i;k})
\le
C\sum_{(k,i)\in\mathcal R}(V_{1-i;k-1}-V_{-i;k})
\le
C\sum_{k=1}^K V_{0;k}.
\]
To see why the last inequality holds, again fix a diagonal $d:=k-i$.
For some $d$, a summand on that diagonal is
\[
V_{1-j;\,d+j-1}-V_{-j;\,d+j}=V_{-(j-1);\;d+j-1}-V_{-j;\;d+j},
\]
for $j=\{1,\ldots,j_{\max}\}$, where $j_{\max}$ is the largest number such that $(j+d,j)\in\mathcal R$.
Summing over all such $j$ on the diagonal telescopes to
\[
V_{0;d}-V_{-j_{\max};\,d+j_{\max}}\le V_{0;d},
\]
since values are nonnegative.
Summing over $d=1,\dots,K$ gives the desired inequality.

\textbf{Step 6 (combine with \eqref{eq:big-telescope}).}
Plugging Steps 4--5 into \eqref{eq:big-telescope} gives
\[
Kv-\sum_{k=1}^K V_{1;k}
\le
C\sum_{k=1}^K V_{0;k}
+\frac{1-\widetilde \pi}{\widetilde\pi}\Bigl(Kv-\sum_{k=1}^K V_{0;k}\Bigr).
\]
Since $V_{1;k}\ge V_{0;k}$ for each $k$, we have $\sum V_{0;k}\le \sum V_{1;k}$, so
\[
Kv-\sum_{k=1}^K V_{1;k}
\leq
C\sum_{k=1}^K V_{1;k}+\frac{1-\widetilde \pi}{\widetilde\pi}Kv.
\]
Rearrange:
\begin{equation}\label{eq:avgV1}
\frac1K\sum_{k=1}^K V_{1;k}\ge \frac{(2\tpi-1)}{(C+1)\tpi}v.
\end{equation}

\textbf{Step 7 (from $V_{1;k}$ to $V_{0;k}$).}
By \Cref{lemma:gainratio}\ref{pi-ratio},
\[
V_{0;k}=[1-\pi^*(\Delta_{0;k},\Delta_{0;k})]V_{-1;k}+\pi^*(\Delta_{0;k},\Delta_{0;k})V_{1;k}\ge \pi_1^*V_{1;k}.
\]
Average and use \eqref{eq:avgV1}:
\[
\frac1K\sum_{k=1}^K V_{0;k}\ge  \frac{\pi_1^*(2\tpi-1)}{(C+1)\tpi}v.
\]
The proof is completed by setting $\widetilde\alpha:= \frac{\pi_1^*(2\tpi-1)}{(C+1)\tpi}>0$.\end{proof}

\subsection{Proof of the Existence Part of \Cref{prop:cw-mpe} Based on \Cref{prop:gensf,lemma:gainratio}}
\begin{proof}
Our approach is to build a mapping whose fixed point is $(\widehat V(-1),\widehat V(1))$, and show that a fixed point exists for the mapping.

Recall the single-battle equilibrium net gain (over the losing continuation payoff) function $\bm\Pi^*(\Delta',\Delta):=(\Pi^*(\Delta',\Delta),\Pi^*(\Delta,\Delta'))$, whose components are defined by \eqref{def:payoff} in the proof of \Cref{prop:tow-mpe}. Define
%%%\medskip
%%%\noindent\textbf{Step 1: A single-battle equilibrium operator.}
%%%Let $\Pi^*(\Delta_A,\Delta_B)$ be player $A$'s equilibrium \emph{net gain}
%%%(over his losing continuation payoff) in the single battle where $A$'s
%%%winning value is $\Delta_A$ and $B$'s winning value is $\Delta_B$.
%%%Define the equilibrium payoff vector
%%%\[
%%%\bm{\Pi}^*(\Delta_A,\Delta_B)
%%%:=
%%%\begin{pmatrix}
%%%\Pi^*(\Delta_A,\Delta_B)\\[2pt]
%%%\Pi^*(\Delta_B,\Delta_A)
%%%\end{pmatrix}.
%%%\]
%%%By Assumption~\ref{assumption2}, $\bm{\Pi}^*$ is continuous (in winning values).
%%%
%%%Now fix two continuation payoff vectors:
%%%\[
%%%\text{if $A$ wins: }(v_1,v_2), \qquad \text{if $B$ wins: }(v_a,v_b),
%%%\]
%%%where $v_1\ge v_a$ and $v_b\ge v_2$ so that both players have
%%%nonnegative winning values
%%%\[
%%%\Delta_A=v_1-v_a,\qquad \Delta_B=v_b-v_2.
%%%\]
%%%\emph{Interpretation:} the current battle is exactly a single-battle contest
%%%with winning values $(\Delta_A,\Delta_B)$, and the losing continuation payoffs
%%%are $v_a$ for $A$ and $v_2$ for $B$. Therefore the equilibrium payoffs before
%%%the battle are
%%%\[
%%%(V_A^0,V_B^0)=(v_a,v_2)+\bm{\Pi}^*(v_1-v_a,\;v_b-v_2).
%%%\]
%%%This motivates the definition of the operator
\begin{equation}\label{eq:zeta-def-readable}
\bm{\zeta}_{(v_1,v_2)}(v_a,v_b)
:=
\begin{pmatrix}
v_a\\
v_2
\end{pmatrix}
+\bm{\Pi}^*(v_1-v_a,\;v_b-v_2),
\end{equation}
which gives the equilibrium payoffs when the continuation values are $(v_1,v_2)$ after player $A$ wins the current battle and $(v_a,v_b)$ after player $B$ wins. By definition, we have that
\begin{equation}\label{eq:vrange}
\bm{\zeta}_{(v_1,v_2)}(v_a,v_b)\in [v_a,v_1]\times[v_2,v_b],
\end{equation}
which means each player's equilibrium payoff lies between his losing and winning continuation payoffs.
Moreover, 
\begin{equation}\label{eq:vrange2}
\bm{\zeta}_{(v_1,v_2)}(v_a,v_b)\in(v_a,v_1)\times(v_2,v_b),\text{ if }v_1>v_a\text{ and }v_b>v_2.
\end{equation}

%%\medskip
%%\noindent\textbf{Step 2: The fixed-point map for $(\widehat V(-1),\widehat V(1))$.}
%%Let us write
%%\[
%%u:=\widehat V(-1),\qquad w:=\widehat V(1).
%%\]

Define the domain
\[
\mathcal X:=\{(u,w)\in[0,v]^2:\ u+w\le v\},
\]
and define the mapping $\bm{\xi}:\mathcal X\to\mathcal X$ by
\begin{equation}\label{eq:xi-def-readable}
\bm{\xi}(u,w):=\bm{\zeta}_{(w,u)}^{\,K-1}(0,v),
\end{equation}
where $\bm{\zeta}^{\,m}$ denotes the $m$-fold composition.
Note that, if an equilibrium exists, $\left(\widehat V(-1),\widehat V(1)\right)$ is a fixed point of $\bm{\xi}(u,w)$, for the following reason. 
%%%Consider the situation in which $B$ currently
%%%has a winning streak and is $1$ win away from ending the contest.
%%%From $A$'s perspective (as the laggard), the continuation payoffs are:
%%%\[
%%%\text{if $A$ wins: }(w,u)\quad \text{(reset to states $+1$ and $-1$),}
%%%\qquad
%%%\text{if $A$ loses: }(0,v)\quad \text{(terminal).}
%%%\]
%%%Hence the equilibrium payoff vector one step away from the terminal outcome is
%%%.
By definition of $\bm{\zeta}$, it is easy to see that the continuation values when $B$ is one battle win away from winning overall is $\bm{\zeta}_{\left(\widehat V(1),\widehat V(-1)\right)}(0,v)$.
Applying this relation repeatedly, $\bm{\zeta}^{K-1}_{\left(\widehat V(1),\widehat V(-1)\right)}(0,v)$ is the continuation payoffs when $B$ is $K-1$ battle wins away from ultimately winning, or equivalently, has a winning streak of $1$. By symmetry, we have that $\left(\widehat V(-1),\widehat V(1)\right)=\bm{\zeta}^{K-1}_{\left(\widehat V(1),\widehat V(-1)\right)}(0,v)$, which is exactly the fixed point equation for $\bm{\xi}(u,w)$.

We proceed to show that a fixed point exists for $\bm{\xi}$.
We first prove that $\bm{\xi}$ maps $\mathcal X$ into itself and is continuous.
Continuity follows immediately from continuity of $\bm{\Pi}^*$.
To see $\bm{\xi}(\mathcal X)\subseteq\mathcal X$, it suffices to show that
the sum of coordinates never exceeds $v$ along the iteration. Fix $(u,w)\in\mathcal X$
and let $(a_0,b_0)=(0,v)$ and $(a_1,b_1)=\bm{\zeta}_{(w,u)}(a_0,b_0)$, etc.
Write $(a',b')=\bm{\zeta}_{(w,u)}(a,b)$. Suppose $(a,b)\in\mathcal X$, since $\Pi^*(\Delta',\Delta)+\Pi^*(\Delta,\Delta')\leq\max\{\Delta,\Delta'\}$, we have that
\[
\begin{split}
a'+b' &= a+u + \Pi^*(w-a,\;b-u)+\Pi^*(b-u,\;w-a)\\
&\le a+u+\max\{w-a,\;b-u\}=\max\{u+w,\;a+b\}\le v.
\end{split}
\]
Hence all iterates satisfy $a_j+b_j\le v$, and in
particular $\bm{\xi}(u,w)=(a_{K-1},b_{K-1})\in\mathcal X$.  
Since $\mathcal X$ is nonempty, compact, and convex, Brouwer's fixed-point
theorem yields $(u^*,w^*)\in\mathcal X$ such that $\bm{\xi}(u^*,w^*)=(u^*,w^*)$.

Next, we construct $\widehat V(\cdot)$ from the fixed point $(u^*,w^*)$.
% For that purpose, we first show that $u^*<w^*$. 
Define the iterates
\begin{equation}\label{eq:fix-point-consec-win}
    (a_0^*,b_0^*):=(0,v),\qquad
(a_j^*,b_j^*):=\bm{\zeta}_{(w^*,u^*)}(a_{j-1}^*,b_{j-1}^*),
\quad j=1,\ldots,K-1.
\end{equation}
By the fixed-point condition, $(a_{K-1}^*,b_{K-1}^*)=(u^*,w^*)$.

By \eqref{eq:vrange}, 
$0=a_0^*\leq a_1^*\leq\cdots\leq a_{K-1}^*\leq w^*$, so $u^*=a_{K-1}^*\leq w^*$. Since $\bm{\zeta}^{K-1}_{(0,0)}(0,v)\neq (0,0)$, $(u^*, w^*)\neq(0,0)$. In combination, $u^*\leq w^*$, $(u^*, w^*)\neq(0,0)$, and $u^*+w^*\leq v$ imply that $w^*>0$ and $v>u^*$. 
Then repeatedly applying \eqref{eq:vrange2} yields that 
\[
0=a_0^*<a_1^*<\cdots<a_{K-1}^*=u^*<w^*=b_{K-1}^*<\cdots<b_1^*<b_0^*=v.
\]
%%Again by \eqref{eq:vrange}, 
%%$w^*=b_{K-1}^*\leq\cdots\leq b_1^*\leq b_0^*=v$

Then it is straightforward to verify that the following construction satisfies the equilibrium condition:
\[
\widehat V(-K)=0,\quad \widehat V(K)=v,
\qquad
\widehat V(-K+j):=a_j^*,\quad
\widehat V(K-j):=b_j^* \quad (j=1,\ldots,K-1),
\]
and 
\[
\widehat V(0):=u^*+\Pi^*(w^*-u^*,\,w^*-u^*).
\]
\end{proof}

\subsection{Proof of the Uniqueness Part of \Cref{prop:cw-mpe} Based on \Cref{cond:bsf}}
\begin{proof}
We prove uniqueness under \Cref{cond:bsf}. The $K=1$ case follows from
\Cref{lemma:single}. So suppose $K\ge 2$.
Let $(u,w):=(\widehat V(-1),\widehat V(1))$ be induced by a symmetric MPE.
As in the existence proof, define
\[
(a_0,b_0):=(0,v), \qquad
(a_j,b_j):=\bm\zeta_{(w,u)}(a_{j-1},b_{j-1}), \quad j=1,\ldots,K-1.
\]
Then the fixed-point condition is
\begin{equation}
(a_{K-1},b_{K-1})=(u,w).
\label{eq:cw-fixed-point}
\end{equation}

Under \Cref{cond:bsf}, \eqref{eq:zeta-def-readable} becomes
\begin{align}
a_j
&=a_{j-1}
+(w-a_{j-1})\phi\!\left(\frac{w-a_{j-1}}{b_{j-1}-u}\right), \label{eq:cw-aj}\\
b_j
&=u
+(b_{j-1}-u)\phi\!\left(\frac{b_{j-1}-u}{w-a_{j-1}}\right). \label{eq:cw-bj}
\end{align}

Now define
\[
r_j:=\frac{b_j-u}{\,w-a_j\,}, \qquad j=0,\ldots,K-1.
\]
Using \eqref{eq:cw-aj}--\eqref{eq:cw-bj}, we obtain
\[
w-a_j=(w-a_{j-1})\Bigl[1-\phi(1/r_{j-1})\Bigr]
\]
and
\[
b_j-u=(b_{j-1}-u)\phi(r_{j-1}),
\]
so
\[
r_j
=
\frac{r_{j-1}\phi(r_{j-1})}{1-\phi(1/r_{j-1})}
=: \psi(r_{j-1}).
\]
Hence
\[
r_j=\psi^j(r_0), \qquad r_0=\frac{v-u}{w}.
\]

By \Cref{lemma:psi}, $\psi$ is strictly increasing, so $\psi^{K-1}$ is strictly increasing as well.
From \eqref{eq:cw-fixed-point}, we have
\[
r_{K-1}=\frac{b_{K-1}-u}{w-a_{K-1}}=\frac{w-u}{w-u}=1.
\]
Therefore $r_0$ is uniquely determined as the unique solution to
\[
\psi^{K-1}(r)=1.
\]
Denote this unique solution by $\rho_K$. Then
\begin{equation}
\frac{v-u}{w}=\rho_K,
\qquad\text{so}\qquad
u=v-\rho_K w.
\label{eq:cw-rho}
\end{equation}

Next, iterate the recursion for $w-a_j$:
\[
w-a_j
=
w\prod_{t=0}^{j-1}\Bigl[1-\phi(1/r_t)\Bigr].
\]
Since $r_t=\psi^t(\rho_K)$, the product
\[
P_K
:=
\prod_{t=0}^{K-2}
\Bigl[1-\phi\!\bigl(1/\psi^t(\rho_K)\bigr)\Bigr]
\]
depends only on $K$ and the success function.
Using again \eqref{eq:cw-fixed-point}, namely $a_{K-1}=u$, we get
\[
w-u=w-a_{K-1}=wP_K.
\]
Substituting \eqref{eq:cw-rho} gives
\[
(1+\rho_K-P_K)w=v.
\]
Since $0<\phi(\theta)<1$ for every $\theta>0$, we have $0<P_K<1$, hence
$1+\rho_K-P_K>0$. Therefore
\[
w=\frac{v}{1+\rho_K-P_K},
\qquad
u=v-\rho_K w
\]
are uniquely determined.

So $(u,w)=(\widehat V(-1),\widehat V(1))$ is unique. Once $(u,w)$ is fixed,
the whole sequence $(a_j,b_j)$ is uniquely generated by
\[
(a_j,b_j)=\bm\zeta_{(w,u)}^j(0,v),
\]
and therefore the value function $\widehat V(\cdot)$ is unique as well.
\end{proof}

\subsection{Proof of \Cref{lemma:cw-reinf} Based on \Cref{prop:gensf,lemma:gainratio}}
\begin{proof}
By \eqref{consective_win_valuation}, we have that
    \begin{align*}
         \quad & \frac{\widehat{\Delta}(i)}{\widehat{\Delta}(-i)} = \frac{\widehat{V}(i+1) - \widehat{V}(-1) }{\widehat{V}(1) - \widehat{V}(-i-1)} = \frac{\left[\widehat V(-1)+\widehat{\pi}(i+1)\widehat\Delta(i+1)\right]-\widehat V(-1)}{\widehat V(1) - \left[\widehat V(-i-2)+\widehat{\pi}(-i-1)\widehat\Delta(-i-1)\right]} \\
        \quad & = \frac{\widehat{\pi}(i+1)\widehat\Delta(i+1)}{\widehat\Delta(-i-1)-\widehat{\pi}(-i-1)\widehat\Delta(-i-1)}  = \frac{\widehat{\pi}(i+1)}{1-\widehat{\pi}(-i-1)}\frac{\widehat{\Delta}(i+1)}{\widehat{\Delta}(-i-1)}.
    \end{align*}
    for $i = 0,1,\cdots,K-2$. 
\end{proof}

\subsection{Proof of \Cref{prop:cw-adv} Based on \Cref{prop:gensf,lemma:gainratio} }
\begin{proof}
Consider a symmetric MPE in the $K$-consecutive-win contest, with valuation profile
\[
     v=\widehat V(K)>\widehat V(K-1)>\cdots>\widehat V(1)>\widehat V(-1)>\cdots>\widehat V(-K)=0 .
\]
Let $\widehat p(i)$ denote the probability that the player at state $i$ wins the next battle,
for $i=K-1,\ldots,1-K$, with $\widehat p(i)+\widehat p(-i)=1$.
For convenience, define the equilibrium effort at state $i\geq0$ by
\[
\widehat x(i)
:=x^*\!\big(\widehat V(i+1)-\widehat V(-1),\;\widehat V(1)-\widehat V(-i-1)\big).
\]

We now focus on the $i\geq0$ case and turn to the $i\leq-1$ case at the end of the proof. By definition of the value function, for $i\ge 0$,
\[
\widehat V(i)
= \widehat p(i)\widehat V(i+1)
+[1- \widehat p(i)]\widehat V(-1)
- \widehat x(i).
\]
Rearranging yields that
\begin{equation}
\frac{\widehat V(i+1)-\widehat V(-1)}{\widehat V(i)-\widehat V(-1)}
=\frac{\widehat V(i)-\widehat V(-1)+\widehat x(i)}{\widehat p(i)[\widehat V(i)-\widehat V(-1)]}.
\label{eq:consecutive_recursion}
\end{equation}
Iterating \eqref{eq:consecutive_recursion} from $i=1$ to $K-1$, and using the fact that $\widehat V(K)=v$, gives
\[
%%%\frac{\widehat V(K)-\widehat V(-1)}{\widehat V(1)-\widehat V(-1)}
%%%=
\frac{v-\widehat V(-1)}{\widehat V(1)-\widehat V(-1)}
=\prod_{i=1}^{K-1}
\frac{\widehat V(i)-\widehat V(-1)+\widehat x(i)}
{\widehat p(i)\,[\widehat V(i)-\widehat V(-1)]}.
\]
%%%%\[
%%%%\begin{split}
%%%%\frac{\widehat V(K)-\widehat V(-1)}{\widehat V(1)-\widehat V(-1)}
%%%%&=\prod_{i=1}^{K-1}
%%%%\frac{\widehat V(i)-\widehat V(-1)+\widehat x(i)}
%%%%{\widehat p(i)\,[\widehat V(i)-\widehat V(-1)]}
%%%%\\
%%%%&\leq 
%%%%\prod_{i=1}^{K-1}
%%%%\frac{\widehat V(i)-\widehat V(-1)+\widehat p(i)[\widehat V(i+1)-\widehat V(-1)]}
%%%%{\widehat p(i)\,[\widehat V(i)-\widehat V(-1)]}.
%%%%\end{split}
%%%%\]
As will be shown in the proof of \Cref{thm:cw-rent}, both $\widehat V(1)/v$ and $\widehat V(-1)/v$ approach 0 as $K\to+\infty$. Therefore, the left-hand side diverges. 
This implies that the right-hand side also diverges. 
Notably,
\[\begin{split}
1<\frac{\widehat V(i)-\widehat V(-1)+\widehat x(i)}
{\widehat p(i)\,[\widehat V(i)-\widehat V(-1)]}
&\leq
\frac{\widehat V(i)-\widehat V(-1)+\widehat p(i)\left[\widehat V(i+1)-\widehat V(-1)\right]}
{\widehat p(i)\,[\widehat V(i)-\widehat V(-1)]}
\\
&=
\frac{1+\frac{\widehat p(i)}{\widehat \pi(i)}}{\widehat p(i)}\leq
\frac{1+\frac{1}{\pi_1^*}}{\pi_1^*},
\end{split}
\]
where the last inequality follows from $1\geq\widehat p(i)\geq\widehat \pi(i)\geq\pi_1^*$ for $i\geq0$, by \Cref{lemma:gainratio}\ref{pi-ratio}.
Notice that for a series $\frac{1+\frac{1}{\pi_1^*}}{\pi_1^*}>a_{Ki}>1$, with $1\leq i\leq K-1$,  $\prod_{i=1}^{K-1}a_{Ki}\to+\infty$ as $K\to+\infty$ is equivalent to $\sum_{i=1}^{K-1}\left(1-\frac{1}{a_{Ki}}\right)\to+\infty$ as $K\to+\infty$. It follows that 
\begin{align}
\sum_{i=1}^{K-1}
\left(
1-\frac{\widehat p(i)[\widehat V(i)-\widehat V(-1)]}
{\widehat V(i)-\widehat V(-1)+\widehat x(i)}
\right)
&=
\sum_{i=1}^{K-1}
\left(
\frac{\widehat p(i)\widehat x(i)}{\widehat V(i)-\widehat V(-1)+\widehat x(i)}
+\widehat p(-i)
\right)
\to+\infty .
\label{eq:consecutive_sum}
\end{align}
By \eqref{eq:consecutive_recursion}, \Cref{prop:gensf}\ref{cond1}, and the fact that $\widehat x(-i)\leq \widehat p(-i)\widehat\Delta(-i)$, we have that 
\[
\frac{\widehat p(i)\widehat x(i)}{\widehat V(i)-\widehat V(-1)+\widehat x(i)}
=\frac{\widehat x(i)}{\widehat V(i+1)-\widehat V(-1)}=\frac{\widehat x(i)}{\widehat\Delta(i)}\leq C\frac{\widehat x(-i)}{\widehat\Delta(-i)}\leq C\widehat p(-i).
%%%
%%%\le C\,\widehat p(-i).
\]
Substituting this bound into \eqref{eq:consecutive_sum} implies
\[
(C+1)\sum_{i=1}^{K-1}\widehat p(-i)\to+\infty,
\qquad\text{hence}\qquad
\sum_{i=1}^{K-1} [1-\widehat p(i)]\to+\infty,
\]
which, together with the fact that $\widehat p(i)\leq e^{-[1-\widehat p(i)]}$, yields that 
\[
\prod_{i=1}^{K-1}\widehat p(i)\to0.
\]

Finally, we turn to the probability of ultimately winning the contest from state $i$, $\widehat Q(i;K,v)$.
It satisfies
\[
\widehat Q(i;K,v)
=\widehat p(i)\widehat Q(i+1;K,v)
+\widehat p(-i)\widehat Q(-1;K,v).
\]
Rearranging,
\[
\widehat Q(i+1;K,v)-\widehat Q(-1;K,v)
=\frac{\widehat Q(i;K,v)-\widehat Q(-1;K,v)}{\widehat p(i)}.
\]
Iterating this relation yields
\[
\widehat Q(K;K,v)-\widehat Q(-1;K,v)
=\frac{\widehat Q(i;K,v)-\widehat Q(-1;K,v)}
{\prod_{j=i}^{K-1}\widehat p(j)}.
\]
Since $\widehat Q(K;K,v)=1$ and $\widehat Q(-1;K,v)\ge0$, the left-hand side is bounded,
while for fixed $i$ the denominator converges to zero as $K\to+\infty$.
Therefore,
\[
\widehat Q(i;K,v)-\widehat Q(-1;K,v)\to0.
\]
Since $\widehat Q(i;K,v)\geq \frac{1}{2}\geq  \widehat Q(-1;K,v)$, $\lim_{K>i,K\to +\infty}\widehat Q(i;K,v)=\frac{1}{2}$ for any given $i\geq0$. Then $\lim_{K>i,K\to +\infty}\widehat Q(-i;K,v)=\frac{1}{2}$ follows from the fact that $\widehat Q(-i;K,v)=1-\widehat Q(i;K,v)$.%
%%Finally, by symmetry $\widehat Q(-1,K,v)=\frac12$, implying
%%\[
%%\lim_{K\to\infty,\;K>|i|}\widehat Q(i,K,v)=\frac12 .
%%\]
\end{proof}

\subsection{Proof of \Cref{thm:cw-rent} Based on \Cref{prop:gensf,lemma:gainratio}}
\begin{proof}
We show that for all $R''>1$ there exists $K_{R''}\in\N_{++}$ such that $\frac{\widehat{\Delta}(i;K,v)}{\widehat{\Delta}(-i;K,v)}>R''$ for all $i\geq K_{R''}$, $K>i$ and $v>0$.
This is because, by \Cref{lemma:gainratio}\ref{pi-high}, for all $i$ such that $\frac{\widehat{\Delta}(i;K,v)}{\widehat{\Delta}(-i;K,v)}\leq R''$, $\frac{1-\widehat{\pi}(-i-1)}{\widehat{\pi}(i+1)}\geq 1+d_{R''}$. Together with \Cref{lemma:cw-reinf}, this implies that 
%%$\frac{\widehat{\Delta}(i;K,v)}{\widehat{\Delta}(-i;K,v)}$ grows faster than exponentially and 
$K_{R''}:=\log_{1+d_{R''}}R''$ satisfies our requirement.

Therefore, 
$\frac{\widehat{\Delta}(K-1;K,v)}{\widehat{\Delta}(1-K;K,v)}=\frac{\widehat{V}(K;K,v) - \widehat{V}(-1;K,v) }{\widehat{V}(1;K,v) - \widehat{V}(-K;K,v)}=\frac{v- \widehat{V}(-1;K,v) }{\widehat{V}(1;K,v) }\to+\infty$ as $K\to+\infty$, which implies that $\frac{v}{\widehat{V}(1;K,v) }\to+\infty$.  
    As a result, for all $\epsilon >0$, there exists $K^\dagger$ such that when $K>K^\dagger$, $\widehat{V}(-1)<\widehat{V}(1)<\frac{\epsilon}{2}v$, which indicates that $\widehat{V}(0)<\epsilon v$. 
\end{proof}

\subsection{Proof of \Cref{thm:sufficient} Based on \Cref{prop:gensf,lemma:gainratio}}
\begin{proof}
\textbf{The sufficiency part.} By symmetry, it is without loss to focus on the initial value of player $A$, denoted by $V_A^*$.
For a terminal history $h\in H^\dagger$, 
let $\widetilde H(h):=\{h': h'\text{ is a prefix of }h\}$ denote the set of all subhistories of $h$.
If $\widetilde H(h)\cap H^B\neq\emptyset$, let $h_B$ denote the \emph{shortest} prefix (i.e., subhistory) of $h$ that belongs to $H^B$.
It holds that
\[
V_A^* \le \sum_{\substack{h\in H^\dagger\text{ and} \\ \widetilde H(h)\cap H^B=\emptyset}} \Pr(h)v+
\sum_{\substack{h\in H^\dagger\text{ and} \\ \widetilde H(h)\cap H^B\neq\emptyset}}\Pr(h) V_A(h_B),
\]
where $V_A(h_B)$ denotes the continuation value of player $A$ when a terminal history $h$ first reaches $H^B$. 

By assumption and \Cref{def:si}, we have that
\[\sum_{\substack{h\in H^\dagger\text{ and} \\ \widetilde H(h)\cap H^B=\emptyset}} \Pr(h)\leq \epsilon \text{ and } V_A(h_B)\leq \epsilon v.\]
%%which implies $V_A(h_B)< \frac{v}{M}$. 
Therefore, 
$V_A^*\leq 2\epsilon v,$
and the expected total effort is 
$v-V_A^*-V_B^*\geq\left(1-4\epsilon\right)v$.

\noindent\textbf{The necessity part.} Since the expected total effort in equilibrium is greater than $(1-\epsilon)v$, the total equilibrium payoff is less than $\epsilon v$. This means that the transient dominance property is satisfied by setting
\(H_A^- = H_B^- = \{h^0\}\), where \(h^0\) denotes the initial (empty) history before any battle.\end{proof}

\subsection{Proof of \Cref{thm:towp} under \Cref{cond:bsf}}
\label{sec:towp-proof}
\begin{proof}
Since the success function is homogeneous, we normalize the prize to $v=1$ throughout this proof. 
We first prove the existence and uniqueness of the symmetric MPE. 

\textbf{Step 1: States and value functions.}
Let the state $i\in\{-N,-N+1,\dots,N\}$ denote a player's lead in the number
of net battle wins since the last reset. States $\pm N$ are absorbing: if $i=N$,
the player has won the contest and receives $1$, whereas if $i=-N$, the player
receives $0$.

To incorporate the reset lottery, it is convenient to distinguish two types of
continuation values:
\begin{itemize}
\item $V(i)$ is the player's continuation value at a \emph{decision node} in
state $i$, i.e., immediately before the next battle when the reset lottery (if
any) has already been realized.
\item $\widetilde V(i)$ is the player's continuation value at an
\emph{intermediate node} in state $i$, i.e., immediately after the most recent
battle outcome has updated the lead to $i$, but \emph{before} the reset lottery
is realized.
\end{itemize}

By definition of the reset rule, for every nonterminal $i\in\{-N+1,\dots,N-1\}$,
\begin{equation}\label{eq:reset-link}
\widetilde V(i)=p\,V(0)+(1-p)\,V(i).
\end{equation}
At terminal states, $\widetilde V(N)=V(N)=1$ and $\widetilde V(-N)=V(-N)=0$.

\textbf{Step 2: Bellman equations using the single-battle characterization.}
By \Cref{lemma:single}, the current battle has a unique equilibrium, and player
$A$'s equilibrium payoff from this battle equals the losing continuation value
plus the equilibrium gain:
\begin{equation}\label{eq:Vprime}
V(i)=\widetilde V(i-1)
+\bigl(\widetilde V(i+1)-\widetilde V(i-1)\bigr)\,
\phi\!\left(
\frac{\widetilde V(i+1)-\widetilde V(i-1)}{\widetilde V(1-i)-\widetilde V(-i-1)}
\right),\text{ for }i=1-N,\dots,N-1.
\end{equation}
Combining \eqref{eq:reset-link} and \eqref{eq:Vprime} yields a closed system for
$\{\widetilde V(i)\}_{i=-N}^N$. We proceed to show that this system has a unique solution with a constructive approach. 

\textbf{Step 3: Normalization and recursive construction.}
We define the normalized values:
\begin{equation}\label{eq:Delta-def}
\widetilde \Delta(i):=\frac{\widetilde V(i)-\widetilde V(0)}{\widetilde V(1)-\widetilde V(-1)},
\qquad i=-N,\dots,N.
\end{equation}
Then $\widetilde\Delta(0)=0$ and $\widetilde\Delta(1)-\widetilde\Delta(-1)=1$.

At state $i=0$, we have $\widetilde V(0)=V(0)$ and \eqref{eq:Vprime} becomes
\[
\widetilde V(0)=\widetilde V(-1)+(\widetilde V(1)-\widetilde V(-1))\phi(1).
\]
Rearranging and using the normalization \eqref{eq:Delta-def} gives the
\emph{initial conditions}
\begin{equation}\label{eq:Delta-init}
\widetilde\Delta(1)=1-\phi(1),
\qquad
\widetilde\Delta(-1)=-\phi(1).
\end{equation}

For each $i\in\{1,\dots,N-1\}$ define the ratio
\begin{equation}\label{eq:theta-def}
\theta_i
:=
\frac{\widetilde\Delta(i+1)-\widetilde\Delta(i-1)}{\widetilde\Delta(1-i)-\widetilde\Delta(-i-1)}
\;=\;
\frac{\widetilde V(i+1)-\widetilde V(i-1)}{\widetilde V(1-i)-\widetilde V(-i-1)}.
\end{equation}
Subtract $\widetilde V(0)$ from both sides of \eqref{eq:reset-link}--\eqref{eq:Vprime},
divide by $\widetilde V(1)-\widetilde V(-1)$, and use \eqref{eq:theta-def}. This yields, for $i=1,\dots,N-1$,
\begin{equation}\label{eq:Delta-rec}
\begin{aligned}
\widetilde \Delta(i)
&=(1-p)\Bigl[\widetilde\Delta(i-1)+\bigl(\widetilde\Delta(i+1)-\widetilde\Delta(i-1)\bigr)\phi(\theta_i)\Bigr],\\
\widetilde \Delta(-i)
&=(1-p)\Bigl[\widetilde\Delta(-i-1)+\bigl(\widetilde\Delta(1-i)-\widetilde\Delta(-i-1)\bigr)\phi(1/\theta_i)\Bigr]\\
&=(1-p)\Bigl[\widetilde\Delta(1-i)-\bigl(\widetilde\Delta(1-i)-\widetilde\Delta(-i-1)\bigr)\left(1-\phi(1/\theta_i)\right)\Bigr].
\end{aligned}
\end{equation}

Using \eqref{eq:Delta-rec} and \eqref{eq:theta-def}, a direct rearrangement
gives, for each $i=1,\dots,N-1$,
\begin{equation}\label{eq:psi-eq}
-\frac{\widetilde\Delta(i)-(1-p)\widetilde\Delta(i-1)}
{\widetilde\Delta(-i)-(1-p)\widetilde\Delta(-i-1)}=\psi(\theta_i),
\end{equation}
where
\[
\psi(\theta):=\frac{\theta\,\phi(\theta)}{1-\phi(1/\theta)},\qquad \theta>0.
\]
%%Importantly, the right-hand side of \eqref{eq:psi-eq} depends only on
%%$\{\widetilde\Delta(j)\}_{j=-i}^{i}$, so it can be computed inductively.
We establish the following result to facilitate the inductive computation of $\widetilde\Delta(i)$. 
\begin{lemma}
\label{lemma:psi}
The function $\psi(\theta):=\frac{\theta\phi(\theta)}{1-\phi(1/\theta)}$ is strictly increasing in $\theta>0$. 
\end{lemma}
\begin{proof}
Note that
\[
\psi(\theta)
=\frac{\theta\phi(\theta)}{1-\phi(\frac{1}{\theta})}
=\frac{\phi(\theta)}{\frac{\gamma(\theta)}{\theta}+\gamma'(\theta)},
\]
where the second equality is because
$1-\gamma\!\left(\frac{1}{\theta}\right)=\gamma(\theta)$
and
$\theta\gamma'(\theta)=\frac{1}{\theta}\gamma'\!\left(\frac{1}{\theta}\right)$.
Since $\phi(\theta)$ is strictly increasing, and
$\frac{\gamma(\theta)}{\theta}+\gamma'(\theta)$ is decreasing due to the concavity of $\gamma(\theta)$,
it is clear that $\psi(\theta)$ is strictly increasing.
\end{proof}

\emph{Induction/construction.}
Start from \eqref{eq:Delta-init} and $\widetilde\Delta(0)=0$. Suppose that for
some $k\in\{1,\dots,N-1\}$ the values
$\widetilde\Delta(j)$ have been constructed for all $j\in\{-k,\dots,k\}$ and
satisfy
\[
\widetilde\Delta(k)>\widetilde\Delta(k-1)>\dots>\widetilde\Delta(0)=0
>\dots>\widetilde\Delta(1-k)>\widetilde\Delta(-k).
\]

Then 
\[
-\frac{\widetilde\Delta(k)-(1-p)\widetilde\Delta(k-1)}
{\widetilde\Delta(-k)-(1-p)\widetilde\Delta(1-k)}=
\frac{\widetilde\Delta(k)-(1-p)\widetilde\Delta(k-1)}
{\widetilde\Delta(1-k)-\widetilde\Delta(-k)-p\widetilde\Delta(1-k)},
\]
the numerator of the right-hand side is strictly
positive (because $\widetilde\Delta(k)>(1-p)\widetilde\Delta(k-1)$), and the
denominator is strictly positive (because $\widetilde\Delta(1-k)>\widetilde\Delta(-k)$
and $\widetilde\Delta(1-k)\leq 0$). Hence $\theta_k$ is uniquely determined as
\begin{equation}\label{eq:theta-sol}
\theta_k
=
\psi^{-1}\!\left(-
\frac{\widetilde\Delta(k)-(1-p)\widetilde\Delta(k-1)}
{\widetilde\Delta(-k)-(1-p)\widetilde\Delta(1-k)}
\right).
\end{equation}
Given $\theta_k$, solve \eqref{eq:Delta-rec} for $\widetilde\Delta(k+1)$ and
$\widetilde\Delta(-k-1)$:
\begin{align}
\widetilde\Delta(k+1)
&=
\widetilde\Delta(k-1)
+\underbrace{\frac{\widetilde\Delta(k)/(1-p)-\widetilde\Delta(k-1)}{\phi(\theta_k)}}_{\substack{{>\widetilde\Delta(k)/(1-p)-\widetilde\Delta(k-1)}\\{\text{ because }0<\phi(\theta)<1\text{ for all }\theta>0}}}>\frac{\widetilde\Delta(k)}{1-p},
\label{eq:Delta-plus}\\
\widetilde\Delta(-k-1)
&=
\frac{\widetilde\Delta(-k)}{(1-p)(1-\phi(1/\theta_k))}+\frac{\phi(1/\theta_k)[-\widetilde\Delta(1-k)]}{1-\phi(1/\theta_k)}\notag\\
&<\frac{\widetilde\Delta(-k)}{(1-p)(1-\phi(1/\theta_k))}<\frac{\widetilde\Delta(-k)}{1-p}.
%\frac{\widetilde\Delta(-k)/(1-p)-\phi(1/\theta_k)\,\widetilde\Delta(1-k)}{1-\phi(1/\theta_k)}.
\label{eq:Delta-minus}
\end{align}
Because $0<\phi(\theta)<1$ for all $\theta>0$ under \Cref{cond:bsf},
 the
%%%%denominators in \eqref{eq:Delta-plus} and \eqref{eq:Delta-minus} are strictly
%%%%positive, so $(\widetilde\Delta(k+1),\widetilde\Delta(-k-1))$ is uniquely
%%%%determined. Moreover, \eqref{eq:Delta-plus} implies
%%%%$\widetilde\Delta(k+1)>\widetilde\Delta(k)/(1-p)>\widetilde\Delta(k)$, and
%%%%\eqref{eq:Delta-minus} implies $\widetilde\Delta(-k-1)<\widetilde\Delta(-k)$.
%%%%Thus 
the strict monotonicity of $\widetilde\Delta(\cdot)$ extends to $\pm(k+1)$.
This completes the inductive construction and shows that, for each fixed $(N,p)$,
there is a unique normalized sequence $\{\widetilde\Delta(i)\}_{i=-N}^N$
satisfying \eqref{eq:Delta-init}--\eqref{eq:Delta-rec}.

\textbf{Step 4: Recovering $\widetilde V$ and concluding uniqueness of the
symmetric MPE.}
Given $\{\widetilde\Delta(i)\}_{i=-N}^N$, the boundary conditions
$\widetilde V(-N)=0$ and $\widetilde V(N)=1$ pin down the affine scaling in
\eqref{eq:Delta-def} uniquely. Indeed, since
$\widetilde\Delta(N)>\widetilde\Delta(-N)$, set 
\[
\widetilde V(i)=\frac{\widetilde\Delta(i)-\widetilde\Delta(-N)}{\widetilde\Delta(N)-\widetilde\Delta(-N)}.
\]
Then $\widetilde V(-N)=0$, $\widetilde V(N)=1$, and the series $\widetilde V(i)$ solves the equilibrium equations \eqref{eq:reset-link} and \eqref{eq:Vprime} by construction.

Uniqueness follows because any symmetric MPE induces (via \eqref{eq:Delta-def})
a normalized sequence $\{\widetilde\Delta(i)\}_{i=-N}^N$ satisfying the same
system \eqref{eq:Delta-init}--\eqref{eq:Delta-rec}. The inductive construction
above shows this normalized sequence is unique, and therefore the associated
$(\widetilde V,V)$ and equilibrium efforts are unique as well.

\textbf{Next, we establish the transient dominance property for large $N$ when $p\in(0,1)$.} 
Fix $p\in(0,1)$ and $\epsilon>0$. 
For this part, it is convenient to extend the normalized sequence $\{\widetilde\Delta(i)\}_{i=-N}^N$
constructed above to an \emph{infinite} sequence $\{\widetilde\Delta(i)\}_{i=-\infty}^{+\infty}$ by iterating the
recursion \eqref{eq:theta-sol}--\eqref{eq:Delta-minus} for $k=1,2,\dots$ (starting from the initial conditions
\eqref{eq:Delta-init}).  For any finite $N$, the normalized sequence in the margin-$N$ game is just the truncation
of this infinite sequence on $\{-N,\dots,N\}$.

For $k\ge1$ define the ratio
\[
R_k:= -\frac{\widetilde\Delta(k)}{\widetilde\Delta(-k)} \;>\;0.
\]
We first show that $R_k$ is strictly increasing and diverges to $+\infty$; we then use this to verify the two
conditions in \Cref{def:si}.

\textbf{Step 5: Monotonicity of $R_k$.}
Since $1>\phi(\theta)+\phi(1/\theta)$ for $\theta>0$, we have
$\psi(\theta)<\theta$ and $\psi(1/\theta)<1/\theta$, hence $1/\psi(1/\theta)>\theta$. Together with \eqref{eq:Delta-plus}--\eqref{eq:Delta-minus}, we know that for every $k\ge1$,
\begin{align}\label{eq5_reset}
\frac{\widetilde\Delta(k+1)-\frac{\widetilde\Delta(k)}{1-p}}{\frac{\widetilde\Delta(-k)}{1-p}-\widetilde\Delta(-k-1)}
=\frac{1}{\psi(1/\theta_k)}
>\theta_k>\psi(\theta_k)
=\frac{\widetilde\Delta(k)-(1-p)\widetilde\Delta(k-1)}
{-\widetilde\Delta(-k)+(1-p)\widetilde\Delta(1-k)} .
\end{align}

We prove that $R_{k+1}>R_k$ by induction on $k$.  
For $k=1$, since $\widetilde\Delta(0)=0$ the rightmost fraction in \eqref{eq5_reset} equals
$-\widetilde\Delta(1)/\widetilde\Delta(-1)=R_1$, so \eqref{eq5_reset} implies
\[
\frac{\widetilde\Delta(2)-\frac{\widetilde\Delta(1)}{1-p}}{\frac{\widetilde\Delta(-1)}{1-p}-\widetilde\Delta(-2)}
>R_1.
\]
Writing
$\widetilde\Delta(2)=\frac{\widetilde\Delta(1)}{1-p}+\Big(\widetilde\Delta(2)-\frac{\widetilde\Delta(1)}{1-p}\Big)$
and
$-\widetilde\Delta(-2)=\frac{-\widetilde\Delta(-1)}{1-p}+\Big(\frac{\widetilde\Delta(-1)}{1-p}-\widetilde\Delta(-2)\Big)$,
this inequality implies $R_2=-\widetilde\Delta(2)/\widetilde\Delta(-2)>R_1$.

Now fix $k\ge2$ and suppose $R_k>R_{k-1}$.  Let
$c:=\widetilde\Delta(k)$, $d:=-\widetilde\Delta(-k)$, $c':=(1-p)\widetilde\Delta(k-1)$, and
$d':=(1-p)(-\widetilde\Delta(1-k))$.  Then $c/d=R_k$ and $c'/d'=R_{k-1}$.
The induction hypothesis $c/d>c'/d'$ implies
\[
\frac{c-c'}{d-d'}>\frac{c}{d},
\]
i.e., the rightmost fraction in \eqref{eq5_reset} satisfies $\psi(\theta_k)>R_k$.
Combining with \eqref{eq5_reset} yields
\[
\frac{\widetilde\Delta(k+1)-\frac{\widetilde\Delta(k)}{1-p}}{\frac{\widetilde\Delta(-k)}{1-p}-\widetilde\Delta(-k-1)}
>R_k.
\]
Using the same ``ratio-of-sums'' argument as in the base case, we conclude
$R_{k+1}=-\widetilde\Delta(k+1)/\widetilde\Delta(-k-1)>R_k$.  Hence $\{R_k\}$ is strictly increasing.

\textbf{Step 6: Divergence of $R_k$.}
Since $R_k$ is increasing, either $R_k\to+\infty$ or $R_k\to M<+\infty$.
Suppose toward a contradiction that $R_k\to M<+\infty$.  Then $R_k\le M+1$ for all large $k$.

First, $\theta_k$ is uniformly bounded from above.  
By \eqref{eq:theta-sol} and the monotonicity of $\psi^{-1}$,
\begin{equation}\label{upper-bound-theta1}
\theta_k
=\psi^{-1}\!\left(
-\frac{\widetilde\Delta(k)-(1-p)\widetilde\Delta(k-1)}{\widetilde\Delta(-k)-(1-p)\widetilde\Delta(1-k)}
\right)
\le \psi^{-1}\!\left(-\frac{\widetilde\Delta(k)}{p\,\widetilde\Delta(-k)}\right)
\le \psi^{-1}\!\left(\frac{M+1}{p}\right),
\end{equation}
where the first inequality is because
\begin{align}\label{upper-bound-theta}
-\frac{\widetilde\Delta(k)-(1-p)\widetilde\Delta(k-1)}{\widetilde\Delta(-k)-(1-p)\widetilde\Delta(1-k)}
%%%&=\frac{\widetilde\Delta(k)-(1-p)\widetilde\Delta(k-1)}{(1-p)\widetilde\Delta(1-k)-\widetilde\Delta(-k)} \nonumber\\
%%%&
=\frac{\widetilde\Delta(k)-(1-p)\widetilde\Delta(k-1)}
{(1-p)(\widetilde\Delta(1-k)-\widetilde\Delta(-k))-p\,\widetilde\Delta(-k)}
\;\le\; -\frac{\widetilde\Delta(k)}{p\,\widetilde\Delta(-k)}.
\end{align}
On the other hand, Step~5 gives $\theta_k>\psi(\theta_k)>R_k$, so $\liminf_{k\to\infty}\theta_k\ge M$.
Let $\theta_-:=\liminf_{k\to\infty}\theta_k\in[M,\psi^{-1}((M+1)/p)]$.

Taking $\liminf$ in \eqref{eq5_reset} yields
\begin{equation}\label{eq6_reset}
\liminf_{k\to\infty}
\frac{\widetilde\Delta(k+1)-\frac{\widetilde\Delta(k)}{1-p}}{\frac{\widetilde\Delta(-k)}{1-p}-\widetilde\Delta(-k-1)}
=\liminf_{k\to\infty}\frac{1}{\psi(1/\theta_k)}
=\frac{1}{\psi(1/\theta_-)}.
\end{equation}
Since $\psi(1/\theta)<1/\theta$ for $\theta>0$, we have $1/\psi(1/\theta_-)>\theta_-\ge M$, and hence
\[
\liminf_{k\to\infty}
\frac{\widetilde\Delta(k+1)-\frac{\widetilde\Delta(k)}{1-p}}{\frac{\widetilde\Delta(-k)}{1-p}-\widetilde\Delta(-k-1)}
> M.
\]

Next, note that $R_{k+1}$ can be written as a convex combination:
\begin{equation}\label{eq:contradiction}
\frac{\widetilde \Delta(k+1)}{-\widetilde \Delta(-k-1)}=
\left(
\frac{\widetilde\Delta(k+1)-\frac{\widetilde\Delta(k)}{1-p}}
{\frac{\widetilde\Delta(-k)}{1-p}-\widetilde\Delta(-k-1)}
\right)r_k
+
\left(-\frac{\widetilde\Delta(k)}{\widetilde\Delta(-k)}\right)(1-r_k),
\end{equation}
where $r_k:=\frac{\frac{\widetilde\Delta(-k)}{1-p}-\widetilde\Delta(-k-1)}{-\widetilde\Delta(-k-1)}\in(0,1)$.
Moreover, using \eqref{eq:Delta-rec} we have
\[
\frac{\widetilde\Delta(-k)}{1-p}-\widetilde\Delta(-k-1)
=
\big(\widetilde\Delta(1-k)-\widetilde\Delta(-k-1)\big)\,\phi(1/\theta_k),
\]
and since $\widetilde\Delta(1-k)>\widetilde\Delta(-k)>(1-p)\widetilde\Delta(-k-1)$,
\begin{equation}\label{lower-bound-A}
r_k
=\frac{\big(\widetilde\Delta(1-k)-\widetilde\Delta(-k-1)\big)\,\phi(1/\theta_k)}{-\widetilde\Delta(-k-1)}
\ge p\,\phi(1/\theta_k)
\ge p\,\phi\!\left(\frac{1}{\psi^{-1}((M+1)/p)}\right)
=:\widehat r>0,
\end{equation}
where the last inequality uses \eqref{upper-bound-theta1} and the monotonicity of $\phi$.

Taking $\liminf$ in \eqref{eq:contradiction} and using \eqref{eq6_reset}--\eqref{lower-bound-A} gives
\begin{align*}
\liminf_{k\to\infty} R_{k+1}
&\ge \widehat r\left(\liminf_{k\to\infty}\frac{\widetilde\Delta(k+1)-\frac{\widetilde\Delta(k)}{1-p}}{\frac{\widetilde\Delta(-k)}{1-p}-\widetilde\Delta(-k-1)}\right)
+(1-\widehat r)\left(\liminf_{k\to\infty}R_k\right)\\
&> \widehat r M+(1-\widehat r)M=M,
\end{align*}
contradicting $R_k\to M$.  Hence $R_k\to+\infty$.

\textbf{Step 7: Constructing $H_A^-$ and $H_B^-$ and verifying \Cref{def:si}.}
Step~6 implies $\lim_{k\to\infty}\frac{-\widetilde\Delta(-k)}{\widetilde\Delta(k)}=0$. 
Choose $k$ large enough so that $\frac{-\widetilde\Delta(-k)}{\widetilde\Delta(k)}\le \epsilon$,
and then choose $m\in\mathbb{N}_{++}$ so that $(1-\frac p2)^{m-1}\le \epsilon$.
Set $N^\dagger:=k+m$, and consider $N\geq N^\dagger$. 

Let $H_A^-$ (resp. $H_B^-$) collect the intermediate states (i.e., right after a battle outcome and before the reset lottery) where player $B$ (resp. $A$) has a lead of at least $k$ battle wins since the last reset. 
%%%Let $i(h)\in\{-N,\dots,N\}$ denote player $A$'s lead (the state) at an \emph{intermediate} history $h$
%%%(i.e., right after a battle outcome and before the reset lottery).  
%%%Define
%%%\[
%%%H_A^-:=\{h:\; i(h)\le -k\},
%%%\qquad
%%%H_B^-:=\{h:\; i(h)\ge k\}.
%%%\]
Since $\widetilde V(\cdot)$ is increasing, it suffices to bound $\widetilde V(-k)$.  Using the scaling in Step~4 and the monotonicity of $-\frac{\widetilde\Delta(i)}{\widetilde\Delta(-i)}$,
for any $i\ge k$,
\begin{align*}
\frac{1-\widetilde V(i)}{\widetilde V(-i)}
=\frac{\widetilde\Delta(N)-\widetilde\Delta(i)}{\widetilde\Delta(-i)-\widetilde\Delta(-N)}
\ge -\frac{\widetilde\Delta(i)}{\widetilde\Delta(-i)}
\ge -\frac{\widetilde\Delta(k)}{\widetilde\Delta(-k)}.
\end{align*}
Taking $i=k$ and using $1-\widetilde V(k)\le 1$ yields
\[
\widetilde V(-k)\le \frac{-\widetilde\Delta(-k)}{\widetilde\Delta(k)}\le \epsilon.
\]
This verifies Condition~(i) in \Cref{def:si} (recall $v=1$ under our normalization).

%%%%%For Condition~(ii), let $\widetilde P_A$ (resp.\ $\widetilde P_B$) denote the equilibrium probability that the realized
%%%%%history reaches $H_B^-$ (resp.\ $H_A^-$), i.e., that player $A$ (resp.\ $B$) ever attains a lead of at least $k$.
%%%%%Starting from the initial state $0$, before the contest can terminate at $\pm N$ it must pass through $\pm k$.
%%%%%Moreover, once some player has reached a lead of $k$, reaching $\pm N$ requires at least $m-1$ additional
%%%%%post-battle reset lotteries while the lead is strictly above $k$ (or below $-k$).  At each such lottery, with probability
%%%%%$p$ the contest resets to state $0$; conditional on a reset, symmetry implies that the next time the process hits
%%%%%$\{\pm k\}$, player $A$ is the one in the lead with probability $1/2$.  Hence the conditional probability that player $A$
%%%%%\emph{still} has not attained a $k$-lead after one such lottery is at most $1-\frac p2$.  Iterating over the $m-1$
%%%%%lotteries gives
%%%%%\[
%%%%%1-\widetilde P_A\le \left(1-\frac p2\right)^{m-1},
%%%%%\qquad
%%%%%1-\widetilde P_B\le \left(1-\frac p2\right)^{m-1}.
%%%%%\]
%%%%%Therefore the probability of reaching both $H_A^-$ and $H_B^-$ is
%%%%%\[
%%%%%\Pr\big(\text{reach both }H_A^-\text{ and }H_B^-\big)
%%%%%\ge \widetilde P_A+\widetilde P_B-1
%%%%%\ge 1-2\left(1-\frac p2\right)^{m-1}
%%%%%\ge 1-\epsilon.
%%%%%\]
%%%%%This verifies Condition~(ii) in \Cref{def:si} and completes the proof.

For Condition~(ii), 
since one of $H_A^-$ or $H_B^-$ must be reached before the game can end, 
it suffices to bound the probability that the realized history \emph{fails} to reach, say, $H_A^-=\{h:i(h)\le -k\}$, starting from a situation in which player $A$ already has a lead of $k$, which lies in $H_B^-$. 
We denote this conditional probability by $\widetilde P$. From $i=k$ to termination at $+N$ the contest must advance through at least $m-1$ further post-battle reset lotteries. 
%%during which the state remains strictly above $k$.  
At each such lottery, a reset occurs with probability $p$ and sends the game back to state $0$.  Conditional on a reset, symmetry implies that the next time the process hits the set $\{\pm k\}$, player $B$ is the one leading (i.e., the history reaches $H_A^-$) with probability $1/2$.  Hence, at each of these $m-1$ lotteries, the conditional probability that we \emph{do not} end up reaching $H_A^-$ via a reset is at most $1-\frac{p}{2}$.
This means that the probability that $H_A^-$ is not reached after \emph{any} of these lotteries is at most $\left(1-\frac{p}{2}\right)^{m-1}$. Therefore,
\[
1-\widetilde P\le\text{$H_A^-$ is not reached after \emph{any} of these $m-1$ lotteries} \;\le\; \Bigl(1-\frac{p}{2}\Bigr)^{m-1}\leq\epsilon.\]
This verifies Condition~(ii) in \Cref{def:si} and completes the proof.
\end{proof}

\subsection{Proof of \Cref{lemma:ratio-instances} Based on \Cref{prop:gensf,lemma:gainratio}}
To accommodate the case where the success function may not be homogeneous, we reformulate this lemma as follows. 
For an unbalanced
contest $\mathcal{M}$ with final prize $v^{\prime }>0$, let $V_+^{\mathcal{M}%
}(v^{\prime })$ denote the advantaged player's expected equilibrium payoff
and $V_-^{\mathcal{M}}(v^{\prime })$ denote the disadvantaged player's
expected equilibrium payoff. 

The following definition formalizes the condition of sufficiently unbalanced subcontests, and the lemma that follows provides two such instances. 
%%The following definition formalizes the condition that subcontests are sufficiently unbalanced, and the 
\begin{definition}\label{cond:subcontest}
\normalfont
A sequence of subcontests $\{\mathcal M_K\}_{K=1}^{+\infty}$ is said to \emph{become infinitely unbalanced} if 
 $\frac{v'-V_+^{\mathcal M_K}(v')}{V_-^{\mathcal M_K}(v')}$ uniformly converges to $+\infty$ as $K\to+\infty$ for all $v'>0$.
\end{definition}

\setcounter{lemmaprime}{\numexpr\getrefnumber{lemma:ratio-instances}-1\relax}
\begin{lemmaprime}
\label{lemma:gen-ratio-instances} When the success function admits properties in \Cref{prop:gensf} and \Cref{lemma:gainratio}, both the sequence of $\mathcal{M}(K,1)$ contests and the sequence of tug-of-war with margin $K+1$ and initiated from state $K$ become infinitely unbalanced.
\end{lemmaprime}

\begin{proof}
\textbf{$\mathcal M(K,1)$ contest.}
Consider a $\mathcal M(K,1)$ contest with prize $v'>0$. 
Let $V_+^K:=V_+^{\mathcal M(K,1)}(v')$ and $V_-^K:=V_-^{\mathcal M(K,1)}(v')$.
The following holds for the first battle. 
%%%The equilibrium payoffs can be solved by backward induction. In detail, analyze the first battle. If the advantaged player wins, he will get the prize; otherwise, the subcontests $\mathcal M(K-1,1)$ will be entered, we have that
\begin{align*}
\begin{cases}
        V_+^{K} = V_+^{K-1}+ \pi^*\left(v'-V_+^{K-1}, V_-^{K-1}\right)\left(v'-V_+^{K-1}\right),\\
        V_-^{K} =\pi^*\left(V_-^{K-1},v'-V_+^{K-1}\right)V_-^{K-1}.\\
\end{cases}
\end{align*}
Subtracting the first equation from $v'$ and dividing the result by the second equation yields
\begin{align*}
    \frac{v'-V_+^K}{V_-^K} =\underbrace{ \frac{1-\pi^*\left(v'-V_+^{K-1}, V_-^{K-1}\right)}{\pi^*\left(V_-^{K-1},v'-V_+^{K-1}\right)}}_{>1}\frac{v'-V_+^{K-1}}{V_-^{K-1}}.
\end{align*}
For any large $R''>1$, by \Cref{lemma:gainratio}\ref{pi-high}, if $\frac{v'-V_+^{K-1}}{V_-^{K-1}}<R''$, then $\frac{1-\pi^*\left(v'-V_+^{K-1}, V_-^{K-1}\right)}{\pi^*\left(V_-^{K-1},v'-V_+^{K-1}\right)}>1+d_{R''}$, and the above equation implies that $\frac{v'-V_+^K}{V_-^K}$ grows faster than geometrically with a ratio $1+d_{R''}$. 
Since $\frac{v'-V_+^{1}}{V_-^{1}}\geq1$, $\frac{v'-V_+^K}{V_-^K}>R''$ for all $K>\log_{1+d_{R''}} R''+1$, implying uniform convergence across all $v'>0$. 

\noindent\textbf{Tug-of-war contest with margin $K+1$ and initiated from state $K$.} By \eqref{eqn:tow-eqm-4} we have
\begin{equation*}
\begin{split}
    \frac{v' - V_+^{\mathcal M_K}(v')}{V_-^{\mathcal M_K}(v')} &=  \frac{V(K+1;K+1,v') - V(K;K+1,v')}{V(-K;K+1,v') - V(-K-1;K+1,v')} \\
    &= \frac{1-\pi(K;K+1,v')}{\pi(-K;K+1,v')}\frac{\Delta(K;K+1,v')}{\Delta(-K;K+1,v')} \geq \frac{\Delta(K;K+1,v')}{\Delta(-K;K+1,v')}. 
    \end{split}
\end{equation*}
It is easy to see from the proof of \Cref{lemma:int-N} that $\frac{\Delta(K;K+1,v')}{\Delta(-K;K+1,v')}$ uniformly converges to $+\infty$ as $K\to+\infty$ across $v'>0$. The uniform convergence of $ \frac{v' - V_+^{\mathcal M_K}(v')}{V_-^{\mathcal M_K}(v')}$ follows immediately. 
\end{proof}

\subsection{Proof of \Cref{prop:full-suff} Based on \Cref{prop:gensf,lemma:gainratio}}
\begin{proof}

We first establish the following lemma, which will be useful for establishing the ``transient'' part of the transient dominance property. 
\begin{lemma}\label{lemma:v-bound2}
If the success function admits properties in \Cref{prop:gensf} and \Cref{lemma:gainratio}, 
there exists $\lambda(\cdot)>0$ defined on $\N_{++}$ such that for any contest with any prize $v'>0$,
if player $\ell\in\{A,B\}$ wins the contest after winning the first $L\in\N_{++}$ battles consecutively, then player $\ell$'s equilibrium payoff is greater than $\lambda(L)v'$.
\end{lemma}
\begin{proof}
Without loss of generality, we focus on player $A$ and suppose that player $A$ wins the contest after winning the first $L$ battles. 
Consider the history $\underbrace{(A,\ldots,A)}_{L\text{ times}}$, and let $V_A^n$ denote player $A$'s continuation value at history $\underbrace{(A,\ldots,A)}_{n\text{ times}}$ for $n\in\{0,1,\ldots,L\}$. 
Clearly, $V_A^0$ is player $A$'s equilibrium payoff and $V_A^L=v'$. 

We prove by induction that $V_A^{n}\geq\alpha_nv'$ for some $\alpha_n>0$ that is independent of the contest form and prize value. 
This trivially holds for $n=L$. Suppose that $V_A^n\geq \alpha_n v'$ for some $1\leq n\leq L$ with $\alpha_L=1$. 
Consider the $n$th battle---i.e., the battle that follows the history $\underbrace{(A,\ldots,A)}_{n-1\text{ times}}$. 
Fix some $d\in(0,1)$, and let $\Delta_\ell^n$ denote player $\ell$'s incentive to win in this battle. 
If $\Delta_A^n>V_A^nd\geq\alpha_ndv'$, then $\Delta_A/\Delta_B\geq \alpha_nd$, and by \Cref{lemma:gainratio}\ref{pi-ratio}, $V_A^{n-1}\geq\pi_{\alpha_nd}^*V_A^{n}$. 
On the other hand, if $\Delta_A^n\leq V_A^nd$, 
we know that for player $A$, losing the battle leads to a continuation value that is at least $(1-d)V_A^n$, so $V_A^{n-1}\geq(1-d)V_A^{n}\geq(1-d)\alpha_nv'$. In any case, 
\[
V_A^{n-1}\geq\alpha_{n-1}:=\min\{\pi_{\alpha_nd}^*,(1-d)\alpha_n\}v',
\]
which completes the proof. 
\end{proof}

We now proceed to prove the theorem. For expositional ease, we first prove the $q=1$ case (so the subcontest in each round always happens), and then show how the proof can be extended to the general $q\in(0,1]$ case in a straightforward manner. 

%%For notational convenience, we introduce the following shorthands: 

By \Cref{lemma:gen-ratio-instances}, there exists $K^\star$ such that $\frac{v'-V_+(\mathcal M_{K^\star};v')}{V_-(\mathcal M_{K^\star};v')}\geq \frac{1}{\epsilon}$ for all $v'>0$.
By \Cref{lemma:v-bound2}, the winning probability of the disadvantaged player in $\mathcal M_{K^\star}$ with any final prize $v'>0$ is at least $\lambda(L_{K^\star})>0$.  
Let $N^\star$ satisfy $[1-\lambda(L_{K^\star})]^{N^\star}\leq \epsilon$. 

We show next that the $N^\star$-round iterated incumbency contest with component subcontest $\mathcal M_{K^\star}$ satisfies the transient dominance property. Then by \Cref{thm:sufficient}, expected total effort in the contest is greater than $(1-4\epsilon)v$.

Let $H_{\ell}^-$ be the histories where the opponent $-\ell\in\{A,B\}$ just became the incumbent for a new round. Let $W_+^{\mathcal{M}%
}(n;N,v)$ denote the incumbent's continuation value at the beginning of
round $n\geq1$ and $W^{\mathcal{M}}_-(n;N,v)$ denote that of the laggard.

Then player $\ell$'s continuation value at a history in $H_\ell^-$ is $W_-^{\mathcal M_{K^\star}}(n;N^\star,v)$ for some $n\in\{1,\ldots,N\}$ and the opponent's continuation value is $W_+^{\mathcal M_{K^\star}}(n;N^\star,v)$. 
To establish the transient dominance property, we need to show that (i)
$\frac{v}{W_-^{\mathcal M_{K^\star}}(n;N^\star,v)}\geq\frac{1}{\epsilon}$ for all $n\in\{1,\ldots,N\}$ and (ii) the probability that the initial incumbent loses its incumbency is at least $1-\epsilon$. 
Point (ii) follows immediately from the construction of $N^\star$ because the probability that the initial incumbent never loses his incumbency is less than $[1-\lambda(L_{K^\star})]^{N^\star}\leq \epsilon$.

%%For notational convenience, we introduce the following shorthands:
%%\begin{align*}
%%\end{align^*}
We proceed to prove (i). 
Let $v_n:=W_+^{\mathcal M_{K^\star}}(n+1;N^\star,v)-W_-^{\mathcal M_{K^\star}}(n+1;N^\star,v)$ for $n=1,\ldots,N$, with the understanding that $W_+^{\mathcal M_{K^\star}}(N+1;N^\star,v)=v$ and $W_-^{\mathcal M_{K^\star}}(N+1;N^\star,v)=0$. 
It follows from the structure of the iterated incumbency contest that
\begin{align*}
W_+^{\mathcal M_{K^\star}}(n;N^\star,v)&=W_-^{\mathcal M_{K^\star}}(n+1;N^\star,v)+V_+^{\mathcal M_{K^\star}}(v_n),\\
W_-^{\mathcal M_{K^\star}}(n;N^\star,v)&=W_-^{\mathcal M_{K^\star}}(n+1;N^\star,v)+V_-^{\mathcal M_{K^\star}}(v_n),
\end{align*}
where the first equation can be rewritten as 
\[
v-W_+^{\mathcal M_{K^\star}}(n;N^\star,v)=v-W_+^{\mathcal M_{K^\star}}(n+1;N^\star,v)+[v_n-V_+^{\mathcal M_{K^\star}}(v_n)].
\]
Consequently, 
\begin{align*}
v-W_+^{\mathcal M_{K^\star}}(n;N^\star,v)&=\sum_{j=n}^N[v_n-V_+^{\mathcal M_{K^\star}}(v_n)],\\
W_-^{\mathcal M_{K^\star}}(n;N^\star,v)&=\sum_{j=n}^NV_-^{\mathcal M_{K^\star}}(v_n).
\end{align*}
By construction, $\frac{v_n-V_+^{\mathcal M_{K^\star}}(v_n)}{V_-^{\mathcal M_{K^\star}}(v_n)}\geq\frac{1}{\epsilon}$. It follows that
\[
\sum_{j=n}^N[v_n-V_+^{\mathcal M_{K^\star}}(v_n)]\geq\frac{1}{\epsilon}\sum_{j=n}^NV_-^{\mathcal M_{K^\star}}(v_n),
\]
which implies that $\frac{v-W_+^{\mathcal M_{K^\star}}(n;N^\star,v)}{W_-^{\mathcal M_{K^\star}}(n;N^\star,v)}\geq\frac{1}{\epsilon}$, and hence $\frac{v}{W_-^{\mathcal M_{K^\star}}(n;N^\star,v)}\geq\frac{1}{\epsilon}$.
This finishes the proof for the $q=1$ case.

For the $q\in(0,1]$ case, notice that the round-$n$ competition is the combination of the exogenous shock and some subcontest $\mathcal M$. 
The incumbent and the laggard's expected payoffs from such a competition with prize $v'$ are, respectively,
\begin{align*}
\widetilde V_+^{\mathcal M}(v')= (1-q)v'+qV_+^{\mathcal M}(v')\text{ and }\widetilde V_-^{\mathcal M}(v')= qV_-^{\mathcal M}(v').
\end{align*}
Therefore, 
\[
\frac{v'-\widetilde V_+^{\mathcal M}(v')}{\widetilde V_-^{\mathcal M}(v')}=\frac{v'- V_+^{\mathcal M}(v')}{ V_-^{\mathcal M}(v')}.
\]
Then it is clear that the same proof for the $q=1$ case applies with $K^\star$ taken to satisfy $\frac{v'-V_+(\mathcal M_{K^\star};v')}{V_-(\mathcal M_{K^\star};v')}\geq \frac{1}{\epsilon}$ for all $v'>0$ and $N^\star$ taken to satisfy $\{1-q+q[1-\lambda(L_{K^\star})]\}^{N^\star}\leq \epsilon$.
\end{proof}

\end{appendices}

\end{document}